\PassOptionsToPackage{dvipsnames}{xcolor}
\documentclass[11pt,a4paper]{article}

\ifdefined\pdfobjcompresslevel
\fi

\usepackage[margin=1in]{geometry}
\usepackage[T1]{fontenc}
\usepackage{amsmath,amssymb,amsthm,mathtools}
\usepackage{booktabs,array}
\usepackage{graphicx}
\usepackage{tikz}
\usetikzlibrary{shadings}
\usepackage{placeins}
\usepackage{enumitem}
\usepackage{parskip}
\usepackage[dvipsnames]{xcolor}
\usepackage[mathlines]{lineno}

\usepackage{hyperref}
\hypersetup{colorlinks=true,linkcolor=black,citecolor=black,urlcolor=black,
  pdftitle={Chamber geometry and specification numbers of Boolean threshold functions},
  pdfauthor={Martin Anthony},
  pdfsubject={Boolean threshold functions, specification numbers, hyperplane arrangements},
  pdfkeywords={threshold function, specification number, essential point, hyperplane arrangement, chamber, resonance arrangement, zonotope, teaching dimension}}

\usepackage{caption}
\newtheorem{theorem}{Theorem}[section]
\newtheorem{proposition}[theorem]{Proposition}
\newtheorem{corollary}[theorem]{Corollary}
\newtheorem{lemma}[theorem]{Lemma}

\theoremstyle{definition}
\newtheorem{definition}[theorem]{Definition}
\newtheorem{problem}[theorem]{Problem}
\newtheorem{example}[theorem]{Example}
\theoremstyle{remark}
\newtheorem{remark}[theorem]{Remark}

\newcommand{\B}{\{0,1\}}
\newcommand{\R}{\mathbb{R}}
\newcommand{\Hn}{\mathcal{H}_n}
\newcommand{\Tn}{\mathcal{T}_n}

\newcommand{\Tk}[1]{\mathcal{T}_{#1}}
\newcommand{\C}{\mathcal{C}}
\newcommand{\Arr}{\mathcal{A}}
\newcommand{\Ess}{\operatorname{Ess}}
\newcommand{\TD}{\operatorname{TD}}
\newcommand{\sgn}{\operatorname{sgn}}
\newcommand{\Ch}{\operatorname{Ch}}
\newcommand{\Res}{\mathcal{R}_n}
\DeclareMathOperator{\Rep}{Rep}
\newcommand{\eps}{\varepsilon}

\tikzset{
  flab/.style={font=\small, fill=white, fill opacity=0.9, text opacity=1,
               inner sep=1pt, rounded corners=1pt},
  rayext/.style={thick, dashed, ->},
  cut/.style={gray!70, dashed, semithick},
}
\newcommand{\chambertriad}{%
  \draw[->,gray!70] (0,0)--(0.55,0)      node[right=-1pt,font=\scriptsize,black]{$a_1$};
  \draw[->,gray!70] (0,0)--(0,0.55)      node[above=-2pt,font=\scriptsize,black]{$a_2$};
  \draw[->,gray!70] (0,0)--(-0.40,-0.34) node[below=-2pt,font=\scriptsize,black]{$a_0$};
}

\definecolor{zonoblue}{HTML}{2C6FBB}
\definecolor{zonoorange}{HTML}{E08A1E}
\definecolor{zonoink}{HTML}{1F2933}
\definecolor{zonored}{HTML}{B03030}
\tikzset{
  lblblue/.style  ={draw=zonoblue,   rounded corners=2pt, fill=white, inner sep=3pt, font=\small, anchor=south},
  lblorange/.style ={draw=zonoorange, rounded corners=2pt, fill=white, inner sep=3pt, font=\small, anchor=west},
}

\title{Chamber geometry and specification numbers of Boolean threshold functions}

\author{Martin Anthony\\
{Department of Mathematics, London School of Economics}\\
{\texttt{m.anthony@lse.ac.uk}}}

\date{}

\begin{document}
\maketitle

\begin{abstract}
The specification number $\sigma_n(f)$ of a Boolean threshold function $f$ on
$n$ variables is the least number of labelled examples that determine $f$
uniquely among all threshold functions, and the smallest such set consists of
its essential examples. We develop the geometric interpretation of this
invariant observed by Zuev: the threshold functions are the chambers of a
central hyperplane arrangement in the $(n+1)$-dimensional space of weights and
thresholds, and the essential examples of a function correspond to the facets of
its chamber. The specification number is therefore the chamber's facet number,
and this correspondence organises the paper.

The viewpoint gives a unified geometric account of known and new results. The
lower bound $\sigma_n(f)\ge n+1$ amounts to the fact that a pointed
full-dimensional cone in $(n+1)$-dimensional space has at least $n+1$ facets,
with equality exactly for simplicial chambers. The average specification number
$\overline\sigma_n$ becomes an average facet count. Gutekunst, M\'esz\'aros, and
Petersen ask how the average number of facets per threshold chamber grows. We
show this average is the average specification number $\overline\sigma_n$,
evaluate it exactly via the resonance arrangement, and bound it using a theorem
of Fukuda, Tamura, and Tokuyama on average facet counts, obtaining
$\overline\sigma_n\le 2n$ and hence $\overline\sigma_n=\Theta(n)$. This settles
their question.

The same geometry connects threshold functions with a threshold zonotope whose
vertices are modified Chow vectors: its one-skeleton is the one-inclusion graph
of threshold functions, and the degree of a vertex is the specification number
of the corresponding function. The chamber method also extends to polynomial
threshold functions, giving analogous average-specification bounds.

Finally, we apply the correspondence to the operations of Lozin \emph{et al.}\
on threshold functions of minimum specification number. Adding a variable and
extending on a variable both amount to taking the product of a chamber closure
with a half-line, and so preserve simpliciality. For the symmetric-variables
extension we give an exact thresholdness criterion and show that, whenever the
extension is threshold, minimum specification number is preserved. We resolve
their fourth operation on $f_{n,k}$: it produces a function of minimum
specification number exactly in the boundary case $k=n-1$.
\end{abstract}


{\setlength{\parskip}{0pt}\tableofcontents}

\section{Introduction}\label{sec:intro}

\paragraph{Specification number and chamber geometry.}
A threshold function on $\B^n$, also called a halfspace or a linearly separable
Boolean function, is a Boolean function $f\colon\B^n\to\B$ for which $f(x)=1$ if
and only if $\langle w,x\rangle\ge\theta$ for some weight vector $w\in\R^n$ and
threshold $\theta\in\R$. (Here, and throughout, $\langle w,x\rangle$ denotes the standard Euclidean inner product, $\sum_{i=1}^n w_ix_i$.)
 We denote by  $\Hn$ the class of all threshold functions on
$\B^n$. A subset $S\subseteq\B^n$ \emph{specifies} $f\in\Hn$ if no threshold
function other than $f$ agrees with $f$ on $S$, and the \emph{specification
number} $\sigma_n(f)$ is the minimum cardinality of such a set. The
\emph{teaching dimension} of $\Hn$ is $\TD(\Hn)=\max_f\sigma_n(f)$, and the
\emph{average specification number} is
$\overline\sigma_n=\frac{1}{|\Hn|}\sum_{f\in\Hn}\sigma_n(f)$. The
specification-number terminology follows Anthony, Brightwell, and Shawe-Taylor
\cite{ABST1995}. Closely related terminology appeared earlier in computational
learning theory: Goldman and Kearns used teaching sequences and teaching
dimensions \cite{GoldmanKearns1995}, and Shinohara and Miyano used the term key
\cite{ShinoharaMiyano1991}. In that language a specifying set is a teaching set
for the target $f$ relative to $\Hn$, and $\sigma_n(f)$ is the concept-level
teaching dimension, or key size, of $f$. The definition involves no
distributional assumption and no learning algorithm.

A point $x\in\B^n$ is essential for $f\in\Hn$ if flipping the
value of $f$ at $x$ produces another threshold function; we write $\Ess(f)$ for the
set of essential points. The paper organises its questions around a single
geometric correspondence. Each $f\in\Hn$ corresponds to a chamber $C_f$ of the
\emph{Boolean threshold arrangement} $\Arr_n$, the finite central arrangement of
the $2^n$ hyperplanes in the $(n+1)$-dimensional parameter space of weights and
thresholds (defined formally in Section~\ref{sec:setup}). 
 The hyperplane $H_x$
associated with $x\in\B^n$ consists of the parameter vectors for which 
$\langle w,x\rangle=\theta$, and the chambers of $\Arr_n$ are the connected components
of the complement of the union of the $H_x$. Theorem~\ref{thm:facet} gives a
self-contained proof that
\[
        x\in\Ess(f)
        \quad\Longleftrightarrow\quad
        H_x\ \text{supports a facet of the closed chamber}\ \overline{C_f},
\]
where a \emph{facet} is a codimension-one boundary face of the polyhedral cone
$\overline{C_f}$. Consequently
\[
        \sigma_n(f)=\#\{\text{facets of }\overline{C_f}\}.
\]
Specification questions for Boolean threshold functions therefore become
facet-counting questions for chambers of $\Arr_n$.

\paragraph{Main contributions.}
The facet correspondence of Theorem~\ref{thm:facet} is the organising result of
the paper. Its statement, that the essential points of $f$ index the facets of
$\overline{C_f}$, is observed in Zuev's graph-theoretic discussion of the
facets of threshold cones \cite[\S5]{Zuev1991}; the contribution here is a
self-contained arrangement-theoretic formulation and proof. Several of its
consequences are not new. The specification number $\sigma_n(f)$ equals the
degree of $f$ in the one-inclusion graph $G(\Hn)$, essential specification
coincides with shortest-path closure, and $\overline\sigma_\C<2\,\mathrm{VCD}(\C)$
for every nontrivial shortest-path-closed class $\C$, all from Zuev
\cite{Zuev1991} and from Doliwa, Fan, Simon, and Zilles \cite{Doliwa2014}. The
lower bound $\sigma_n(f)\ge n+1$ (Corollary~\ref{cor:lower}) is in Hu
\cite{Hu1965}, and Zuev states it in facet form; the criterion that
$\sigma_n(f)=n+1$ exactly when $\overline{C_f}$ is simplicial
(Corollary~\ref{cor:simplicial}) is the definition of a simplicial cone applied
to the facet count, and for $n\ge2$ it identifies $\Tn$, the class of
all-variable-relevant threshold functions attaining the lower bound, with the
threshold functions whose chamber closures are simplicial
(Theorem~\ref{thm:Tn-simplicial}).

We then apply the correspondence. Viewing the standard chamber--facet
incidence identity through it, together with the resonance-arrangement
identification of \cite{GMP2021}, yields the exact formula
$\overline\sigma_n=2^{n+1}r(\Res)/|\Hn|$ (Corollary~\ref{cor:exact}), a synthesis
of that identity with Zuev's facet interpretation. A theorem of Fukuda, Tamura, and Tokuyama \cite{Fukuda1993} yields the sharper bound $\overline\sigma_n\le 2n$ through the spherical view of the average developed in Section~\ref{sec:average} (Theorem~\ref{thm:theta-n}); the order $\overline\sigma_n=\Theta(n)$ also follows from the bound of Doliwa, Fan, Simon, and Zilles \cite{Doliwa2014}. Because the formula is an equality, it sharpens the
region-count estimate of \cite{GMP2021} for $n\ge3$ and shows the average number
of facets per threshold chamber is linear in $n$, resolving
\cite[Problem~3]{GMP2021}. The
cardinality bound $\overline\sigma_\C\le\log_2|\C|$ for essentially specified (equivalently, shortest-path-closed) classes (Theorem~\ref{thm:average}) recovers
the $n^2$ bound of \cite{ABST1995} from a single edge-isoperimetric inequality.
The same chamber mechanism extends to degree-$d$ polynomial threshold functions.
There the lifted points are not in general position, so we prove essential
specification by a facet-sign lemma rather than by a general-position count
(Proposition~\ref{prop:ptf-average}). We also extend Stenson's identification
of the threshold functions with the vertices of the Chow-parameter zonotope $Z_n$
\cite{Stenson2014} to its one-skeleton, the one-inclusion graph $G(\Hn)$: the
edges at $f$ are indexed by the essential points of $f$, and the vertex degree is
$\sigma_n(f)$ (Theorem~\ref{thm:zonotope}).

We also use the chamber view to account for the minimum-specification-preserving
operations of Lozin, Razgon, Zamaraev, Zamaraeva, and Zolotykh \cite{Lozin2022}.
Adding a variable and extension on a variable are the same construction, the
product of the closed chamber with a closed half-line, and we compute the
resulting essential sets (Subsection~\ref{ssec:lozin-operations}); the
symmetric-variables extension is threshold exactly when a representation of the
input is sufficiently unbalanced, and minimum specification is preserved when it
is. We resolve the fourth operation left open in \cite{Lozin2022}, showing that
the image of $f_{n,k}$ has minimum specification number if and only if $k=n-1$
(Section~\ref{sec:fourth}). We do not attempt to classify the simplicial
chambers of $\Arr_n$, to determine the limiting constant in
$\overline\sigma_n=\Theta(n)$, or to give an output-polynomial algorithm for
computing $\Ess(f)$ from a compact weight representation;
Section~\ref{sec:open} collects these directions.

\paragraph{Relation to earlier work.}
The lower bound $\sigma_n(f)\ge n+1$ appears in Hu's threshold-logic monograph
\cite{Hu1965}. Anthony, Brightwell, and Shawe-Taylor developed the
specification-number theory of $\Hn$ in \cite{ABST1995}: every $f\in\Hn$
satisfies $\sigma_n(f)\ge n+1$; every nested function, now usually called linear
read-once (lro), attains the bound; the unique minimum specifying set is
$\Ess(f)$; and $\overline\sigma_n\le n^2$ for $n\ge2$. Section~\ref{sec:chamber}
gives the recursive definition of linear read-once, where, in the terminology of
\cite{ABST1995}, the lro functions depending on all ambient variables are the
\emph{nested} functions.

Zuev introduced the graph of threshold functions, whose vertices are the members
of $\Hn$ and whose edges join two functions that take different values at exactly
one point of the Boolean cube $\B^n$; this is the one-inclusion graph $G(\Hn)$ used here.
 In the same discussion he observed the
connection between adjacent cones in coefficient space, the facets of the cone,
and the cube vertices at which the function can be changed, noting that the facet
points form a smallest determining set, that their number is the graph degree,
and that graph distance agrees with Hamming distance
\cite[\S5]{Zuev1991}. Doliwa, Fan, Simon, and Zilles later
studied shortest-path-closed concept classes. They proved that a class is
shortest-path-closed if and only if, for every concept, the coordinates incident
with it in its one-inclusion graph form its unique minimum teaching set, and that
the average specification number of every nontrivial shortest-path-closed class is
less than twice its Vapnik--Chervonenkis dimension
\cite[Lemma~30 and Theorem~31]{Doliwa2014}. Section~\ref{sec:average} records
these results in the notation of this paper.

The 1995 paper conjectured that nested functions are the only $f\in\Hn$ with
$\sigma_n(f)=n+1$. Lozin, Razgon, Zamaraev, Zamaraeva, and Zolotykh disproved this
in \cite{Lozin2017,Lozin2018LRO}, and \cite{Lozin2022} broadened the
counterexamples to the family
\[
        f_{n,k}=x_1x_2\vee x_1x_3\vee\cdots\vee x_1x_k\vee x_2x_3\cdots x_n,
        \qquad 3\le k\le n-1,
\]
of non-nested threshold functions of $n$ relevant variables with
$\sigma_n(f_{n,k})=n+1$. That paper studied the class $\Tn$ of
all-variable-relevant threshold functions attaining the lower bound, giving
inductive operations that produce members of $\Tn$, a complete description for
$n\le5$, and a finite description for general $n$ left open. In the present
chamber language, for $n\ge2$, $\Tn$ is exactly the class of threshold functions
whose chamber closures are simplicial, and two of the operations of
\cite{Lozin2022} coincide: adding a variable and extension on a variable are the
same chamber operation, the product $\overline C\mapsto\overline C\times\R_{\ge0}$,
while the symmetric-variables extension is not of this form
(Remark~\ref{rem:two-extensions-coincide}).

The extremal-point picture points in a different direction. The 2017 conference paper proved the
threshold version \cite[Theorem 13]{Lozin2017}, and the 2018 journal paper
strengthened it to all positive Boolean functions \cite[Theorem 6]{Lozin2018LRO}:
a positive Boolean function with $k$ relevant variables has exactly $k+1$ extremal
points if and only if it is linear read-once. The 1995 conjecture therefore
identified the right structural class, lro, but the wrong characterising quantity,
specification number in place of extremal-point count. The 2018 journal paper also
supplies lro characterisations through read-once Chow functions and forbidden
subfunctions \cite[Theorems 7--8]{Lozin2018LRO}, which Section~\ref{ssec:lro-chow}
records.

The parameter-space view of threshold functions, and the geometry of their Chow
parameters, are older. Chow's injectivity theorem determines a threshold function
from its Chow parameters \cite{Chow1961}, and the corner theorem stated by Winder
identifies the threshold functions with the extreme points, the corners, of the
Chow-parameter set \cite[Theorem~7]{Winder1971}. Stenson introduced the threshold
zonotope and identified its vertices with the threshold functions
\cite{Stenson2014}.

\paragraph{Organisation of the paper.}
In Section~\ref{sec:setup}, we describe the setting and notation, distinguishing the
(homogenised) input space from the parameter space and introducing the Boolean
threshold arrangement. In Section~\ref{sec:chamber}, we prove the facet correspondence
(Theorem~\ref{thm:facet}) and discuss its immediate consequences, including the lower
bound and the simpliciality criterion.  Section~\ref{sec:average}
treats the average specification number through the one-inclusion graph and
shortest-path closure. We establish the cardinality bound
$\overline\sigma_\C\le\log_2|\C|$ using an edge-isoperimetric inequality for the
Hamming cube (Theorem~\ref{thm:average}), record the bound
$\overline\sigma_n<2(n+1)$ of Doliwa, Fan, Simon, and Zilles
(Theorem~\ref{thm:vc-average}), and obtain the sharper bound $\overline\sigma_n\le 2n$ from a result of Fukuda, Tamura, and Tokuyama. We give a similar argument for spanning feature arrangements
(Proposition~\ref{prop:feature-arrangements}), which, in particular, yields an
average-specification bound for polynomial threshold functions
(Proposition~\ref{prop:ptf-average}). Section~\ref{sec:exact} derives the exact resonance-arrangement
formula (Corollary~\ref{cor:exact}) and its degree-$d$ analogue.
Section~\ref{sec:growth} gives further consideration to the order of growth of the average specification number, establishing $\overline\sigma_n=\Theta(n)$ (Theorem~\ref{thm:theta-n}). In Section~\ref{sec:zonotope}, we discuss the Chow-parameter zonotope
(Theorem~\ref{thm:zonotope}). Section~\ref{sec:Tn} investigates three operations of Lozin \emph{et al.} \cite{Lozin2022} on the class $\Tn$ (the class of threshold functions with minimum specification number), framing these as chamber operations and determining when the third of these (the symmetric variables extension) preserves specification number minimality. 
 Section~\ref{sec:fourth} considers a fourth operation from \cite{Lozin2022} and answers an open problem concerning it.  
We collect open problems and further directions in Section~\ref{sec:open}.

\section{Setting and basic notation}\label{sec:setup}

We distinguish two vector spaces. The \emph{homogenised input space} is
$X=\R^{n+1}$, into which the Boolean cube $\B^n$ is embedded by
\[
        x \longmapsto \widetilde x=(1,x_1,\ldots,x_n) \in X .
\]
For each \emph{cube point} $x \in \B^n$, we refer to $\widetilde x$ as the corresponding \emph{lifted cube point}. 
The \emph{parameter space} is the dual space $A=X^*$, which we identify with
$\R^{n+1}$ using the standard inner product. A parameter vector
$a=(a_0,a_1,\ldots,a_n) \in A$ represents the affine form
\[
        x \longmapsto \langle a,\widetilde x\rangle
        =a_0+a_1x_1+\cdots+a_nx_n .
\]
Thus $a$ and $\widetilde x$ play dual roles. Fixing $a$ gives a threshold classifier on the
input cube: explicitly,  $a=(a_0, a_1, \dots, a_n)$ determines the threshold function
$f$ with $f(x)=1$ if and only if $\langle a,\widetilde x\rangle\ge 0$.
Fixing a cube point $x$ gives the hyperplane
\[
        H_x = \{a \in A : \langle a,\widetilde x\rangle=0\}
\]
in parameter space, the set of parameter vectors orthogonal to the lifted point
$\widetilde x$, equivalently those whose affine form vanishes at $x$.

We use the following standard arrangement terminology: a finite real
hyperplane arrangement is a finite set of hyperplanes, each of the form
$\{v:\langle u,v\rangle=c\}$ with a non-zero normal vector $u$. (The pair $(u,c)$
is determined by the hyperplane up to a common non-zero scalar.) The arrangement is said to be \emph{central} if every hyperplane passes
through the origin;  that is, each has $c=0$.

The \emph{Boolean threshold arrangement} is $\Arr_n=\{H_x:x\in\B^n\}$, a
finite central arrangement in $A$: each
$H_x=\{a:\langle a,\widetilde x\rangle=0\}$ passes through the origin and has
normal vector $\widetilde x$, and Lemma~\ref{lem:Hx-injective} below shows the
$H_x$ are distinct, so $|\Arr_n|=2^n$.
This parameter-space viewpoint is classical in threshold logic and pattern
recognition; see, for example, Hu \cite{Hu1965}, Winder \cite{Winder1963},
Cover \cite{Cover1965}, and Muroga \cite{Muroga1971}.  We use the language of
hyperplane arrangements, following the general face-counting framework of
Zaslavsky \cite{Zaslavsky1975}.

A point $x \in \B^n$ is \emph{essential} for $f \in \Hn$ if the function
$f \oplus \mathbf 1_{\{x\}}$, obtained by flipping the value of $f$ at $x$, is
also threshold. Here $\mathbf 1_{\{x\}}$ denotes the indicator of $x$, and
$\oplus$ is pointwise addition modulo $2$, so adding $\mathbf 1_{\{x\}}$ flips the
value at $x$ and leaves all other values unchanged. We write $\Ess(f)$ for the
set of essential points of $f$.

A subset $S\subseteq\B^n$ is a \emph{specifying set} for $f\in\Hn$ if no other
threshold function agrees with $f$ on $S$. The \emph{specification number}
$\sigma_n(f)$ is the minimum cardinality of a specifying set. A specifying set of
cardinality $\sigma_n(f)$ is called a \emph{minimum specification} of $f$.

\begin{lemma}\label{lem:Hx-injective}
For $x,y\in\B^n$, we have $H_x=H_y$ if and only if $x=y$.
\end{lemma}

\begin{proof}
If $x=y$ then $H_x=H_y$. The normals $\widetilde x=(1,x)$ all have first
coordinate $1$, so for $x\ne y$ they are non-parallel, and so $H_x\ne H_y$.
\end{proof}

A \emph{chamber} of $\Arr_n$ is a connected component of
$A \setminus \bigcup_x H_x$. Equivalently, it is a maximal region on which, for
every cube point $x$, the evaluation $\langle a,\widetilde x\rangle$ has
constant non-zero sign as $a$ ranges over the region. Writing
$\eps_C(x)\in\{-1, 1\}$ for this sign defines the \emph{sign vector} $\eps_C$ of
the chamber $C$. More generally, the chambers of any arrangement are the
connected components of its complement. We write $\Ch(\Arr)$ for the set of
chambers and $r(\Arr)=|\Ch(\Arr)|$ for their number. The closure $\overline C$
is obtained by relaxing the strict inequalities to weak ones:
\[
        \overline C
        =\{a\in A: \eps_C(x)\langle a,\widetilde x\rangle \ge 0
          \text{ for all } x\in\B^n\}.
\]
Because $\Arr_n$ is central, these defining inequalities are homogeneous, so
$\overline C$ is a \emph{polyhedral cone}: a finite intersection of closed
halfspaces whose bounding hyperplanes pass through the origin; equivalently, it
is the set of non-negative combinations of finitely many vectors.  It is
full-dimensional, since $C$ is a non-empty open subset of $A$. It is also
\emph{pointed}, meaning $\overline C\cap(-\overline C)=\{0\}$, so that it
contains no line.  Indeed, if $a$ and $-a$ both lie in $\overline C$, then for every $x$ both
$\eps_C(x)\langle a,\widetilde x\rangle \ge 0$ and
$\eps_C(x)\langle a,\widetilde x\rangle \le 0$ hold, so
$\langle a,\widetilde x\rangle=0$ for all $x$. In particular
$\langle a,\widetilde 0\rangle=a_0=0$, and
$\langle a,\widetilde{e_i}\rangle=a_0+a_i=0$ then gives $a_i=0$ for each $i$ and hence 
 $a=0$.

A \emph{facet} of the full-dimensional cone $\overline C$ is an intersection
$\overline C\cap H$ of dimension $\dim A-1=n$, where $H$ is a supporting
hyperplane of $\overline C$, meaning $\overline C$ lies in one of the two closed
halfspaces bounded by $H$.  For each $x\in\B^n$, $\overline C$ lies in
the closed halfspace
\[
        \eps_C(x)\langle a,\widetilde x\rangle \ge 0,
\]
whose boundary is $H_x$. We say $H_x$ \emph{supports} a facet of $\overline C$ if the intersection
$F_x(\overline C)=\overline C\cap H_x$ has dimension $n$; equivalently, if the
inequality $\eps_C(x)\langle a,\widetilde x\rangle \ge 0$ is \emph{facet-defining}
for $\overline C$.  In that case $F_x(\overline C)$ is the facet supported by $H_x$. 
Note that the hyperplane $H_x$ is not itself the facet, but the ambient hyperplane
containing it. If $\overline C\cap H_x$ has smaller dimension, $H_x$ supports no
facet.

A full-dimensional pointed cone in $\R^d$ is said to be \emph{simplicial} if it is
generated by $d$ linearly independent rays; equivalently, it has exactly $d$
facets. Thus, in the present parameter space $A$, a chamber
closure is simplicial precisely when it has $n+1$ facets.

\begin{lemma}[Facet halfspaces suffice]\label{lem:facet-halfspaces}
Let $K\subseteq\R^d$ be a full-dimensional polyhedral cone. Then $K$ is the
intersection of its facet-defining closed halfspaces.
\end{lemma}

\begin{proof}
The correspondence between the facets of a full-dimensional polyhedron and
the inequalities of an irredundant representation is standard; see
Schrijver \cite[Theorem~8.1 and~(17)]{Schrijver1986}. For completeness we give
the short argument in the form needed here.

Since $K$ is contained in each of its facet-defining halfspaces, it lies in
their intersection. For the reverse inclusion, take any finite representation of
$K$ and discard redundant inequalities, giving an irredundant one,
\[
        K=\{a\in\R^d:\langle u_i,a\rangle\ge 0,\ i=1,\ldots,m\},
\]
meaning one in which no inequality is implied by the others. It suffices that
every $u_i$ be facet-defining, since then the facet-defining halfspaces include
all the halfspaces $\{a:\langle u_i,a\rangle\ge 0\}$, $i=1,\ldots,m$, whose
intersection is $K$.

 Fix $i$. By irredundancy, the $i$-th inequality is not implied by the
others, so some $b$ has $\langle u_j,b\rangle\ge 0$ for all $j\ne i$ and
$\langle u_i,b\rangle<0$. Choose an interior point $c$ of $K$, so
$\langle u_j,c\rangle>0$ for every $j$. Along the segment from $b$ to $c$ the
value of $\langle u_i,\cdot\rangle$ rises from negative to positive, vanishing
at $\bar a=(1-t)b+tc$ for a unique $t\in(0,1)$; and for $j\ne i$,
$\langle u_j,\bar a\rangle=(1-t)\langle u_j,b\rangle+t\langle u_j,c\rangle>0$.
Thus, $\bar a\in K$ lies on $H_i=\{a:\langle u_i,a\rangle=0\}$, with every
inequality $j\ne i$ strict. These strict inequalities persist on a relative
neighbourhood of $\bar a$ within $H_i$, which therefore lies in $K\cap H_i$;
hence $\dim(K\cap H_i)=d-1$, so $u_i$ is facet-defining.
\end{proof}

\begin{lemma}[Facet inequalities determine the chamber]\label{lem:facet-signs}
Let $\Arr$ be a finite central hyperplane arrangement in $\R^d$, and let $C$ be a
chamber of $\Arr$ with closure $\overline C$. If a chamber $D$ of $\Arr$ lies on
the same side as $C$ of every hyperplane that supports a facet of $\overline C$,
then $D=C$.
\end{lemma}

\begin{proof}
Because $C$ is a chamber, $\overline C$ lies in one closed halfspace $H^{+}$
of each $H\in\Arr$, and by Lemma~\ref{lem:facet-halfspaces} the facet halfspaces
alone intersect to give $\overline C$. The hypothesis then gives $D\subseteq\overline C$,
and, since $D$ is open, we have $D=\operatorname{int}D\subseteq\operatorname{int}\overline C$.
The interior of a finite intersection is the intersection of the interiors, so
$\operatorname{int}\overline C$ is the intersection of the open facet halfspaces.
Take a point $p\in\operatorname{int}\overline C$. Then $p$ lies on no facet
hyperplane and, being interior to $\overline C$, which is contained in  $H^{+}$ for every
$H\in\Arr$, lies on no hyperplane of $\Arr$ at all, so $p$ lies in a single
chamber $C'$. Since $p\in\overline C$, every neighbourhood of $p$ meets $C$. The
open chamber $C'$ is such a neighbourhood, so $C'\cap C\neq\varnothing$ and hence
$C'=C$. Therefore $p\in C$. As $p$ was an arbitrary point of
$\operatorname{int}\overline C$, we have $\operatorname{int}\overline C\subseteq C$,
so $D\subseteq C$ and, because distinct chambers are disjoint, $D=C$.
\end{proof}

For Boolean threshold functions, Lemma~\ref{lem:facet-signs} yields, once the
facet hyperplanes are identified with the essential points in
Theorem~\ref{thm:facet} below, that the essential points specify $f$. The lemma
is proved for a general arrangement, so it also covers the feature arrangements
of Proposition~\ref{prop:feature-arrangements}, such as those of the polynomial
threshold functions, which fall outside the Boolean-threshold setting of
\cite{ABST1995}.

\begin{lemma}[Facet normals span]\label{lem:facet-normals-span}
Let $K\subseteq\R^d$ be a full-dimensional pointed polyhedral cone. If
$K$ is written as the intersection of its facet-defining halfspaces,
\[
        K=\{a\in\R^d:\langle u_i,a\rangle\ge 0
        \text{ for } i=1,\ldots,m\},
\]
then the facet normals $u_1,\ldots,u_m$ span $\R^d$. In particular, $m\ge d$.
\end{lemma}

\begin{proof}
If the normals did not span $\R^d$, then there would be a non-zero vector $v$ orthogonal to their span, so that 
\[
        \langle u_i,v\rangle=0 \quad (i=1,\ldots,m).
\]
The same equations hold with $-v$ in place of $v$, so both $v$ and $-v$
satisfy all the defining inequalities of $K$. Thus the line $\R v$ is
contained in $K$, contradicting the fact that it is pointed. Hence the normals span $\R^d$, and the lower bound follows because at least $d$ vectors are needed to do so.
\end{proof}

We use the Boolean sign convention
\[
        \sgn(t)=
        \begin{cases}
        1, & t\ge 0,\\
        0, & t< 0,
        \end{cases}
\]
which matches the definition $f(x)=1 \iff \langle w,x\rangle\ge\theta$ of
Section~\ref{sec:intro}: the positive class is the weak inequality and the
negative class the strict one. The value at $t=0$ is immaterial, since every
chamber satisfies the strict condition $\langle a,\widetilde x\rangle\ne 0$ for
all $x\in\B^n$, so $\sgn$ is only ever evaluated away from $0$.
Each chamber $C$ defines a function $f_C \colon \B^n \to \B$ by
\[
        f_C(x)=\sgn\langle a,\widetilde x\rangle
        \qquad (a\in C),
\]
where the right-hand side does not depend on the choice of $a\in C$, since
$\sgn\langle a,\widetilde x\rangle$ is constant on $C$.

Conversely, every threshold function $f \in \Hn$ arises in this form: we
perturb the threshold $\theta$ slightly so that no cube point satisfies
$\langle a,\widetilde x\rangle=0$. The map $C \mapsto f_C$ is a bijection
between chambers of $\Arr_n$ and threshold functions on $\B^n$ and we write $C_f$
for the chamber corresponding to $f$. Functions in $\Hn$ may have irrelevant variables. We state dependence on all $n$
variables explicitly when required.

The case $n=2$ already exhibits the phenomena of interest, so we record it in
full. The parameter space is $\R^3$, and a function on two bits is
$f(x_1,x_2)=\sgn(a_0+a_1x_1+a_2x_2)$. The four cube points give four planes
$H_x$ through the origin. On each chamber the signs of
the four evaluations are fixed and equal the truth table entries of the corresponding
function, so each chamber is labelled by that truth table. Of the $2^4=16$
truth tables, only two fail to be realised by a threshold function, namely exclusive-or, $x_1\oplus x_2$, and its complement. So there are fourteen chambers.
Table~\ref{tab:n2} lists the fourteen threshold functions with
their essential points and specification numbers, computed from the definition.
Figure~\ref{fig:n2-chambers} shows two of these chambers as cones in the
parameter space~$\R^3$: the simplicial chamber of $x_1\wedge x_2$, with three
facets, and the non-simplicial chamber of $x_1$, with four. 
By Theorem~\ref{thm:facet} the essential points correspond bijectively to the
facets of $\overline{C_f}$, so the eight functions with $\sigma_2(f)=3$ are
exactly those whose chamber is simplicial (Corollary~\ref{cor:simplicial}). A point is
inessential precisely when flipping there yields a parity function, so the simplicial
functions are the eight at Hamming distance one from a parity, namely the
conjunctions and disjunctions of two literals. The two constants and the four
functions of a single variable lie at distance two from both parities and have
all four points essential. Thus a function as simple as $x_1$ has all four
points essential, while a conjunction such as $x_1\wedge x_2$ has only three.
The average is $\overline\sigma_2=(8\cdot3+6\cdot4)/14=24/7$, the entry for
$n=2$ in Table~\ref{tab:numerics}.

\begin{table}[htbp]
\centering
\small
\begin{tabular}{lcccl}
\toprule
$f$ & truth table & $\sigma_2(f)$ & simplicial & $\Ess(f)$ \\
\midrule
$x_1\wedge x_2$          & $0001$ & $3$ & yes & $\{01,10,11\}$ \\
$x_1\wedge\neg x_2$      & $0010$ & $3$ & yes & $\{00,10,11\}$ \\
$\neg x_1\wedge x_2$     & $0100$ & $3$ & yes & $\{00,01,11\}$ \\
$\neg x_1\wedge\neg x_2$ & $1000$ & $3$ & yes & $\{00,01,10\}$ \\
$x_1\vee x_2$            & $0111$ & $3$ & yes & $\{00,01,10\}$ \\
$x_1\vee\neg x_2$        & $1011$ & $3$ & yes & $\{00,01,11\}$ \\
$\neg x_1\vee x_2$       & $1101$ & $3$ & yes & $\{00,10,11\}$ \\
$\neg x_1\vee\neg x_2$   & $1110$ & $3$ & yes & $\{01,10,11\}$ \\
\midrule
$x_1$       & $0011$ & $4$ & no  & $\{00,01,10,11\}$ \\
$\neg x_1$  & $1100$ & $4$ & no  & $\{00,01,10,11\}$ \\
$x_2$       & $0101$ & $4$ & no  & $\{00,01,10,11\}$ \\
$\neg x_2$  & $1010$ & $4$ & no  & $\{00,01,10,11\}$ \\
$0$         & $0000$ & $4$ & no  & $\{00,01,10,11\}$ \\
$1$         & $1111$ & $4$ & no  & $\{00,01,10,11\}$ \\
\bottomrule
\end{tabular}
\caption{The fourteen threshold functions on two bits, realised as the chambers
of $\Arr_2$. A cube point $(x_1,x_2)=(a,b)$ is written $ab$, in the order
$00,01,10,11$; the truth-table column is $f(00)\,f(01)\,f(10)\,f(11)$, which
encodes the chamber sign vector under $0\leftrightarrow-1$ and $1\leftrightarrow+1$. The two non-threshold functions, the parity
$x_1\oplus x_2=0110$ and its complement $1001$, have no chamber. By
Theorem~\ref{thm:facet}, $\sigma_2(f)=|\Ess(f)|$ is the number of facets of
$\overline{C_f}$; the eight functions with $\sigma_2(f)=3$ have simplicial
chambers, the six with $\sigma_2(f)=4$ do not. The average is
$\overline\sigma_2=24/7$.}
\label{tab:n2}
\end{table}

\begin{figure}[htbp]
\centering
\begin{tikzpicture}[line cap=round, line join=round, scale=1.0]
  \useasboundingbox (-1.7,-3.2) rectangle (4.1,3.9);
  \coordinate (O)  at (0,0);
  \coordinate (R1) at (2.9,2.74);  \coordinate (Q1) at (3.71,3.51);
  \coordinate (R2) at (2.9,0.84);  \coordinate (Q2) at (3.71,1.08);
  \coordinate (R3) at (1.0,2.74);  \coordinate (Q3) at (1.28,3.51);
  \fill[Green!55,      opacity=0.26] (O)--(R2)--(R3)--cycle;  
  \fill[BurntOrange!70,opacity=0.26] (O)--(R1)--(R3)--cycle;  
  \fill[RoyalBlue!60,  opacity=0.26] (O)--(R1)--(R2)--cycle;  
  \draw[cut] (R1)--(R2)--(R3)--cycle;
  \draw[thick] (O)--(R1) (O)--(R2) (O)--(R3);
  \draw[rayext] (R1)--(Q1);  \draw[rayext] (R2)--(Q2);  \draw[rayext] (R3)--(Q3);
  \fill (O) circle (1.7pt);
  \node[flab, text=RoyalBlue!70!black]   at (2.55,1.55) {$H_{10}$};
  \node[flab, text=BurntOrange!75!black] at (1.55,2.95) {$H_{01}$};
  \node[flab, text=Green!45!black]       at (1.45,1.10) {$H_{11}$};
  \node[font=\small, anchor=north] at (-0.05,-0.12) {apex};
  \begin{scope}[shift={(-1.25,2.85)}] \chambertriad \end{scope}
  \node[font=\small] at (1.2,-2.9) {(a) $x_1\wedge x_2$:\ $\sigma_2=3$};
\end{tikzpicture}%
\hspace{0.03\textwidth}
\begin{tikzpicture}[line cap=round, line join=round, scale=1.0]
  \useasboundingbox (-1.7,-3.2) rectangle (4.1,3.9);
  \coordinate (O)  at (0,0);
  \coordinate (R1) at (1.9,0);     \coordinate (Q1) at (2.43,0);
  \coordinate (R2) at (1.9,-1.9);  \coordinate (Q2) at (2.43,-2.43);
  \coordinate (R3) at (2.9,2.74);  \coordinate (Q3) at (3.71,3.51);
  \coordinate (R4) at (2.9,0.84);  \coordinate (Q4) at (3.71,1.08);
  \fill[gray!45,       opacity=0.26] (O)--(R1)--(R2)--cycle;  
  \fill[BurntOrange!70,opacity=0.26] (O)--(R1)--(R3)--cycle;  
  \fill[RoyalBlue!60,  opacity=0.26] (O)--(R3)--(R4)--cycle;  
  \fill[Green!55,      opacity=0.26] (O)--(R2)--(R4)--cycle;  
  \draw[cut] (R1)--(R3)--(R4)--(R2)--cycle;
  \draw[thick] (O)--(R1) (O)--(R2) (O)--(R3) (O)--(R4);
  \draw[rayext] (R1)--(Q1); \draw[rayext] (R2)--(Q2);
  \draw[rayext] (R3)--(Q3); \draw[rayext] (R4)--(Q4);
  \fill (O) circle (1.7pt);
  \node[flab, text=gray!40!black]        at (2.15,-0.85) {$H_{00}$};
  \node[flab, text=BurntOrange!75!black] at (1.75,1.55) {$H_{01}$};
  \node[flab, text=RoyalBlue!70!black]   at (3.15,1.95) {$H_{10}$};
  \node[flab, text=Green!45!black]       at (2.35,-1.35) {$H_{11}$};
  \node[font=\small, anchor=east] at (-0.12,0.05) {apex};
  \begin{scope}[shift={(-1.25,2.85)}] \chambertriad \end{scope}
  \node[font=\small] at (1.2,-2.9) {(b) $x_1$:\ $\sigma_2=4$};
\end{tikzpicture}
\caption{Two chambers of $\Arr_2$ as cones in the parameter space $\R^3$ with
coordinates $(a_0,a_1,a_2)$, where $H_x=\{a:\langle a,\widetilde x\rangle=0\}$.
Each cone is unbounded: solid segments are the bounded portion of the extreme
rays, dashed arrows mark their extension to infinity, and the 
cross-section is an arbitrary truncation, not a face. (a) The chamber of
$x_1\wedge x_2$ is simplicial, with three facets, supported by $H_{10}$,
$H_{01}$, and $H_{11}$; the plane $H_{00}$ meets the closure only at the apex,
so $00$ is inessential and $\sigma_2(x_1\wedge x_2)=3$. (b) The chamber of $x_1$
is not simplicial: all four planes $H_{00}$, $H_{01}$, $H_{10}$, and $H_{11}$
support facets, so every cube point is essential and $\sigma_2(x_1)=4$. By
Theorem~\ref{thm:facet}, a cube point $x$ is essential precisely when $H_x$
supports a facet. The TikZ drawing for this figure was produced with the
assistance of Anthropic Claude and verified by the author.}
\label{fig:n2-chambers}
\end{figure}
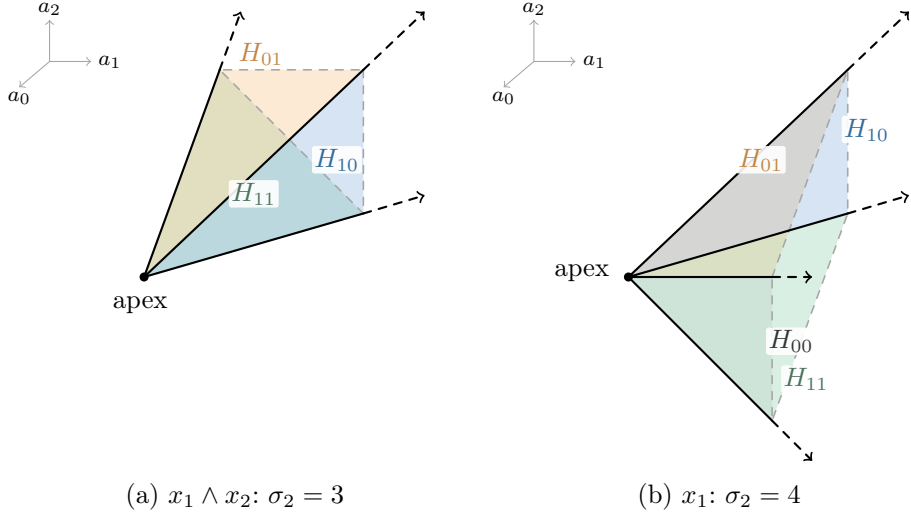

\section{The chamber-facet identification}\label{sec:chamber}

\subsection{Facet crossing and the signature theorem}

We now prove the geometric lemma used throughout the paper. We first isolate two
arrangement-theoretic facts it uses: a local characterisation of when a
hyperplane supports a facet, and the one-sign crossing lemma built on it.

\begin{lemma}[Relative-interior characterisation of facets]\label{lem:facet-relint}
Let $C$ be a chamber of $\Arr_n$ and let $x\in\B^n$. Then $H_x$ supports a facet
of $\overline C$ if and only if there is a point $c$ with the following properties: 
\[
        c\in\overline C\cap H_x
        \qquad\text{and}\qquad
        c\notin H_y \text{ for every } y\neq x .
\]
 When $H_x$ supports a facet, such a point can be taken in the relative
interior of $\overline C\cap H_x$.
\end{lemma}

\begin{proof}
Suppose such a $c$ exists. Since $c\in\overline C$ and $c\notin H_y$ for every
$y\neq x$, each defining inequality $\eps_C(y)\langle a,\widetilde y\rangle\ge0$
with $y\neq x$ is strict at $c$. By continuity these inequalities remain strict
on a relative neighbourhood $U$ of $c$ inside $H_x$, and on $U\subseteq H_x$ the
$x$-inequality holds with equality, so $U\subseteq\overline C\cap H_x$. As $U$
is a non-empty relatively open subset of the hyperplane $H_x$, it has dimension
$n$; hence $\overline C\cap H_x$ has dimension $n$ and $H_x$ supports a facet.

Conversely, suppose $F=\overline C\cap H_x$ is a facet. If $F\subseteq H_y$ for
some $y\neq x$, then $F$ has dimension $n$ and lies in both $H_x$ and $H_y$,
forcing $H_x=H_y$, contrary to Lemma \ref{lem:Hx-injective}. Thus
$F\not\subseteq H_y$ for every $y\neq x$,  and each $F\cap H_y$ lies in
$H_x\cap H_y$, a proper affine subspace of $H_x$ since $H_x\neq H_y$. As there
are finitely many $y\neq x$, the set $\bigcup_{y\neq x}(F\cap H_y)$ lies in a
finite union of proper affine subspaces of $H_x\cong\R^n$. A non-empty open
subset of $H_x$ cannot be covered by finitely many proper affine subspaces, and
$\operatorname{relint}(F)$ is such a subset because $\dim F=n=\dim H_x$; hence
there is a point $c\in\operatorname{relint}(F)$ with $c\notin H_y$ for every
$y\neq x$. This $c$ lies in $\overline C\cap H_x$ and on no $H_y$ with $y\neq x$,
as required. 
\end{proof}

\begin{lemma}[One-sign crossing]\label{lem:one-sign-crossing}
Let $C$ be a chamber of $\Arr_n$ and let $x\in\B^n$. There exists a chamber
$C'$ with
\[
        \eps_{C'}(x)=-\eps_C(x),
        \qquad
        \eps_{C'}(y)=\eps_C(y) \quad (y\neq x)
\]
if and only if $H_x$ supports a facet of $\overline C$. When such a chamber
exists, it is uniquely determined by its sign vector.
\end{lemma}

\begin{proof}
Suppose first that such a chamber $C'$ exists. Choose $a\in C$ and $b\in C'$,
and set
\[
        u(t)=(1-t)a+tb,\qquad 0\le t\le 1.
\]
For $y\neq x$, the numbers $\langle a,\widetilde y\rangle$ and
$\langle b,\widetilde y\rangle$ have the same strict sign, so
$\langle u(t),\widetilde y\rangle$ keeps that sign throughout $[0,1]$ and the
segment avoids $H_y$. At $x$ the signs are opposite, so there is a unique
$t_0\in(0,1)$ with $c=u(t_0)\in H_x$; for $t<t_0$ the point $u(t)$ lies in $C$,
so $c\in\overline C$. Thus $c\in\overline C\cap H_x$ and $c$ lies on no $H_y$
with $y\neq x$, and Lemma \ref{lem:facet-relint} shows that $H_x$ supports a
facet of $\overline C$.

Conversely, suppose $F=\overline C\cap H_x$ is a facet. By Lemma
\ref{lem:facet-relint} there is a point $p\in\operatorname{relint}(F)$ lying on
no $H_y$ with $y\neq x$. Choose $v\in A$ with
$\eps_C(x)\langle v,\widetilde x\rangle>0$. For sufficiently small $\delta>0$,
the points $p+tv$ with $|t|<\delta$ avoid every $H_y$, $y\neq x$, by continuity
at $p$. Since $p\in H_x$, we have
$\langle p+tv,\widetilde x\rangle=t\langle v,\widetilde x\rangle$, so
\[
        \eps_C(x)\langle p+tv,\widetilde x\rangle
        =t\,\eps_C(x)\langle v,\widetilde x\rangle,
\]
which is positive for $t>0$ and negative for $t<0$ by the choice of $v$. Hence
the side $t>0$ lies in $C$, while the side $t<0$ lies in a chamber $C'$ whose
sign vector agrees with that of $C$ at every $y\neq x$ and reverses at $x$.
Finally, a chamber is determined by its sign vector, so $C'$ is unique.
\end{proof}

\begin{theorem}[Facet/signature correspondence]\label{thm:facet}
For $f \in \Hn$ corresponding to chamber $C_f$ of $\Arr_n$, a point
$x \in \B^n$ is essential for $f$ if and only if $H_x$ supports a facet of
$\overline{C_f}$. Moreover, the map
\[
        x\longmapsto \overline{C_f}\cap H_x
\]
is a bijection from $\Ess(f)$ to the facets of $\overline{C_f}$.
Furthermore, $\Ess(f)$ specifies $f$, and every specifying set for $f$ contains
$\Ess(f)$. Consequently $\Ess(f)$ is the \emph{signature} of $f$, the unique
cardinality-minimum specifying set, and
\[
        \sigma_n(f)=|\Ess(f)|.
\]
\end{theorem}

\begin{proof}
We first prove the equivalence between essentiality and facets. Under the
bijection $C\mapsto f_C$, the function $f\oplus\mathbf 1_{\{x\}}$ is
threshold if and only if there is a chamber $C'$ whose sign vector agrees with
that of $C_f$ at every $y\neq x$ and differs at $x$. By Lemma
\ref{lem:one-sign-crossing}, this happens if and only if $H_x$ supports a
facet of $\overline{C_f}$.

The same equivalence gives the stated bijection between essential points and
facets. If $x\in\Ess(f)$, then $\overline{C_f}\cap H_x$ is a facet. Conversely,
every facet of $\overline{C_f}$ is supported by one of the hyperplanes appearing
in the chamber inequalities, say by $H_x$, and then the equivalence gives
$x\in\Ess(f)$. The map is injective by Lemma \ref{lem:Hx-injective}: if two
essential points gave the same facet, that facet would lie in both $H_x$ and
$H_y$, forcing $H_x=H_y$, and hence $x=y$.

We now identify the signature. Every specifying set contains $\Ess(f)$:
if $x\in\Ess(f)\setminus S$, then $f\oplus\mathbf 1_{\{x\}}$ is threshold and
agrees with $f$ on $\B^n\setminus\{x\}\supseteq S$, so $S$ does not specify $f$.

Conversely, $\Ess(f)$ specifies $f$. Let $g\in\Hn$ agree with $f$ on
$\Ess(f)$, and let $C_g$ be its chamber. By the facet bijection proved above,
the hyperplanes $H_x$ with $x\in\Ess(f)$ are exactly the hyperplanes supporting
facets of $\overline{C_f}$. Agreement of $g$ and $f$ on $\Ess(f)$ says precisely
that $C_g$ lies on the same side as $C_f$ of each of these facet hyperplanes.
Lemma~\ref{lem:facet-signs} therefore gives $C_g=C_f$, and hence $g=f$.
Thus $\Ess(f)$ specifies $f$.

It follows that $\Ess(f)$ is contained in every specifying set and is itself a
specifying set. Hence it is the unique cardinality-minimum specifying set, and
\[
        \sigma_n(f)=|\Ess(f)|.
\]
\end{proof}

In an arbitrary concept class there need not be a unique minimum specification;
for threshold functions there is. The specification assertion was observed for
point sets in general position by Cover \cite{Cover1965}, without proof, in the
language of ambiguous or extreme points. The Boolean cube is not in general
position, so Cover's result does not directly apply to $\Hn$. Anthony,
Brightwell, and Shawe-Taylor note a proof via Mays' boundary matrices due to Hu
\cite{Hu1965}, and give an independent, self-contained proof for $\Hn$, where the term
signature originates \cite[Theorem~2.11 and Corollary~2.12]{ABST1995};
Theorem~\ref{thm:facet} recovers it from the chamber geometry. 
 Zuev had earlier recorded the corresponding observation in the language of the
graph of threshold functions: the cube points associated with facets of the
threshold cone form a smallest determining set, and their number is the degree
in that graph \cite[\S5]{Zuev1991}.

\subsection{Immediate consequences}

We now derive the lower bound on specification number, the affine-spanning
property for specifying sets, the simplicial-chamber criterion, and the exact
specification number of linear read-once functions.

\begin{corollary}\label{cor:lower}
$\sigma_n(f) \ge n+1$ for every $f \in \Hn$.
\end{corollary}

\begin{proof}
Let $K=\overline{C_f}$. By Theorem \ref{thm:facet}, $\sigma_n(f)$ is the
number of facets of $K$, and the facet-defining normals may be taken to be
$\eps_{C_f}(x)\widetilde x$ for $x\in\Ess(f)$. The cone $K$ is full-dimensional
and pointed in the $(n+1)$-dimensional space $A$ by Section \ref{sec:setup}, so
Lemma \ref{lem:facet-normals-span} shows that these facet normals span
$\R^{n+1}$. A spanning set of $\R^{n+1}$ has at least $n+1$ elements, so $K$ has
at least $n+1$ facets, and hence
\[
        \sigma_n(f)\ge n+1 .
\]
\end{proof}

\begin{corollary}[Affine-spanning property for specifying sets]\label{cor:general-position}
For every $f \in \Hn$ and every specifying set $S$ for $f$, the lifted cube points
\[
        \{\widetilde x : x \in S\}
\]
span $\R^{n+1}$. Equivalently, $S$ contains $n+1$ affinely independent cube
points, meaning points $x_0,\ldots,x_n$ for which
$x_1-x_0,\ldots,x_n-x_0$ are linearly independent.
\end{corollary}

\begin{proof}
The proof of Corollary~\ref{cor:lower} shows that the facet-defining normals
$\eps_{C_f}(x)\widetilde x$, $x\in\Ess(f)$, span $\R^{n+1}$. Since multiplying a
vector by $\eps_{C_f}(x)\in\{-1,+1\}$ does not change its span,
\[
        \operatorname{span}\{\widetilde x:x\in\Ess(f)\}=\R^{n+1}.
\]
Let $S$ be any specifying set for $f$. By Theorem~\ref{thm:facet}, $S$ contains
$\Ess(f)$, so
\[
        \operatorname{span}\{\widetilde x:x\in S\}=\R^{n+1},
\]
and $S$ contains $n+1$ linearly independent lifted cube points. Finally, for
points $x_0,\ldots,x_n\in\B^n$ the lifted cube points
$\widetilde x_0,\ldots,\widetilde x_n$ are linearly independent in $\R^{n+1}$ if
and only if $x_1-x_0,\ldots,x_n-x_0$ are linearly independent in $\R^n$, which
gives the stated equivalence.
\end{proof}

For threshold functions, the statement that every specifying set contains
$n+1$ affinely independent points is \cite[Theorem~2.9]{ABST1995}; the proof
above derives it from the chamber geometry.

\begin{corollary}\label{cor:simplicial}
For $f \in \Hn$, $\sigma_n(f) = n+1$ if and only if $\overline{C_f}$ is a simplicial cone in $\R^{n+1}$.
\end{corollary}

\begin{proof}
By Theorem \ref{thm:facet}, $\sigma_n(f)$ is the number of facets of
$\overline{C_f}$. A full-dimensional pointed cone in $\R^{n+1}$ is simplicial
if and only if it has exactly $n+1$ facets.
\end{proof}

\begin{remark}\label{rem:simplicial-zuev}
Corollaries~\ref{cor:lower} and~\ref{cor:simplicial} are the two immediate
consequences of the facet count of Theorem~\ref{thm:facet}: a full-dimensional pointed
cone in $\R^{n+1}$ has at least $n+1$ facets, with equality exactly when it is
simplicial. Zuev noted the corresponding lower-bound in terms of the
degree in his graph of threshold functions, which equals the number of facets of
the associated cone \cite[\S5]{Zuev1991}. We isolate
the simplicial reformulation for use in Section~\ref{sec:Tn}.
\end{remark}

We now record the standard recursive definition of linear read-once
functions, used in Proposition~\ref{prop:lro-upper} below.

\begin{definition}[Linear read-once function]\label{def:lro}
A \emph{literal} is a variable $x_i$ or its negation $\neg x_i$. A Boolean
function is \emph{linear read-once} (\emph{lro}) if it is constant, or if it
can be represented by a nested formula obtained recursively as follows: every
literal is lro, and if $\ell$ is a literal whose underlying variable does not
occur in an lro formula $t$, then
\[
        \ell\vee t
        \qquad\text{and}\qquad
        \ell\wedge t
\]
are lro. Thus each variable occurs at most once, and the formula is built by
successively adjoining one new literal.
\end{definition}

These are the nested functions of \cite{ABST1995}. The standard
induction, recorded there, shows that every nested Boolean function
is linearly separable. Hence every lro function belongs to $\Hn$.

\begin{proposition}\label{prop:lro-upper}
For every linear read-once threshold function $f$ depending on all $n$
variables, the chamber $\overline{C_f}$ in $\Arr_n$ is simplicial; equivalently,
$\sigma_n(f)=n+1$.
\end{proposition}

\begin{proof}
Linear read-once functions depending on all $n$ variables have
$\sigma_n(f)=n+1$: this is the bound of Anthony, Brightwell, and Shawe-Taylor for
their nested functions \cite[Theorem~2.10]{ABST1995}, which in the re-labelled
form are exactly the linear read-once functions. By
Theorem~\ref{thm:facet}, $\sigma_n(f)=|\Ess(f)|$ is the number of facets of
$\overline{C_f}$, so $\overline{C_f}$ has $n+1$ facets, and by
Corollary~\ref{cor:simplicial} it is simplicial.
\end{proof}

\subsection{Other characterisations of essentiality}

\begin{remark}[Equivalent tests for essentiality]\label{rem:three-languages}
The chamber-facet criterion can be translated into the summability and convex
separation languages used elsewhere in threshold logic. Let
$g=f\oplus\mathbf 1_{\{x\}}$. Then $x$ is essential for $f$ if and only if $g$ is
a threshold function. By the classical summability criterion \cite{Elgot1961}, $x$ is inessential if and only if the
zero and one sets of $g$ admit a balanced multiset equality
\[
        z_1+\cdots+z_k=y_1+\cdots+y_k,
\]
with all $z_i\in g^{-1}(0)$ and all $y_i\in g^{-1}(1)$, repetitions allowed.
Since $f$ itself is threshold, no such equality can use only points whose
labels are unchanged when passing from $f$ to $g$; any certificate for
non-thresholdness of $g$ must involve the flipped point $x$.

Equivalently, $x$ is essential if and only if the convex hulls
\[
        \operatorname{conv}\{\widetilde z:z\in g^{-1}(0)\}
        \quad\text{and}\quad
        \operatorname{conv}\{\widetilde y:y\in g^{-1}(1)\}
\]
are disjoint. 
 This is the standard strict-separation formulation, and gives a
linear-programming test. The separation condition and the summability
certificate are primal and dual: by the theorem of the alternative, the lifted
zero and one sets are strictly separable exactly when no balanced multiset
equality of the above kind exists. Proposition~\ref{prop:lp-test} restates this
in parameter space, testing essentiality by feasibility of a facet of
$\overline{C_f}$.

That a point is essential exactly when some separating hyperplane passes through
it is due to Lozin, Razgon, Zamaraev, Zamaraeva, and Zolotykh
\cite[Theorem~5]{Lozin2018LRO}. Theorem~\ref{thm:facet}
refines this to a bijection with the facets of the weight-space chamber.
\end{remark}

\subsection{A linear-programming certificate for essentiality}\label{ssec:lp-cert}

\begin{proposition}[Linear-programming certificate for essentiality]\label{prop:lp-test}
Let $f\in\Hn$, let $C_f$ be its chamber, and write
$\eps_f(y)=\eps_{C_f}(y)$. For $x\in\B^n$, the following are equivalent:
\begin{enumerate}[label=\textup{(\alph*)}]
\item $x\in\Ess(f)$;
\item the linear system
\[
        \langle a,\widetilde x\rangle=0,
        \qquad
        \eps_f(y)\langle a,\widetilde y\rangle\ge 1
        \quad (y\in\B^n,\ y\ne x)
\]
is feasible.
\end{enumerate}
The same criterion holds for any spanning feature arrangement in the sense of
Proposition~\ref{prop:feature-arrangements}, with $\widetilde y$ replaced by the
feature vector $\phi(y)$.
\end{proposition}

The system has $2^n-1$ inequality constraints, one for each cube point $y\ne x$,
so Proposition~\ref{prop:lp-test} should be understood as an exact certificate-level
characterisation rather than as a scalable algorithm from a compact weight
representation. Its practical role is primarily for truth-table-scale
verification, small-$n$ experiments, and possible future implementations in
which candidate constraints are generated separately. The complexity of
computing $\Ess(f)$ from a compact representation is discussed further in
Section~\ref{sec:open}.

\begin{proof}
If $x\in\Ess(f)$, then $H_x$ supports a facet of $\overline{C_f}$ by Theorem
\ref{thm:facet}. By Lemma \ref{lem:facet-relint} the relative interior of this
facet contains a point $p$ lying on no other hyperplane $H_y$, $y\ne x$, so
\[
        \langle p,\widetilde x\rangle=0,
        \qquad
        \eps_f(y)\langle p,\widetilde y\rangle>0
        \quad (y\ne x).
\]
Since there are finitely many $y\ne x$, scaling $p$ by a positive constant
makes all the strict quantities at least $1$, so the displayed linear system is
feasible.

Conversely, suppose the displayed system is feasible, say at $a$. Then
$\langle a,\widetilde x\rangle=0$ and
$\eps_f(y)\langle a,\widetilde y\rangle\ge 1>0$ for every $y\ne x$, so
$a\in\overline{C_f}\cap H_x$ and $a$ lies on no $H_y$ with $y\ne x$. By Lemma
\ref{lem:facet-relint}, $H_x$ supports a facet of $\overline{C_f}$, and Theorem
\ref{thm:facet} gives $x\in\Ess(f)$. The feature-arrangement version is the
same proof with $\phi(y)$ in place of $\widetilde y$ (see
Proposition~\ref{prop:feature-arrangements}).
\end{proof}

Running this feasibility test for each of the $2^n$ cube points computes
$\Ess(f)$, and hence $\sigma_n(f)$ and the signature, in time polynomial in the
truth-table size $2^n$. The remaining algorithmic question, taken up in
Section~\ref{sec:open}, concerns \emph{compact} representations, where the input
size is $\mathrm{poly}(n)$ rather than $2^n$.

\section{\texorpdfstring{Average specification numbers and shortest-path closure}{Average specification numbers and shortest-path closure}}\label{sec:average}

In this section, we provide two upper bounds on the average specification number
$\overline\sigma_\C$ of a finite class $\C$ with the property that the essential
points specify each member of the class. These are both obtained through the
one-inclusion graph $G(\C)$. We first prove the cardinality bound
$\overline\sigma_\C\le\log_2|\C|$ for essentially specified classes, by applying
an edge-isoperimetric inequality to $G(\C)$. We then recall a bound of Doliwa,
Fan, Simon, and Zilles: the essential-specification property is equivalent to
shortest-path closure of $G(\C)$, and every nontrivial class with the property has
average specification number less than twice its Vapnik--Chervonenkis dimension
\cite[Lemma~30 and Theorem~31]{Doliwa2014}. For threshold functions both bounds
apply. By Theorem~\ref{thm:facet}, the class $\Hn$ is essentially specified, and
hence, by Doliwa, Fan, Simon, and Zilles, shortest-path closed. 
The cardinality bound recovers the $n^2$ bound of Anthony, Brightwell, and
Shawe-Taylor \cite[Lemmas 2.16--2.19, Propositions 2.20--2.21, and Corollaries
2.22--2.23]{ABST1995}, the Vapnik--Chervonenkis bound gives $\overline\sigma_n<2(n+1)$, and a theorem of Fukuda, Tamura, and Tokuyama sharpens this to $\overline\sigma_n\le 2n$.

\subsection{The one-inclusion graph and essential specification}

Let $\Omega$ be a finite set, and identify $\B^\Omega$ with the Hamming cube whose
coordinates are indexed by the points of $\Omega$. So, the vertices of this cube can be viewed as 
all Boolean functions $\Omega\to\B$,  and two vertices are adjacent if and only
if they differ at exactly one point of $\Omega$. When we specialise to Boolean threshold
functions, the finite domain is $\Omega=\B^n$. Thus $\B^\Omega=\B^{\B^n}$ is the Hamming
cube whose coordinates are indexed by input points $x\in\B^n$, and whose
vertices are Boolean functions $\B^n\to\B$. The threshold class $\Hn$ is a
subset of these vertices.

For a finite class $\C\subseteq\B^\Omega$, the \emph{one-inclusion graph} $G(\C)$
has vertex set $\C$ and an edge between $f,g\in\C$ if and only if they differ at exactly
one point of $\Omega$. Equivalently, $G(\C)$ is the induced subgraph of the Hamming
cube $\B^\Omega$ on the vertex set $\C$. This is the standard
one-inclusion graph of Haussler, Littlestone and Warmuth \cite{HLW1994},
specialised to the case where the finite sample is the whole domain $\Omega$.
{For $\C=\Hn$, this graph is  Zuev's graph of threshold functions \cite[\S5]{Zuev1991}.}

{For $f,g\in\C$, write
\[
        d_H(f,g)=|\{x\in \Omega:f(x)\ne g(x)\}|
\]
for their Hamming distance. Every path from $f$ to $g$ in $G(\C)$ has length at least $d_H(f,g)$, since each edge changes only one coordinate. The class $\C$ is called \emph{shortest-path closed} if, for every $f,g\in\C$, the graph $G(\C)$ contains a path of length $d_H(f,g)$ from $f$ to $g$. Equivalently, graph distance in $G(\C)$ agrees with Hamming distance on every pair, so $G(\C)$ is an isometric subgraph of the ambient Hamming cube. We use the terminology of Doliwa \emph{et al.}\ \cite{Doliwa2014}.}

For a non-empty finite class $\C\subseteq\B^\Omega$, we extend the specification
number of Section~\ref{sec:setup} from $\Hn$ to $\C$: write $\sigma_\C(f)$ for the
minimum size of a set $S\subseteq \Omega$ that specifies $f$ relative to $\C$, that
is, such that no other member of $\C$ agrees with $f$ on $S$. For $\C=\Hn$ and
$\Omega=\B^n$ this is the earlier $\sigma_n(f)$.
 We set
\[
        \overline\sigma_\C
        =
        \frac{1}{|\C|}\sum_{f\in\C}\sigma_\C(f).
\]
The relative essential points are
\[
        \Ess_\C(f)
        =
        \{x\in \Omega: f\oplus\mathbf 1_{\{x\}}\in\C\}.
\]
The adjective ``essential'' is relative to the class $\C$: these are precisely
the points whose one-point flips remain inside $\C$. Equivalently, they are the
points that produce neighbours of $f$ in $G(\C)$. We call $\C$
\emph{essentially specified} if $\Ess_\C(f)$ specifies $f$ relative to $\C$ for
every $f\in\C$. For a graph $G$ and a vertex $v$, write $\deg_G(v)$ for the
degree of $v$ in $G$.

We will use the following equivalence, which is Lemma~30 of Doliwa, Fan, Simon, and Zilles
\cite{Doliwa2014}. We include the short proof to identify their terminology with
ours.

\begin{proposition}[Shortest-path closure and essential specification]\label{prop:shortest-path}
For a non-empty finite class $\C\subseteq\B^\Omega$, the following are equivalent:
\begin{enumerate}[label=\textup{(\roman*)}]
\item $\C$ is shortest-path closed;
\item $\C$ is essentially specified.
\end{enumerate}
When these conditions hold, $\Ess_\C(f)$ is the unique minimum specifying set
for every $f\in\C$.
\end{proposition}

\begin{proof}
Suppose first that $\C$ is shortest-path closed, and fix distinct $f,g\in\C$.
Choose a path from $f$ to $g$ of length $d_H(f,g)$. 
Its first edge flips some coordinate $x\in\Ess_\C(f)$, and because the path has
Hamming-minimal length, $f$ and $g$ differ at $x$. Hence $f$ and $g$ are
distinguished by a point of $\Ess_\C(f)$. Since $g\in\C\setminus\{f\}$ was
arbitrary, no member of $\C$ other than $f$ agrees with $f$ on $\Ess_\C(f)$,
so $\Ess_\C(f)$ specifies $f$.

Conversely, suppose that $\C$ is essentially specified. We prove by induction on
$k=d_H(f,g)$ that every pair $f,g\in\C$ is joined in $G(\C)$ by a path of length
$k$. The assertion is trivial for $k=0$. If $k>0$, then, since
$\Ess_\C(f)$ specifies $f$ and $g\ne f$, there is some
$x\in\Ess_\C(f)$ with $f(x)\ne g(x)$. Put
$f'=f\oplus\mathbf 1_{\{x\}}$. Then $f'\in\C$ and
$d_H(f',g)=k-1$. By induction there is a path of length $k-1$ from $f'$ to $g$;
prepending the edge from $f$ to $f'$ gives a path of length $k$.

Finally, every specifying set for $f$ must contain $\Ess_\C(f)$, since omitting
$x\in\Ess_\C(f)$ leaves the neighbour
$f\oplus\mathbf 1_{\{x\}}$ indistinguishable. Under either equivalent condition,
$\Ess_\C(f)$ itself specifies $f$, so it is the unique minimum specifying set.
\end{proof}

The equivalence has a numerical counterpart, the degree identity that drives
both bounds of this section and appears in the proof of Doliwa \emph{et al.}'s
average bound.
\begin{lemma}\label{lem:degree}
For any finite essentially specified (equivalently, shortest-path-closed) class $\C\subseteq\B^\Omega$ and every $f\in\C$,
\[
        \sigma_\C(f)=\deg_{G(\C)}(f).
\]
\end{lemma}

\begin{proof}
Since $\C$ is essentially specified, $\Ess_\C(f)$ specifies $f$ relative to
$\C$, so $\sigma_\C(f)\le|\Ess_\C(f)|$. Conversely, every relative specifying
set contains $\Ess_\C(f)$: if an essential point $x$ were omitted, then
$f\oplus\mathbf 1_{\{x\}}\in\C$ would agree with $f$ on the remaining
points, so the set would not specify $f$. Hence
$\sigma_\C(f)=|\Ess_\C(f)|$. 
Finally $|\Ess_\C(f)|=\deg_{G(\C)}(f)$, since $x\mapsto f\oplus\mathbf 1_{\{x\}}$
is a bijection from $\Ess_\C(f)$ to the neighbours of $f$ in $G(\C)$. The map is
well-defined and injective: by definition $x\in\Ess_\C(f)$ exactly when
$f\oplus\mathbf 1_{\{x\}}\in\C$, and distinct $x$ will flip $f$ at distinct points, so
give distinct functions. It is surjective, since every neighbour of $f$ differs from
$f$ at a single point $x$, which is then essential.
\end{proof}

\subsection{Essentially specified classes}\label{ssec:instances}

We next separate the argument from the special geometry of the Boolean cube. The following
feature-map principle records the part of the chamber-facet argument that we use
again below.

\begin{proposition}[Spanning feature arrangements]\label{prop:feature-arrangements}
Let $\Omega$ be finite, let $V$ be a finite-dimensional real vector space with dual space $V^*$, and let
$\phi:\Omega\to V$ be a map whose image spans $V$ and satisfies
$\phi(x)\neq 0$ for every $x\in \Omega$. Each $a\in V^*$ is a linear functional on $V$; write $a(v)$ for its value at $v\in V$. For $x\in \Omega$, put
\[
        H_x^\phi=
        \{a\in V^*: a(\phi(x))=0\}.
\]
Assume that $x\mapsto H_x^\phi$ is injective. Let $\C_\phi\subseteq\B^\Omega$ be
the chamber class of the arrangement $\Arr_\phi=\{H_x^\phi:x\in \Omega\}$, 
where a chamber $C$ gives the Boolean function
\[
        f_C(x)=\sgn a(\phi(x))
        \qquad (a\in C),
\]
well-defined because $a(\phi(x))$ has constant sign as $a$ ranges over the
chamber $C$, the sign changing only across the hyperplane $H_x^\phi$.
Then $\C_\phi$ is essentially specified (equivalently, shortest-path closed). More precisely, for every chamber
$C$, the map
\[
        x\longmapsto \overline C\cap H_x^\phi
\]
is a bijection from $\Ess_{\C_\phi}(f_C)$ to the facets of $\overline C$, and
$\Ess_{\C_\phi}(f_C)$ specifies $f_C$ inside $\C_\phi$.
\end{proposition}

\begin{proof}
Since $\phi(\Omega)$ spans $V$,
\[
        \bigcap_{x\in \Omega} H_x^\phi
        =
        \{a\in V^*:a(\phi(x))=0\text{ for all }x\in \Omega\}
        =
        \{0\}.
\]
Thus the chamber closures of $\Arr_\phi$ are full-dimensional pointed
polyhedral cones in $V^*$.

The proof of the one-sign crossing lemma and Theorem~\ref{thm:facet} uses only
this spanning property and the injectivity of $x\mapsto H_x^\phi$. Applying the
same argument to $\Arr_\phi$ gives the stated bijection between
$\Ess_{\C_\phi}(f_C)$ and the facets of $\overline C$.

It remains to show that these essential points specify $f_C$ inside
$\C_\phi$. Let $D$ be a chamber such that $f_D$ agrees with $f_C$ on
$\Ess_{\C_\phi}(f_C)$. Then $D$ lies on the same side as $C$ of every facet
hyperplane of $\overline C$. By Lemma~\ref{lem:facet-signs}, this forces
$D=C$, and hence $f_D=f_C$. Therefore $\Ess_{\C_\phi}(f_C)$ specifies $f_C$
inside $\C_\phi$. Hence $\C_\phi$ is essentially specified and, by
Proposition~\ref{prop:shortest-path}, shortest-path closed.
\end{proof}

We illustrate the proposition with three instances of the feature-map
construction, then record where the hypothesis holds more broadly and where it
fails.

\begin{example}[Homogeneous halfspaces]\label{ex:halfspaces}
Let $Y\subseteq \R^d\setminus\{0\}$ be finite, assume that $Y$
spans $\R^d$, and assume that no two distinct points of $Y$ lie on the same
line through the origin. Taking $\Omega=Y$, $V=\R^d$, and $\phi(y)=y$ gives the class
\[
        y\longmapsto \sgn\langle a,y\rangle,
        \qquad a\in\R^d,
\]
and Proposition~\ref{prop:feature-arrangements} says that this class is
essentially specified. Its relative essential points are precisely the data
points whose inequalities support facets of the corresponding realization cone.
\end{example}

\begin{example}[Boolean threshold functions]\label{ex:threshold-feature}
Taking $\Omega=\B^n$, $V=\R^{n+1}$, and
$\phi(x)=(1,x_1,\ldots,x_n)$, we recover the usual class of Boolean threshold
functions.
The lifted cube points span $\R^{n+1}$, and distinct lifted cube points give
distinct hyperplanes. Thus Proposition~\ref{prop:feature-arrangements}
recovers the chamber-facet statement for Boolean threshold functions.
\end{example}

\begin{example}[Polynomial threshold functions]\label{ex:ptf}
For $1\le d\le n$, the degree-$d$ monomial feature map
\[
        \phi_d(x)=\left(\prod_{i\in S}x_i\right)_{|S|\le d}
\]
gives the feature arrangement for degree-$d$ polynomial threshold functions,
which is the setting used below in Proposition~\ref{prop:ptf-average}.
The map $\phi_d$ satisfies the hypotheses of
Proposition~\ref{prop:feature-arrangements}, verified in the proof of
Proposition~\ref{prop:ptf-average}, so that proposition gives essential
specification for $\C_d$. The VC
estimate is then the theorem of Doliwa \emph{et al.}\ \cite{Doliwa2014}, and the
remaining numerical bounds parallel the threshold case with $N_d$ in place of
$n+1$.
\end{example}

\begin{remark}[Scope of the hypothesis]\label{rem:essential-hypothesis}
Essential specification (equivalently, shortest-path closure) is the standing
hypothesis of both bounds in this section. It is neither automatic nor confined
to linear separation. The monotone classes of Example~\ref{ex:monotone} satisfy
it without arising from linear separation, while the class of
Example~\ref{ex:non-spc} fails it, with $\sigma_\C(f)>|\Ess_\C(f)|$ there, so the
degree identity of Lemma~\ref{lem:degree} requires the hypothesis and is not a
general identity.
\end{remark}

\begin{example}[Monotone Boolean functions]\label{ex:monotone}
Let $P$ be a finite poset and let $\mathcal M(P)$ be the class of monotone
Boolean functions on $P$, those $f$ with $f(u)\le f(v)$ whenever $u\le v$. For
such an $f$, a point $x$ with $f(x)=1$ can have its value flipped to $0$ without destroying
monotonicity precisely when no point strictly below $x$ also takes the value
$1$; that is, when $x\in\operatorname{Min}(f^{-1}(1))$. Dually, a point $x$ with
$f(x)=0$ can have its value flipped to $1$ precisely when no point strictly above $x$ takes
the value $0$; that is, when $x\in\operatorname{Max}(f^{-1}(0))$. Hence
\[
        \Ess_{\mathcal M(P)}(f)
        =\operatorname{Min}(f^{-1}(1))\cup\operatorname{Max}(f^{-1}(0)),
\]
the minimal ones together with the maximal zeros. This boundary set specifies
$f$ inside $\mathcal M(P)$. For,  if $h\in\mathcal M(P)$ agrees with $f$ on it, then
any $p$ with $f(p)=1$ lies above some $m\in\operatorname{Min}(f^{-1}(1))$, so
$h(m)=1$ and monotonicity force $h(p)=1$, while any $p$ with $f(p)=0$ lies
below some $M\in\operatorname{Max}(f^{-1}(0))$, so $h(M)=0$ and monotonicity
force $h(p)=0$; thus $h=f$. The class $\mathcal M(P)$ is therefore essentially
specified, although it is defined by monotonicity rather than by linear
separation.
\end{example}

\begin{example}[A class that is not essentially specified]\label{ex:non-spc}
Write a function on $\Omega=\{a,b,c\}$ as the
string $f(a)f(b)f(c)$, and take
\[
        \C=\{000,001,010,101,110\}\subseteq\B^\Omega.
\]
For $f=101$, flipping $a$ gives $001\in\C$, while flipping $b$ gives
$111\notin\C$ and flipping $c$ gives $100\notin\C$, so $\Ess_\C(101)=\{a\}$. For
$f=110$, flipping $a$ gives $010\in\C$, while the other two flips leave $\C$, so
$\Ess_\C(110)=\{a\}$. The functions $101$ and $110$ thus have the same relative
essential set $\{a\}$ and agree there, both taking the value $1$ at $a$, yet
they are distinct members of $\C$. Neither is separated from the other by its
essential set, so neither essential set specifies its function. Indeed
$\sigma_\C(101)=2$ while $|\Ess_\C(101)|=1$, so the equality
$\sigma_\C(f)=|\Ess_\C(f)|$ underlying Lemma~\ref{lem:degree} fails here.
Equivalently, the class is not shortest-path closed: the Hamming distance between $101$ and $110$ is $2$, whereas the path $101,001,000,010,110$ has length $4$ and no shorter path lies in $G(\C)$.
Essential specification (equivalently, shortest-path closure) is therefore a genuine structural assumption, satisfied by threshold functions (Theorem~\ref{thm:facet}) and by the monotone classes above, but not by every finite class.
\end{example}

\subsection{The cardinality bound}

We use a standard weak edge-isoperimetric inequality for the cube, stated below
in internal-edge form. Write $\B^m$ for the $m$-dimensional hypercube, in which
every vertex has degree $m$. For $A\subseteq\B^m$, let $e(A)$ be the number of
edges with both endpoints in $A$, and let $\partial A$ be the set of edges with
exactly one endpoint in $A$. Counting in two ways the edges incident to $A$
gives
\[
        m|A|=2e(A)+|\partial A|.
\]
Through this identity, the internal-edge bound of Lemma~\ref{lem:isop} is
equivalent to the boundary form
\[
        |\partial A|\ge |A|\log_2(2^m/|A|).
\]
The boundary form is a standard consequence of the sharp edge-isoperimetric
theorem for the cube, established in separate papers by Harper
\cite{Harper1964}, Lindsey \cite{Lindsey1964}, Bernstein \cite{Bernstein1967},
and Hart \cite{Hart1976}, and recorded, for example, by Ellis
\cite{Ellis2011}. The proof below is the direct
induction proof, included for completeness and to keep the argument
self-contained.

\begin{lemma}[Hypercube edge-isoperimetric bound]\label{lem:isop}
For every $A\subseteq\B^m$,
\[
        e(A)\le \frac{1}{2}|A|\log_2 |A|,
\]
with the convention $0\log_2 0=0$.
\end{lemma}

\begin{proof}
We argue by induction on $m$. The case $m=0$ is immediate, since $|A|\le 1$ and there are no edges.

Assume $m\ge 1$, and write the last coordinate as the splitting coordinate. Set
\[
        A_0=\{x\in\B^{m-1}:(x,0)\in A\},
        \qquad
        A_1=\{x\in\B^{m-1}:(x,1)\in A\},
\]
and put $a_i=|A_i|$. Interchanging the labels $0$ and $1$ in the last coordinate if necessary, assume that $a_0\ge a_1$.

Every edge of the induced subgraph on $A$ is of exactly one of the following three types: an edge inside the $0$-slice, an edge inside the $1$-slice, or a vertical edge joining $(x,0)$ to $(x,1)$. Hence
\[
        e(A)=e(A_0)+e(A_1)+|A_0\cap A_1|.
\]
By the induction hypothesis, and since $|A_0\cap A_1|\le a_1$, this gives
\[
        e(A)
        \le
        \frac{1}{2}a_0\log_2 a_0
        +
        \frac{1}{2}a_1\log_2 a_1
        +
        a_1 .
\]

It remains to compare this with $\frac12(a_0+a_1)\log_2(a_0+a_1)$. If
$a_0+a_1=0$ there is nothing to prove. Otherwise set $M=a_0+a_1$ and
$p=a_1/M$, so $0\le p\le 1/2$ because $a_0\ge a_1$. With the binary entropy
function $H_2(p)=-(1-p)\log_2(1-p)-p\log_2 p$,
\[
        \tfrac12 M\log_2 M-\tfrac12 a_0\log_2 a_0-\tfrac12 a_1\log_2 a_1
        =\tfrac12 M H_2(p),
\]
so the inductive step reduces to $a_1\le\tfrac12 M H_2(p)$, that is, to
$2p\le H_2(p)$ on $[0,1/2]$. The function $H_2$ is concave with $H_2(0)=0$ and
$H_2(1/2)=1$, so its graph lies on or above the chord through those endpoints,
the line $p\mapsto 2p$; hence $2p\le H_2(p)$ for $p\in[0,1/2]$, which completes
the induction.
\end{proof}

\begin{theorem}\label{thm:average}
For any non-empty finite essentially specified (equivalently, shortest-path-closed) class $\C\subseteq\B^\Omega$,
\[
        \overline\sigma_\C\le \log_2|\C|.
\]
In particular, for the class of Boolean threshold functions,
\[
\overline\sigma_n \;\le\; \log_2 |\Hn|, \qquad \overline\sigma_n \;\le\; n^2 \quad (n \ge 2).
\]
\end{theorem}

\begin{proof}
By Lemma \ref{lem:degree},
\[
        \sum_{f\in\C}\sigma_\C(f)
        =
        \sum_{f\in\C}\deg_{G(\C)}(f)
        =
        2|E(G(\C))|.
\]
Since $G(\C)$ is the induced subgraph of the Hamming cube $\B^\Omega$ on the vertex set $\C$, Lemma \ref{lem:isop}, applied with $A=\C$, gives
\[
        |E(G(\C))|\le \frac12|\C|\log_2|\C|.
\]
Therefore
\[
        \sum_{f\in\C}\sigma_\C(f)\le |\C|\log_2|\C|,
\]
and division by $|\C|$ gives
\[
        \overline\sigma_\C\le \log_2|\C|.
\]
The threshold class is essentially specified by Theorem \ref{thm:facet}, so the
first threshold bound follows. For the second, Cover's general-position count
\cite{Cover1965} gives
\[
        |\Hn|
        \le 2\sum_{j=0}^{n}\binom{2^n-1}{j}
        \le 2^{n^2}
        \qquad (n\ge2),
\]
hence $\log_2|\Hn|\le n^2$. Remark~\ref{rem:Hn-asymptotics} records sharper
estimates.
\end{proof}

\subsection{The Vapnik--Chervonenkis dimension bound}

The logarithmic bound is cardinality-based and holds for every essentially
specified class. When the class has small Vapnik--Chervonenkis dimension, the
same degree identity yields a stronger bound. A set $T\subseteq \Omega$ is
\emph{shattered} by $\C\subseteq\B^\Omega$ if every Boolean function on $T$ extends
to a member of $\C$, that is, $\{f|_T:f\in\C\}$ is all of $\B^T$. The
\emph{Vapnik--Chervonenkis dimension} $\mathrm{VCD}(\C)$ is the size of the
largest shattered set. For the threshold class, $\mathrm{VCD}(\Hn)=n+1$. The
points $0,e_1,\ldots,e_n$ are affinely independent, hence shattered by
halfspaces, so $\mathrm{VCD}(\Hn)\ge n+1$, while every threshold function is the
restriction to $\B^n$ of an affine halfspace in $\R^n$, a class of
Vapnik--Chervonenkis dimension $n+1$, so no $n+2$ points are shattered and therefore
$\mathrm{VCD}(\Hn)\le n+1$.

\begin{theorem}[Doliwa--Fan--Simon--Zilles bound]\label{thm:vc-average}
Let $\C\subseteq\B^\Omega$ be a finite essentially specified (equivalently,
shortest-path-closed) class with $|\C|>1$, of Vapnik--Chervonenkis dimension
$d=\mathrm{VCD}(\C)$.  Then
\[
        \overline\sigma_\C<2d.
\]
In particular, for the class of Boolean threshold functions on $\B^n$,
\[
        \overline\sigma_n< 2(n+1).
\]
\end{theorem}

\begin{proof}
By Proposition~\ref{prop:shortest-path}, $\C$ is shortest-path closed, so the strict inequality is Theorem~31 of Doliwa, Fan, Simon, and Zilles \cite{Doliwa2014}. We record their mechanism in the present notation.
By Lemma \ref{lem:degree}, $\sigma_\C(f)=\deg_{G(\C)}(f)$ for every $f\in\C$,
so
\[
        \overline\sigma_\C
        =\frac{1}{|\C|}\sum_{f\in\C}\deg_{G(\C)}(f)
        =\frac{2|E(G(\C))|}{|\C|}.
\]
The one-inclusion graph $G(\C)$ is the subgraph of the Hamming cube $\B^\Omega$
induced on $\C$, and Haussler, Littlestone, and Warmuth proved that the density
of this graph is bounded by the Vapnik--Chervonenkis dimension:
\[
        \frac{|E(G(\C))|}{|\C|}< \mathrm{VCD}(\C)
\]
 in the form used by Doliwa \emph{et al.}\ \cite[Theorem~31]{Doliwa2014}; see also \cite[Lemma~2.4]{HLW1994} and Haussler \cite{Haussler1995}. Hence $\overline\sigma_\C<2\,\mathrm{VCD}(\C)$. The threshold class is essentially
specified by Theorem \ref{thm:facet}, and $\mathrm{VCD}(\Hn)=n+1$ as recorded
above, so $\overline\sigma_n< 2(n+1)$.
\end{proof}

\begin{remark}
Theorem~\ref{thm:vc-average} gives $\overline\sigma_n<2(n+1)$ from the combinatorics of the one-inclusion graph alone. A sharper arrangement-theoretic route gives the bound $\overline\sigma_n\le 2n$: a theorem of Fukuda, Tamura, and Tokuyama bounds the average facet count of a spherical arrangement, and the threshold average $\overline\sigma_n$ can be read as such an average on the sphere. We record this next.
\end{remark}

\begin{theorem}[Fukuda--Tamura--Tokuyama bound]\label{thm:ftt-average}
For every $n\ge1$,
\[
        \overline\sigma_n\le 2n.
\]
\end{theorem}

\begin{proof}
Fukuda, Tamura, and Tokuyama proved that, in a $d$-dimensional spherical arrangement, the average number of $j$-subfaces of a $k$-face is at most $2^{k-j}\binom{k}{j}$ \cite[Corollary~1.3]{Fukuda1993}; equivalently this is their oriented-matroid theorem \cite[Theorem~1.2]{Fukuda1993}. The bound assumes neither simplicity nor genericity of the arrangement.

Intersect the central arrangement $\Arr_n$ with the unit sphere $S^n\subseteq\R^{n+1}$. Each chamber of $\Arr_n$ is an open cone, so it meets $S^n$ in a single $n$-face of the resulting spherical arrangement of great $(n-1)$-spheres, and this radial section preserves the number of facets. Pointedness justifies this step: since the normals of $\Arr_n$ span $\R^{n+1}$, each closed chamber is a pointed full-dimensional cone, which meets $S^n$ in a single spherically convex $n$-face whose $(n-1)$-faces are in bijection with the facets of the cone. The chambers therefore correspond bijectively with the $n$-faces, and the average number of facets over the $n$-faces is exactly $\overline\sigma_n$. Applying the bound on $S^n$, which has dimension $n$, with $d=k=n$ and $j=n-1$ gives
\[
        \overline\sigma_n\le 2^{\,n-(n-1)}\binom{n}{n-1}=2n .
\]
\end{proof}

\begin{remark}[The two bounds compared]\label{rem:two-bounds}
The cardinality bound of Theorem~\ref{thm:average} and the VC bound of
Theorem~\ref{thm:vc-average} are related by a two-sided inequality on any finite
domain. In one direction $\mathrm{VCD}(\C)\le\log_2|\C|$ for every finite class,
since a shattered set of size $d$ forces $|\C|\ge 2^d$. In the other direction,
the Sauer--Shelah lemma \cite{Sauer1972,Shelah1972} says that a class of
Boolean functions on an $m$-point set with Vapnik--Chervonenkis dimension $d$
realizes at most
\[
        \sum_{i=0}^d\binom{m}{i}
\]
distinct labellings of that set. For a learning-theory account, see also
\cite[Chapter~3]{AnthonyBartlett1999}. Applying this with the $m$-point set
equal to the whole domain $\Omega$, where $m=|\Omega|$, gives
\[
        |\C|\le \sum_{i=0}^d\binom{m}{i}.
\]
For $d\ge1$, the standard estimate
$\sum_{i=0}^d\binom{m}{i}\le(em/d)^d$, with $d=\mathrm{VCD}(\C)$, gives
\[
        \log_2|\C|\le d\log_2\frac{em}{d}\le \bigl(\log_2|\Omega|+2\bigr)\,d .
\]
When $d=0$ the non-empty class $\C$ is a single function, so
$\log_2|\C|=0=\mathrm{VCD}(\C)$ and the comparison below holds trivially. Hence,
for $d\ge1$,
\[
        \mathrm{VCD}(\C)\le\log_2|\C|\le\bigl(\log_2|\Omega|+2\bigr)\,\mathrm{VCD}(\C),
\]
and the two upper bounds on $\overline\sigma_\C$ are equivalent up to a factor of
order $\log_2|\Omega|$. Within that range neither dominates. The cardinality bound is
the tighter when $\log_2|\C|$ is close to $\mathrm{VCD}(\C)$. For example, for the full cube
$\C=\B^\Omega$ with $d=|\Omega|$ every coordinate is essential and
$\overline\sigma_\C=d=\log_2|\C|$, against a VC bound of $2d$. The monotone
classes of Example~\ref{ex:monotone} are of the
same type, since on $\{0,1\}^k$ they shatter exactly the antichains, giving
$\mathrm{VCD}=\binom{k}{\lfloor k/2\rfloor}$ by Sperner's theorem \cite{Sperner1928}. The
cardinality of this monotone class is the $k$th Dedekind number, whose base-two
logarithm is asymptotic to $\binom{k}{\lfloor k/2\rfloor}$ \cite{Korshunov2003}. The VC bound is the
tighter when $\log_2|\C|$ is far above $\mathrm{VCD}(\C)$. For example, for the threshold
class, 
$\log_2|\Hn|=n^2(1-o(1))$ \cite{Zuev1989} against
$\mathrm{VCD}(\Hn)=n+1$, a ratio of order $n$, the largest the inequality allows
on the domain $\Omega=\B^n$ (where $\log_2 |\Omega|=n$). Combined with the universal lower bound $\sigma_n(f)\ge n+1$
(Corollary~\ref{cor:lower}),
the VC bound fixes the order of growth of $\overline\sigma_n$. We come back to this 
in Section~\ref{sec:growth}.
\end{remark}

\begin{remark}[Sharper estimates on $|\Hn|$]\label{rem:Hn-asymptotics}
The bound $|\Hn|\le 2^{n^2}$ used in the proof admits successively sharper
refinements. 
Zuev \cite{Zuev1989} proved the leading-order estimate
$\log_2 |\Hn| = n^2(1-o(1))$. Kahn, Komlós, and Szemerédi \cite[Section~4]{KahnKomlosSzemeredi1995} sharpened this to
\[
        \log_2 |\Hn| = n^2 - n\log_2 n + O(n).
\]
\end{remark}
\begin{remark}[Comparison with arbitrary finite classes]
The gain from the universal finite-class bound to Theorem~\ref{thm:average} is
structural, not numerical. For an arbitrary finite class $\C$, Bondy's
separating-set lemma with the argument of Kushilevitz, Linial, Rabinovich, and
Saks gives $\overline\sigma_\C\le 2\sqrt{|\C|}$, and the same authors construct
classes with average witness size $\Omega(\sqrt{|\C|})$ from finite projective
planes \cite{Bondy1972,KushilevitzLinialRabinovichSaks1996}, so no
cardinality-only bound improves the order $\sqrt{|\C|}$. For $\C=\Hn$, Zuev's
estimate $\log_2|\Hn|=n^2(1-o(1))$ \cite{Zuev1989} makes this universal bound
$2\sqrt{|\Hn|}=2^{n^2/2+o(n^2)}$, against $\overline\sigma_n\le n^2(1+o(1))$ from
Theorem~\ref{thm:average}, so essential specification, not a sharper count,
accounts for the separation. This does not assert that $n^2$ is sharp for $\Hn$,
since the facet and Vapnik--Chervonenkis bounds give $\overline\sigma_n=\Theta(n)$.
\end{remark}

\begin{remark}[Comparison with generic arrangements]
Cover's general-position calculation gives a second comparison. For $N$ points in general
position and homogeneous halfspaces in $d$ parameters, Cover
\cite[Eq.~(39)]{Cover1965} gives the expected number of extreme points of a
uniformly random separable dichotomy as
\[
        \frac{2N\,C(N-1,d-1)}{C(N,d)},
        \qquad
        C(N,d)=2\sum_{j=0}^{d-1}\binom{N-1}{j},
\]
which is the average facet count of a chamber in a generic central arrangement of
$N$ hyperplanes in $\R^d$, one in which every $d$ of the normals are linearly
independent. For fixed $d\ge2$ this ratio tends to $2(d-1)$ as $N\to\infty$, so
after the affine lift $x\mapsto(1,x)$, where $d=n+1$, the generic fixed-$n$ limit
is $2n$. The Boolean threshold arrangement is not generic, since its normals, the
lifted cube points $(1,x)$, satisfy many linear dependencies, so Cover's value is
a comparison rather than a computation of $\overline\sigma_n$. The exact formula
of Corollary~\ref{cor:exact} replaces the generic restriction count $C(N-1,d-1)$
by the chamber count of the resonance arrangement.
\end{remark}

\subsection{Polynomial threshold functions}

We now consider polynomial threshold functions as the main example of the
feature-arrangement framework of Proposition~\ref{prop:feature-arrangements}.
This lets us apply the same chamber-facet mechanism to a richer feature arrangement and show that the lower bound, simpliciality criterion, and dimension-sensitive
average bound are not specific to the linear setting. Let
\[
        \mathcal M_{n,d}=\{S\subseteq[n]: |S|\le d\},
        \qquad
        N_d=|\mathcal M_{n,d}|=\sum_{i=0}^d\binom{n}{i}.
\]
A \emph{degree-$d$ polynomial threshold function} on $\B^n$ (more precisely, a polynomial threshold function of degree at most $d$) is a Boolean function $f:\B^n\to\B$ for which there are real coefficients $a_S$, $S\in\mathcal M_{n,d}$, such that
\[
        f(x)=
        \sgn\left(
        \sum_{S\in\mathcal M_{n,d}}
        a_S\prod_{i\in S}x_i
        \right)
        \qquad (x\in\B^n),
\]
with the representing polynomial non-zero at every point of $\B^n$. Let $\C_d$
denote the class of such functions. The degree-$d$ monomial-feature arrangement
below has chambers naturally labelled by members of $\C_d$, and this labelling
gives a bijection between its chambers and $\C_d$. The non-vanishing condition
places the coefficient vector strictly inside a chamber, off every one of the
arrangement's hyperplanes, which is what makes the labelling a bijection.

\begin{proposition}[Polynomial threshold functions]\label{prop:ptf-average}
Let $1\le d\le n$, and let $\C_d$ be the class of degree-$d$ polynomial threshold
functions on $\B^n$, in the sense above, with specification number taken inside
$\C_d$. Then every $f\in\C_d$ satisfies
\[
        \sigma_{\C_d}(f)\ge N_d,
\]
with equality precisely when the corresponding chamber in the degree-$d$
monomial-feature arrangement is simplicial, and
\[
        N_d\le\overline\sigma_{\C_d}\le 2N_d-2 .
\]
\end{proposition}

\begin{proof}
Let $V_d=\R^{\mathcal M_{n,d}}$, the coordinate space indexed by the subsets
$S\subseteq[n]$ with $|S|\le d$; its dimension is $N_d$. A coefficient vector
$a=(a_S)_{S\in\mathcal M_{n,d}}\in V_d$ gives the polynomial
\[
        p_a(x)=\sum_{S\in\mathcal M_{n,d}}a_Sx_S,
        \qquad
        x_S=\prod_{i\in S}x_i \quad (x_\varnothing=1),
\]
and the monomial-feature vector is
$\phi_d(x)=(x_S)_{S\in\mathcal M_{n,d}}\in V_d$, so that
$p_a(x)=\langle a,\phi_d(x)\rangle$ in the standard inner product. The central
arrangement is
\[
        \Arr_{n,d}
        =
        \{H_x^{(d)}:x\in\B^n\},
        \qquad
        H_x^{(d)}
        =
        \{a\in V_d: p_a(x)=0\}=\phi_d(x)^\perp,
\]
and its chambers are the functions in $\C_d$.

We verify the hypotheses of Proposition \ref{prop:feature-arrangements}. First,
$\phi_d(\B^n)$ spans $V_d$; equivalently, the monomial functions $x_S$ with
$|S|\le d$ are linearly independent on $\B^n$. This is \cite[Proposition~8]{Anthony1995}. (Evaluate $\sum_{|S|\le d}c_Sx_S=0$ at the indicator vectors
$1_T$ and induct on $|T|$.) 

Second, the map $x\mapsto H_x^{(d)}$ is injective. If $H_x^{(d)}=H_y^{(d)}$, then $\phi_d(x)$ and $\phi_d(y)$ are proportional. Their $\varnothing$-coordinates are both $1$, so the proportionality scalar is $1$. Since $d\ge 1$, the singleton coordinates then give $x_i=y_i$ for all $i$, and hence $x=y$.

Proposition \ref{prop:feature-arrangements} therefore applies: $\C_d$ is
essentially specified (equivalently, shortest-path closed) and the map $x\mapsto\overline{C_f}\cap H_x^{(d)}$
identifies $\Ess_{\C_d}(f)$ with the facets of $\overline{C_f}$. Since $\C_d$ is
essentially specified, Lemma~\ref{lem:degree} gives
$\sigma_{\C_d}(f)=\deg_{G(\C_d)}(f)=|\Ess_{\C_d}(f)|$, and therefore
\[
        \sigma_{\C_d}(f)
        =
        \#\{\text{facets of }\overline{C_f}\}.
\]

Since $\overline{C_f}$ is a full-dimensional pointed polyhedral cone in the $N_d$-dimensional space $V_d$, it has at least $N_d$ facets. Therefore
\[
        \sigma_{\C_d}(f)\ge N_d.
\]
Equality holds exactly when $\overline{C_f}$ has exactly $N_d$ facets, which is equivalent to $\overline{C_f}$ being simplicial.

For the average, each $f\in\C_d$ has the form $x\mapsto\sgn\langle a,\phi_d(x)\rangle$
for some $a\in V_d=\R^{N_d}$, so $\C_d$ is contained in the class of homogeneous
halfspaces on $V_d$ restricted to $\phi_d(\B^n)$. That class has
Vapnik--Chervonenkis dimension $N_d$, and restriction to a subset of the domain
does not increase it, so $\mathrm{VCD}(\C_d)\le N_d$. 
 Since $\C_d$ is essentially specified, Theorem~\ref{thm:vc-average} gives $\overline\sigma_{\C_d}< 2\,\mathrm{VCD}(\C_d)\le 2N_d$, and the cone lower bound $\sigma_{\C_d}(f)\ge N_d$ gives $\overline\sigma_{\C_d}\ge N_d$, so $N_d\le\overline\sigma_{\C_d}< 2N_d$. The sphere-arrangement argument proving Theorem~\ref{thm:ftt-average}, applied to $\Arr_{n,d}$ on $S^{N_d-1}$, shows that $\overline\sigma_{\C_d}$ is the average facet count of the cells there, and the theorem of Fukuda, Tamura, and Tokuyama \cite[Corollary~1.3]{Fukuda1993} with $d=N_d-1$ sharpens this to $N_d\le\overline\sigma_{\C_d}\le 2N_d-2$.
\end{proof}

\begin{remark}
The cardinality bound $\overline\sigma_{\C_d}\le\log_2|\C_d|$ of
Theorem~\ref{thm:average} also applies. The two bounds behave differently at
opposite ends of the degree range. For fixed $d$, Baldi and Vershynin's estimate
$\log_2|\C_d|=(1-o(1))\,nN_d$ \cite{BaldiVershynin2019} makes the cardinality
bound the weaker, giving only $O(nN_d)$ against the Vapnik--Chervonenkis bound
$O(N_d)$. At the other extreme, when $d=n$, the class $\C_n$ is all of
$\B^{\B^n}$ and $\log_2|\C_n|=2^n=N_n$, so the cardinality bound gives the exact
value $N_n$, while the Vapnik--Chervonenkis bound gives $2N_n$; this is the
full-cube case of Remark~\ref{rem:two-bounds}. We do not attempt here to locate
the crossover for intermediate degrees.
\end{remark}

In the next section we develop the exact facet-count approach, which specialises in the linear case to the resonance arrangement.

\section{Exact formulas via resonance arrangements}\label{sec:exact}

In this section, we explore how the chamber-facet identification gives an exact formula for
$\sum_{f\in\Hn}\sigma_n(f)$ (and hence for the average specification number) through 
 chamber counts of restricted arrangements.

\subsection{The chamber--facet incidence formula}

We first record the standard arrangement-theoretic incidence formula that
underlies this computation. This is the codimension-one incidence form of
Zaslavsky's face-count formula \cite{Zaslavsky1975}; see also Stanley's
formulation of face counts for restricted arrangements
\cite[Lecture 2, Definition 2.4 and the subsequent face-count formula]{Stanley2007}.
We include the proof to fix conventions and to make explicit why restrictions
of the arrangement appear.

Let $\Arr$ be a finite central hyperplane arrangement in $\R^d$. We say that
$\Arr$ is \emph{full rank} if its hyperplanes meet only at the origin; in other words, 
$\bigcap_{H\in\Arr}H=\{0\}$.  This means, equivalently, that the normals to the hyperplanes span $\R^d$. The closed
chambers of such an arrangement are then pointed cones. For such an arrangement and
$H \in \Arr$, the \emph{restriction} $\Arr^H$ is the arrangement on $H$ formed
by the proper intersections
\[
        \{H \cap K : K \in \Arr,\ H\cap K \neq H\}.
\]

\begin{lemma}[Standard chamber--facet incidence]\label{lem:facet-count}
For any finite full-rank central arrangement $\Arr$ in $\R^d$,
\[
\sum_{C \in \Ch(\Arr)} \#\{\text{facets of } \overline C\} \;=\; 2 \sum_{H \in \Arr} r(\Arr^H).
\]
For the Boolean threshold arrangement $\Arr_n$ on $\R^{n+1}$,
\[
\sum_{f \in \Hn} \sigma_n(f) \;=\; 2 \sum_{x \in \B^n} r(\Arr_n^{H_x}).
\]
\end{lemma}

\begin{proof}
We first prove the arrangement-theoretic identity. The left-hand side counts
ordered pairs $(C,F)$, where $C$ is a chamber of $\Arr$ and $F$ is a facet of
$\overline C$. Every such facet is supported by a unique hyperplane
$H\in\Arr$. Indeed, the relative interior of a facet cannot be
contained in two distinct hyperplanes, since the intersection of two distinct
hyperplanes has codimension at least two.

Fix $H\in\Arr$. The chambers of the restricted arrangement $\Arr^H$ are the
connected components of
\[
        H\setminus \bigcup_{\substack{K\in\Arr\\ H\cap K\neq H}}(H\cap K).
\]
Let $R$ be a chamber of $\Arr^H$ and fix a point $p\in R$. Since $R$ meets
no hyperplane $K\neq H$ and $\Arr$ is finite, a small enough ball $B$ about $p$
meets no hyperplane of $\Arr$ except $H$. Within $B$ the hyperplane $H$ has two
sides, and each side avoids every hyperplane of $\Arr$, so each lies in a single
chamber. Call these two chambers $C^+$ and $C^-$, distinct because $H$ separates
them. They do not depend on the choice of $p$, since $R$ is connected. Because
$C^+$ is a convex cone lying in the closed halfspace bounded by $H$, the
intersection $\overline{C^+}\cap H$ is a single face of $\overline{C^+}$; it has
dimension $\dim H$ since it contains the relatively open set $R$, so it is a facet
of $\overline{C^+}$, and likewise $\overline{C^-}\cap H$ is a facet of
$\overline{C^-}$. The relative interior of each of these facets avoids every
hyperplane $K\neq H$, by the argument of Lemma~\ref{lem:facet-relint}, so it is a
single chamber of $\Arr^H$; containing $R$ and contained in $\overline R$, it
equals $R$, and the facet equals $\overline R$. This gives the two incidences
$(C^+,\overline R)$ and $(C^-,\overline R)$. Conversely, the relative interior of
any facet supported by $H$ avoids every $K\neq H$ and so is a chamber of
$\Arr^H$, recovering some $R$.  Thus the
chamber-facet incidences supported by $H$ number $2\,r(\Arr^H)$.  Summing over all $H\in\Arr$ gives
\[
\sum_{C \in \Ch(\Arr)} \#\{\text{facets of } \overline C\}
        =2\sum_{H\in\Arr} r(\Arr^H).
\]

For the Boolean threshold arrangement, Theorem~\ref{thm:facet} identifies
$\sigma_n(f)$ with the number of facets of the chamber closure
$\overline{C_f}$, and the chamber-function correspondence identifies
$\Hn$ with $\Ch(\Arr_n)$. Substituting these identifications into the
arrangement-theoretic identity gives the stated threshold-function formula.
\end{proof}

\subsection{\texorpdfstring{Threshold functions and the resonance arrangement}{Threshold functions and the resonance arrangement}}

The restrictions of the Boolean threshold arrangement that occur in this
formula are themselves finite, central and full rank as arrangements inside the
supporting hyperplanes. Finiteness and centrality are immediate from the
construction. For full rank, fix $x\in\B^n$. If $y\neq x$, then
$H_y\neq H_x$ by Lemma~\ref{lem:Hx-injective}, so $H_x\cap H_y$ is a proper
hyperplane inside $H_x$. The common intersection of the restricted hyperplanes
inside $H_x$ is
\[
        \bigcap_{y\neq x}(H_x\cap H_y)
        =H_x\cap\bigcap_{y\neq x}H_y
        =\bigcap_{y\in\B^n}H_y
        =\{0\},
\]
where the third expression simply adds the missing $H_x$ term. The final
intersection is the set of all $a\in\R^{n+1}$ orthogonal to every lifted cube
point $\widetilde y$. Since these lifted points span $\R^{n+1}$, this
intersection is $\{0\}$. Hence $\Arr_n^{H_x}$ is full rank inside $H_x$.

We now identify $\Arr_n^{H_x}$ for the Boolean threshold arrangement. Gutekunst, M\'esz\'aros, and Petersen identify the
resonance arrangement with the intersection of the threshold arrangement and one
of its hyperplanes \cite[Section~5; see also Figure~15]{GMP2021}. The following
records this in the present coordinates.

\begin{lemma}\label{lem:restriction}
Under the identification $H_0=\{a_0=0\}\cong\R^n$, where
$0=(0,\ldots,0)\in\B^n$, the restricted arrangement $\Arr_n^{H_0}$ is the
central arrangement $\Res$ on $\R^n$ consisting of the $2^n - 1$ hyperplanes
$\{b \in \R^n : \langle b, x \rangle = 0\}$ for
$x \in \B^n \setminus \{0\}$.
\end{lemma}

\begin{proof}
The hyperplane $H_0 \subseteq \R^{n+1}$ is $\{a:a_0=0\}$, naturally identified
with $\R^n$ via the last $n$ coordinates: $a=(0,b)$ with $b\in\R^n$. For
$x\in\B^n$,
\[
        H_0\cap H_x
        =\{(0,b):\langle (0,b),\widetilde x\rangle=0\}.
\]
Since
\[
        \langle (0,b),\widetilde x\rangle
        =\sum_{i=1}^n b_i x_i
        =\langle b,x\rangle,
\]
the intersections with $H_x$, for $x\neq 0$, give exactly the hyperplanes
\[
        \{b\in\R^n:\langle b,x\rangle=0\}.
\]
The remaining case $x=0$ corresponds to the member $H_0$ itself, which is
omitted in the definition of the restricted arrangement. Hence $\Arr_n^{H_0}$
is precisely $\Res$.
\end{proof}

\begin{corollary}[Resonance-arrangement formula]\label{cor:exact}
For the Boolean threshold arrangement $\Arr_n$, every restriction
$\Arr_n^{H_x}$, $x\in\B^n$, has $r(\Arr_n^{H_x}) = r(\Res)$. Consequently,
\[
\overline\sigma_n \;=\; \frac{2 \cdot 2^n \cdot r(\Res)}{|\Hn|} \;=\; \frac{2^{n+1}\, r(\Res)}{|\Hn|}.
\]
\end{corollary}

\begin{proof}
We use an invertible linear change of parameter coordinates in order to make
the symmetry of the threshold arrangement explicit. Under the substitution
$x_i=(s_i+1)/2$, with $s_i\in\{-1,+1\}$, the threshold inequality
$a_0+\sum_{i=1}^n a_i x_i\ge 0$ becomes
\[
        a_0+\tfrac12\sum_{i=1}^n a_i
        +\tfrac12\sum_{i=1}^n a_i s_i\ge 0; 
\]
that is,
\[
        b_0+\sum_{i=1}^n b_i s_i\ge 0,
        \qquad
        b_0=a_0+\tfrac12\sum_{i=1}^n a_i,
        \quad
        b_i=\tfrac12 a_i.
\]
The map $(a_0,a_1,\ldots,a_n)\mapsto (b_0,b_1,\ldots,b_n)$ is an invertible
linear map on $\R^{n+1}$. It therefore preserves chamber counts and carries
restricted arrangements to restricted arrangements.

In these coordinates the threshold arrangement is the central arrangement of
$2^n$ hyperplanes
\[
        K_s=\{b\in\R^{n+1}:\langle b,(1,s)\rangle=0\},\qquad
        s\in\{-1,+1\}^n.
\]
Given any two of these hyperplanes $K_s$ and $K_u$, the diagonal sign-change
\[
        \tau_{s,u}(b_0,b_1,\ldots,b_n)=(b_0,\,s_1u_1b_1,\ldots,s_nu_nb_n),
\]
which fixes $b_0$ and flips the sign of $b_i$ precisely when $s_i\ne u_i$,
preserves the whole arrangement and maps $K_s$ to $K_u$, since
\[
        \langle\tau_{s,u}(b),(1,u)\rangle
        =b_0+\sum_{i=1}^n s_iu_i^2b_i
        =\langle b,(1,s)\rangle .
\]
These maps belong to the signed-permutation group, the invertible linear maps
that permute the $n+1$ coordinates and independently change their signs. This
group therefore acts transitively on the hyperplanes $K_s$, meaning that any one
of them can be sent to any other by an arrangement-preserving linear
automorphism. The restriction of $\tau_{s,u}$ to $K_s$ carries the corresponding
restricted arrangement onto that of $K_u$, so the two restricted arrangements
have the same number of chambers.

Lemma \ref{lem:restriction} identifies one such restriction, namely
$\Arr_n^{H_0}$, with $\Res$. Hence $r(\Arr_n^{H_x})=r(\Res)$ for every
$x\in\B^n$. Applying Lemma \ref{lem:facet-count} and dividing by
$|\Hn|=|\Ch(\Arr_n)|$ gives the displayed formula for $\overline\sigma_n$.
\end{proof}

The arrangement $\Res$ is the \emph{resonance arrangement} (also called the all-subsets arrangement, the adjoint braid arrangement, and the arrangement of maximal unbalanced families), a well-studied object in algebraic combinatorics: see Billera \emph{et al.}\ \cite{BTMDW2012} for the connection with maximal unbalanced families, Kühne \cite{Kuhne2023} for Betti number computations, and Gutekunst, Mészáros and Petersen \cite{GMP2021} for the connection with threshold functions. The number of chambers $r(\Res)$ is sequence A034997 in the On-Line Encyclopedia
of Integer Sequences (OEIS) \cite{OEIS}; we use these tabulated values in
Section~\ref{sec:growth}.  The term ``resonance-arrangement formula'' in Corollary \ref{cor:exact} is descriptive rather than a standard named formula: the formula is called this here because the restriction of the Boolean threshold arrangement to any one of its supporting hyperplanes is isomorphic to $\Res$.

\begin{remark}
The exact formula of Corollary~\ref{cor:exact} can also be obtained from the
analysis of the threshold arrangement by Gutekunst, Mészáros, and Petersen, in
particular \cite[Corollary~8 and Remark~6]{GMP2021}. The chamber--facet incidence
identity behind Lemma~\ref{lem:facet-count} is standard; here it acquires
specification-number content, since the average number of facets of a threshold
chamber is $\overline\sigma_n$. In these terms their Theorem~1,
$(n+1)|\Hn|/2^{n+1} < r(\Res) < |\Hn|/2$, is $n+1 < \overline\sigma_n < 2^n$ for
$n\ge2$. A theorem of Fukuda, Tamura, and Tokuyama sharpens the upper bound to $\overline\sigma_n\le 2n$, equivalently $r(\Res)\le 2n\,|\Hn|/2^{n+1}=n\,|\Hn|/2^{n}$, which answers \cite[Problem~3]{GMP2021}; their quantity, the average number of walls per threshold chamber, is the average facet count, which Theorem~\ref{thm:facet} identifies with $\overline\sigma_n$. The Doliwa--Fan--Simon--Zilles bound of Theorem~\ref{thm:vc-average} gives
the same order for every essentially specified class. Both are developed in
Section~\ref{sec:growth}.
\end{remark}

\subsection{Degree-\texorpdfstring{$d$}{d} resonance arrangements}\label{ssec:degree-d}

We next give a degree-$d$ analogue of Corollary~\ref{cor:exact}. No new
chamber-facet argument is needed: we apply Lemma~\ref{lem:facet-count} to the
monomial-feature arrangement from Proposition~\ref{prop:ptf-average}. The only
point is to identify the common restriction of that arrangement, which gives the
following degree-$d$ version of the resonance arrangement.

\begin{definition}[Degree-$d$ resonance arrangement]\label{def:degree-d-resonance}
For $1\le d\le n$, put
\[
        \mathcal M_{n,d}^{+}=\{S\subseteq[n]:1\le |S|\le d\}.
\]
The degree-$d$ resonance arrangement $\mathcal R_{n,d}$ is the central
arrangement in $\R^{\mathcal M_{n,d}^{+}}$ with one hyperplane for each
non-empty $T\subseteq[n]$:
\[
        R_T^{(d)}=
        \left\{b:\sum_{\substack{S\subseteq T\\1\le |S|\le d}} b_S=0\right\}.
\]
For $d=1$ this is the usual resonance arrangement.
\end{definition}

\begin{proposition}[Degree-$d$ exact formula]\label{prop:ptf-exact}
Let $\C_d$ be the degree-$d$ polynomial threshold class from
Proposition~\ref{prop:ptf-average}. Then
\[
        \overline\sigma_{\C_d}
        =
        \frac{2^{n+1}r(\mathcal R_{n,d})}{|\C_d|}.
\]
\end{proposition}

\begin{proof}
Let $\Arr_{n,d}$ be the degree-$d$ monomial-feature arrangement from the proof
of Proposition~\ref{prop:ptf-average}. Proposition~\ref{prop:feature-arrangements}
identifies $\sigma_{\C_d}(f)$ with the number of facets of the corresponding
chamber closure. Lemma~\ref{lem:facet-count}, applied to $\Arr_{n,d}$, gives
\[
        \sum_{f\in\C_d}\sigma_{\C_d}(f)
        =2\sum_{x\in\B^n} r\bigl(\Arr_{n,d}^{H_x^{(d)}}\bigr).
\]

We first note that the restrictions have the same number of chambers. Given two
cube points $x,y\in\B^n$, choose a change of variables on the cube obtained by
permuting coordinates and replacing selected variables $t_i$ by $1-t_i$, so that
$x$ is sent to $y$. If $p(t)$ is a polynomial of degree at most $d$, then $p$
after this change of variables is again a polynomial of degree at most $d$, and
the induced map on coefficient vectors is invertible. Moreover, evaluation at
$x$ before the change of variables is evaluation at $y$ after the change. Thus
this invertible linear change of coefficient coordinates carries $H_x^{(d)}$ to
$H_y^{(d)}$, and carries the whole arrangement $\Arr_{n,d}$ to itself. It
follows that the restricted arrangements $\Arr_{n,d}^{H_x^{(d)}}$ and
$\Arr_{n,d}^{H_y^{(d)}}$ are linearly isomorphic, so they have the same number
of chambers.

We now identify one restriction. For $x=0$, the hyperplane
$H_0^{(d)}$ is $\{a_\varnothing=0\}$. Identifying it with the coefficient space
with coordinates $b_S=a_S$, $1\le |S|\le d$, the intersection with
$H_{1_T}^{(d)}$, where $T\ne\varnothing$, comes from setting $a_\varnothing=0$ in
$\sum_{|S|\le d}a_S(1_T)_S=0$. Since the monomial $(1_T)_S$ equals $1$ when
$S\subseteq T$ and $0$ otherwise, this is
\[
        \sum_{\substack{S\subseteq T\\1\le |S|\le d}} b_S=0.
\]
Thus $\Arr_{n,d}^{H_0^{(d)}}$ is $\mathcal R_{n,d}$. Substituting this into the
incidence formula and dividing by $|\C_d|$ proves the result.
\end{proof}

Combining the exact formula with the two-sided bound $N_d\le\overline\sigma_{\C_d}\le 2N_d-2$ of Proposition~\ref{prop:ptf-average} gives the two-sided estimate for the
degree-$d$ region count:
\[
        \frac{N_d|\C_d|}{2^{n+1}}
        \le
        r(\mathcal R_{n,d})
        \le
        \frac{(2N_d-2)|\C_d|}{2^{n+1}} .
\]
For $d=1$ this recovers the resonance comparison of
Section~\ref{sec:growth}; there the lower inequality is strict, by the witness in
Theorem~\ref{thm:theta-n}, whereas the general degree-$d$ lower bound is
non-strict.

\begin{example}[$n=3$, $d=2$]\label{ex:n3d2}
The smallest informative case is $n=3$, $d=2$. Here $\mathcal R_{3,2}$ is a
central arrangement of $7$ hyperplanes in $\R^{6}$, with coordinates $b_S$
indexed by $S\in\mathcal M_{3,2}^{+}=\{\{1\},\{2\},\{3\},\{1,2\},\{1,3\},\{2,3\}\}$
and one hyperplane for each non-empty $T\subseteq[3]$. Order the columns as
$(b_{\{1\}},b_{\{2\}},b_{\{3\}},b_{\{1,2\}},b_{\{1,3\}},b_{\{2,3\}})$ and the rows
by
\[
        T=\{1\},\{2\},\{3\},\{1,2\},\{1,3\},\{2,3\},\{1,2,3\}.
\]
The seven normals are the rows of
\[
        \begin{pmatrix}
        1&0&0&0&0&0\\
        0&1&0&0&0&0\\
        0&0&1&0&0&0\\
        1&1&0&1&0&0\\
        1&0&1&0&1&0\\
        0&1&1&0&0&1\\
        1&1&1&1&1&1
        \end{pmatrix},
\]
the $0/1$ vectors with a $1$ in coordinate $S$ exactly when $S\subseteq T$. Any
six of the seven normals are linearly independent, so the arrangement is in
general position. Hence
$r(\mathcal R_{3,2})=2\sum_{i=0}^{5}\binom{6}{i}=126$. The degree-$2$ polynomial
threshold class on three variables consists of all Boolean functions on $\B^3$
except parity $x_1\oplus x_2\oplus x_3$ and its complement. 
Indeed, Wang and Williams \cite[Corollary~2.2 and Theorem~3.4]{WangWilliams1991}
prove that for every $n$ the only Boolean functions on $n$ variables of threshold
order $n$ are the two parity functions, every other function having threshold
order at most $n-1$. Here the threshold order of $f$ is the least degree of a
polynomial separator for $f$, so threshold order at most $d$ is membership in
$\C_d$. For $n=3$ this leaves parity and its complement as the only functions of
threshold order $3$, so every other function on $\B^3$ lies in $\C_2$, and
\[
        |\C_2|=2^8-2=254 .
\]
 Proposition
\ref{prop:ptf-exact} gives
\[
        \overline\sigma_{\C_2}
        =\frac{2^{4}\cdot 126}{254}
        =\frac{1008}{127}
        \approx 7.937 .
\]
This is consistent with the degree-$2$ lower bound $\sigma_{\C_2}(f)\ge N_2=7$
of Proposition~\ref{prop:ptf-average}, lies in the interval $[N_2,2N_2-2]=[7,12]$,
and lies slightly below the full-class value $8$, as expected since almost every
Boolean function on three variables is a degree-$2$ polynomial threshold
function. The region-count comparison is $111.125\le r(\mathcal R_{3,2})<190.5$, which the exact value $126$ satisfies.
\end{example}

We record the numerical consequences of Corollary \ref{cor:exact} in Section \ref{sec:growth}.

\section{The growth rate of the average specification number}\label{sec:growth}

Corollary~\ref{cor:lower} gives a lower bound on $\overline\sigma_n$, and Theorem~\ref{thm:ftt-average}
gives an upper bound. Together they
settle the order of growth of $\overline\sigma_n$.

\begin{theorem}[Linear order of the average specification number]\label{thm:theta-n}
$\overline\sigma_n=\Theta(n)$. More precisely, for every $n\ge1$,
\[
        n+1\le\overline\sigma_n\le 2n,
\]
and for $n\ge2$ the lower inequality is strict.
\end{theorem}

\begin{proof}
We obtain the lower bound by averaging the universal bound $\sigma_n(f)\ge n+1$
(Corollary~\ref{cor:lower}). For the upper bound, Theorem~\ref{thm:ftt-average} gives $\overline\sigma_n\le 2n$. Together these give $\overline\sigma_n=\Theta(n)$.
For $n\ge2$ the lower bound is strict: flipping the constant-zero function at any
single vertex yields a single-point indicator, which is threshold, so every vertex
is essential and $\sigma_n(0)=2^n>n+1$; hence the average exceeds $n+1$. For $n=1$, $\overline\sigma_1=2=n+1=2n$.
\end{proof}

Gutekunst, Mészáros and Petersen point out that their companion bound
$r(\Res)<|\Hn|/2$ would improve given a bound on the average number of facets of a
threshold chamber \cite[Remark~6]{GMP2021}. They write that average as $w(\cdot)$,
indexed by the rank of the arrangement. The Boolean threshold arrangement for $n$
variables has rank $n+1$, so their average facet count for it is $w(n+1)$, and by
Corollary~\ref{cor:exact} this equals $\overline\sigma_n$:
\[
        w(n+1)
        =\frac{1}{|\Hn|}\sum_{f\in\Hn}\#\{\text{facets of }\overline{C_f}\}
        =\frac{1}{|\Hn|}\sum_{f\in\Hn}\sigma_n(f)
        =\overline\sigma_n .
\]
The exact identity is, after reindexing, the chamber--facet incidence relation
underlying \cite[Corollary~8 and Remark~6]{GMP2021}; what is new here is its
interpretation as the average specification number through Theorem~\ref{thm:facet}.
Combining it with Theorem~\ref{thm:theta-n} gives, for $n\ge2$, the two-sided
comparison
\[
        \frac{(n+1)|\Hn|}{2^{n+1}}
        <
        r(\Res)
        <
        \frac{2n\,|\Hn|}{2^{n+1}}=\frac{n\,|\Hn|}{2^{n}},
\]
whose lower bound is the strict inequality $\overline\sigma_n>n+1$ of
Theorem~\ref{thm:theta-n} and whose upper bound is the theorem of Fukuda, Tamura, and Tokuyama \cite{Fukuda1993}, both transferred through Corollary~\ref{cor:exact}. The same
strict lower bound appears as the resonance inequality of \cite{GMP2021}. As an estimate of $r(\Res)$ this improves on $r(\Res)<|\Hn|/2$ for $n\ge3$, and coincides with it at $n=2$.
Writing $\Delta_n=\log_2|\Hn|-\log_2 r(\Res)$,
\[
        n-\log_2 n<\Delta_n<n+1-\log_2(n+1),
\]
so $\Delta_n=n-\log_2 n+O(1)$; the interval has width $\log_2(2n)-\log_2(n+1)\to1$, reflecting that $n+1\le\overline\sigma_n\le 2n$ fixes $\overline\sigma_n$ to within a factor of two.

Feeding the Kahn, Koml\'os and Szemer\'edi estimate
$\log_2|\Hn|=n^2-n\log_2 n+O(n)$ (Remark~\ref{rem:Hn-asymptotics}) into the
comparison gives
\[
        \log_2 r(\Res)=n^2-n\log_2 n+O(n).
\]
This logarithmic asymptotic already follows by combining the two-sided resonance
estimate of \cite{GMP2021} with the same threshold-function estimate; the sharper conclusion here is the multiplicative upper bound $r(\Res)\le 2n\,|\Hn|/2^{n+1}=n\,|\Hn|/2^{n}$.

Once this fixes the order of growth, the remaining question is the value of the constant.

\begin{problem}\label{prob:constant}
Does $\lim_{n\to\infty}\overline\sigma_n/(n+1)$ exist, and if so, what is its
value? If the limit exists it lies in $[1,2]$.
\end{problem}

We obtain the values in Table~\ref{tab:numerics} from
Corollary~\ref{cor:exact}, using OEIS A000609 for $|\Hn|$ and OEIS A034997 for
$r(\Res)$ \cite{OEIS}. The largest entry, $r(\Res)$ at $n=9$, is the value
posted to OEIS by Chroman and reported by Brysiewicz, Eble, and K\"uhne
\cite{BEK2023}; as of the access date recorded in \cite{OEIS} it is the largest exact value tabulated there, and it is the largest resonance-arrangement computation reported in \cite{BEK2023}.

\begin{table}[htbp]
\centering
\resizebox{\textwidth}{!}{%
\begin{tabular}{rrrrc}
\toprule
$n$ & $|\Hn|$ & $r(\Res)$ & $\overline\sigma_n$ (decimal) & $\overline\sigma_n/(n+1)$ \\
\midrule
1 & 4 & 2 & $2.0000$ & $1.0000$ \\
2 & 14 & 6 & $3.4286$ & $1.1429$ \\
3 & 104 & 32 & $4.9231$ & $1.2308$ \\
4 & 1{,}882 & 370 & $6.2912$ & $1.2582$ \\
5 & 94{,}572 & 11{,}292 & $7.6417$ & $1.2736$ \\
6 & 15{,}028{,}134 & 1{,}066{,}044 & $9.0799$ & $1.2971$ \\
7 & 8{,}378{,}070{,}864 & 347{,}326{,}352 & $10.6129$ & $1.3266$ \\
8 & 17{,}561{,}539{,}552{,}946 & 419{,}172{,}756{,}930 & $12.2208$ & $1.3579$ \\
9 & 144{,}130{,}531{,}453{,}121{,}108 & 1{,}955{,}230{,}985{,}997{,}140 & $13.8913$ & $1.3891$ \\
\bottomrule
\end{tabular}}
\caption{Average specification numbers for small $n$, computed from
$\overline\sigma_n = 2^{n+1} r(\Res)/|\Hn|$.}
\label{tab:numerics}
\end{table}

These values are illustrative. By Theorem~\ref{thm:theta-n} the order of
$\overline\sigma_n$ is linear, so the table bears not on the order but on the
constant, which it illustrates rather than determines: each ratio $\overline\sigma_n/(n+1)$ lies in $[1,2]$, equal to $1$
only at $n=1$, and over $n\le 9$ the ratios increase monotonically from
$1.0000$ to $1.3891$. From $n=5$ onward the increments
$\overline\sigma_n-\overline\sigma_{n-1}$ increase, from $1.35$ at $n=5$ to
$1.67$ at $n=9$, remaining below $2$.

The generic comparison of Section~\ref{sec:average} is suggestive here. A short computation from Cover's region count \cite{Cover1965} shows that a generic central arrangement of $N\to\infty$ hyperplanes in $\R^{n+1}$ has, on average, $2n$ facets per chamber, a ratio $2n/(n+1)\to2$. By the spherical-arrangement theorem of Fukuda, Tamura, and Tokuyama \cite[Corollary~1.3]{Fukuda1993}, the value $2n$ bounds the average facet count of every spherical arrangement in this dimension; a simple arrangement approaches this value as the number of hyperplanes grows \cite[Corollary~3.6 and Remark~3.7]{Roudneff1991}. For $n\ge2$, the Boolean threshold arrangement is far from generic and lies below this ceiling, so the comparison does not determine the constant, though it places it at the upper end of the interval $[n+1,2n]$.

The range is too short for us to decide Problem~\ref{prob:constant}. The monotone
increase of $\overline\sigma_n/(n+1)$ over $n\le 9$ is consistent both with an
interior limiting constant and with slow convergence to the generic upper
endpoint $2$. No extrapolation from these values is reliable. One reformulation
may help: writing
$\delta_n=\overline\sigma_n-\overline\sigma_{n-1}=2^{n+1}r(\Res)/|\Hn|
-2^{n}r(\mathcal R_{n-1})/|\mathcal H_{n-1}|$, the Stolz--Ces\`aro theorem gives
that $\delta_n\to c$ implies $\overline\sigma_n/(n+1)\to c$, so it would suffice,
though it is not necessary, to settle the convergence of the single sequence
$\delta_n$.

\FloatBarrier

\section{Chow vectors and the threshold zonotope}\label{sec:zonotope}

The chamber-facet identification, the one-inclusion graph, the resonance edge
count, and the Chow parameters of a threshold function all appear in the
combinatorics of a single polyhedral object, the threshold zonotope. Recall that
$X=\R^{n+1}$ is the lifted input space, with
$\widetilde x=(1,x_1,\ldots,x_n)$. A parameter vector
$a=(a_0,\ldots,a_n)$ evaluates a lifted point by
\[
        \langle a,\widetilde x\rangle
        =a_0+a_1x_1+\cdots+a_nx_n .
\]

\begin{definition}[Threshold zonotope]\label{def:threshold-zonotope}
A \emph{zonotope} is a Minkowski sum of finitely many line segments. The
\emph{threshold zonotope} is the zonotope
\[
        Z_n=\sum_{x\in\B^n}[-\widetilde x,\widetilde x]\subseteq X,
\]
generated by the symmetrised lifted cube points. Equivalently,
\[
        Z_n=\left\{\sum_{x\in\B^n}t_x\widetilde x:-1\le t_x\le 1\right\}.
\]
For a convex body $Z\subseteq X$, its \emph{support function} is
\[
        h_Z(a)=\max_{z\in Z}\langle a,z\rangle,\qquad a\in A.
\]
Support functions are additive under Minkowski sums, and a segment centred
at the origin has $h_{[-v,v]}(a)=|\langle a,v\rangle|$, so
\[
        h_{Z_n}(a)=\sum_{x\in\B^n}|\langle a,\widetilde x\rangle| .
\]
\end{definition}

The classical Chow parameters of a Boolean function $f$ are the number of true
points and, for each variable $x_i$, the number of true points with $x_i=1$.
For any Boolean function $f:\B^n\to\B$, the \emph{modified Chow vector} is
\[
        \chi(f)=\sum_{x\in\B^n}(2f(x)-1)\,\widetilde x\in X .
\]
Its coordinates are an invertible affine image of the classical Chow parameters
of $f$ \cite{Chow1961,Winder1971}:
\[
        \chi_0(f)=2|f^{-1}(1)|-2^n,
        \qquad
        \chi_i(f)=2\bigl|\{x\in\B^n:f(x)=1,\ x_i=1\}\bigr|-2^{n-1}
        \quad(1\le i\le n).
\]

Stenson \cite{Stenson2014} uses the vertex map $f\mapsto\sum_{x\in\B^n}(2f(x)-1)(-1,x)$.
Up to the sign of the zeroth coordinate (since Stenson lifts with $(-1,x)$, whereas here we use 
$\widetilde x=(1,x)$), this is the map $\chi$ defined above. Her
signed critical-instance coordinates $\gamma_i$ (Proposition~2.6 of \cite{Stenson2014}) are a
shift of the modified Chow coordinates used here, $\gamma_i=\chi_i(f)-\tfrac12\chi_0(f)$.

For example, when $n=2$, the four generators $\widetilde x=(1,x)$, $x\in\B^2$, are
\[
        (1,0,0),\quad (1,1,0),\quad (1,0,1),\quad (1,1,1).
\]
It turns out that, up to linear equivalence, the resulting three-dimensional
zonotope is a rhombic dodecahedron, with $12$ parallelogram faces, $24$ edges,
and $14$ vertices. These
vertices are the modified Chow vectors of the $14$ two-variable threshold
functions. To make this concrete, write $\widetilde{ij}$ for
$\widetilde{(i,j)}$. For the conjunction $x_1\wedge x_2$,
\[
        \chi(x_1\wedge x_2)
        =
        -\widetilde{00}-\widetilde{10}-\widetilde{01}+\widetilde{11}
        =
        (-2,0,0),
\]
while for the single-literal function $x_1$,
\[
        \chi(x_1)
        =
        -\widetilde{00}+\widetilde{10}-\widetilde{01}+\widetilde{11}
        =
        (0,2,0).
\]
These are two of the fourteen vertices. The two non-threshold functions, XOR
and its complement, both give the central Chow point $(0,0,0)$ rather than
vertices. That a non-threshold Boolean function never gives a vertex is
Stenson's Corollary~2.7 \cite{Stenson2014}. This is the $n=2$ instance of the classical Chow-parameter corner
picture discussed by Winder \cite[Section~2A and Theorem~7]{Winder1971}, now
drawn with its zonotope edge structure.

For the following theorem, recall that a \emph{fan} is a collection of
polyhedral cones fitting together along common faces. A central hyperplane
arrangement induces such a fan by taking the closures of its chambers, together
with all their faces. For a vertex $v$ of a polytope $Z\subseteq X$, its normal
cone is the set of directions $a\in A$ for which $v$ maximises the linear
functional $z\mapsto\langle a,z\rangle$ over $Z$:
\[
        N_Z(v)=\{a\in A:\langle a,v\rangle
        =\max_{z\in Z}\langle a,z\rangle\}.
\]
The \emph{normal fan} of $Z$ is the fan whose full-dimensional cones are these
vertex normal cones, together with all their faces.

\begin{theorem}[Vertices, edges, and degrees of the threshold zonotope]\label{thm:zonotope}
Let $Z_n$ be the threshold zonotope. Then:
\begin{enumerate}[label=\textup{(\alph*)}]
\item The normal cones at the vertices of $Z_n$ are exactly the chamber closures
$\overline{C_f}$, $f\in\Hn$, of the Boolean threshold arrangement.
Equivalently, the normal fan of $Z_n$ is the fan induced by $\Arr_n$.
\item The vertex set of $Z_n$ is
\[
        \operatorname{vert}(Z_n)=\{\chi(f):f\in\Hn\}.
\]
The vertex whose normal cone is $\overline{C_f}$ is $\chi(f)$. Equivalently, on
the open chamber $C_f$ the support function is the linear function
$h_{Z_n}(a)=\langle a,\chi(f)\rangle$.
\item Two vertices $\chi(f)$ and $\chi(g)$ are joined by an edge of $Z_n$ if
and only if $g=f\oplus\mathbf 1_{\{x\}}$ for some $x\in\Ess(f)$; the edge is
parallel to $\widetilde x$, and
\[
        \chi(f\oplus\mathbf 1_{\{x\}})=\chi(f)-2\bigl(2f(x)-1\bigr)\widetilde x .
\]
Consequently the one-inclusion graph $G(\Hn)$ is the $1$-skeleton of $Z_n$, with each $f$ placed at $\chi(f)$, and
\[
        \sigma_n(f)=\deg_{Z_n}\bigl(\chi(f)\bigr).
\]
\end{enumerate}
\end{theorem}

\begin{proof}
For (a), we use the standard normal-fan theorem for zonotopes: if
\[
        Z=\sum_{i=1}^N[-v_i,v_i],
\]
then the normal fan of $Z$ is the fan whose maximal cones are the chambers of the
central arrangement $\{v_i^\perp:1\le i\le N\}$, duplicate hyperplanes ignored 
\cite[\S7.3]{Ziegler1995}. Applied to the generators $\widetilde x$,
$x\in\B^n$, this gives exactly $\Arr_n$, since $H_x=\widetilde x^\perp$ and the
$\widetilde x$ are distinct. Hence the full-dimensional normal cones of $Z_n$ are
the chamber closures of $\Arr_n$.

The support function gives the same description. By additivity over the Minkowski
sum, $h_{Z_n}(a)=\sum_{x\in\B^n}\lvert\langle a,\widetilde x\rangle\rvert$. On an
open chamber $C$ of $\Arr_n$ the chamber signs $\eps_C(x)\in\{-1,1\}$,
characterised by $\lvert\langle a,\widetilde x\rangle\rvert=\eps_C(x)\langle a,\widetilde x\rangle$,
are constant, so
\[
        h_{Z_n}(a)=\sum_{x\in\B^n}\eps_C(x)\langle a,\widetilde x\rangle
        =\Bigl\langle a,\ \sum_{x\in\B^n}\eps_C(x)\widetilde x\Bigr\rangle ,
\]
and $h_{Z_n}$ is linear on $C$, extending to $\overline C$ by continuity. Each
chamber is therefore a cone of linearity of the support function, in agreement
with the normal-fan description.

For (b), fix $a$ in the open chamber $C_f$. The linear functional
$\langle a,\cdot\rangle$ is maximised over
$Z_n=\{\sum_x t_x\widetilde x:-1\le t_x\le 1\}$ by choosing
$t_x=1$ when $\langle a,\widetilde x\rangle>0$ and $t_x=-1$ when
$\langle a,\widetilde x\rangle<0$. On $C_f$, this choice is
$t_x=2f(x)-1$, so the unique maximising point is
\[
        \sum_{x\in\B^n}(2f(x)-1)\widetilde x=\chi(f).
\]
Thus every threshold function $f$ gives a vertex $\chi(f)$ of $Z_n$, and
every $a\in C_f$ maximises there, so $C_f\subseteq N_{Z_n}(\chi(f))$. By (a)
this normal cone is one of the chamber closures, and the only chamber
closure containing the open chamber $C_f$ is $\overline{C_f}$, since distinct
open chambers are disjoint. Hence $N_{Z_n}(\chi(f))=\overline{C_f}$, so the
vertex with normal cone $\overline{C_f}$ is $\chi(f)$.

The lifted cube points span $X$, so $Z_n$ is full-dimensional and every
vertex has a normal cone of full dimension. By (a) that cone is a chamber
closure $\overline{C_f}$, and the previous paragraph identifies the
corresponding vertex as $\chi(f)$. So every vertex equals $\chi(f)$ for some
$f\in\Hn$, giving
\[
        \operatorname{vert}(Z_n)=\{\chi(f):f\in\Hn\}.
\]
The coordinate formulas follow by evaluating
$\chi(f)=\sum_x(2f(x)-1)\widetilde x$ in the coordinates of $X$, using
$|\{x:x_i=1\}|=2^{n-1}$.

For (c), we determine the edges of $Z_n$ from its normal cones. For a direction
$a$, the points of $Z_n$ maximising $\langle a,\cdot\rangle$ form a face, the
convex hull of the vertices attaining the maximum. When $a$ is in the
interior of a normal cone $\overline{C_f}$, the unique maximiser is the
single vertex $\chi(f)$. When $a$ is in the relative interior of a facet
shared by two such cones, exactly two vertices attain the maximum, so the
maximising face is the segment joining them, a one-dimensional face, that is,
an edge of $Z_n$. Conversely, the normal cone of an edge is the
codimension-one cone common to the normal cones of its two endpoints. Hence
two vertices are joined by an edge if and only if their normal cones share a
facet. By Lemma~\ref{lem:one-sign-crossing}, this
means crossing a single supporting hyperplane, namely $H_x$, so
$g=f\oplus\mathbf 1_{\{x\}}$ with $H_x$ supporting a facet of
$\overline{C_f}$, equivalently $x\in\Ess(f)$ by
Theorem~\ref{thm:facet}. The two Chow vectors then satisfy
\[
        \chi(f\oplus\mathbf 1_{\{x\}})-\chi(f)
        =
        -2(2f(x)-1)\widetilde x,
\]
so the joining edge is parallel to $\widetilde x$. Changing the function at the
single input $x$ changes the Chow vector only in the lifted direction
$\widetilde x$. The number of edges at $\chi(f)$ equals the number of facets of
$\overline{C_f}$, which is $|\Ess(f)|=\sigma_n(f)$ by
Theorem~\ref{thm:facet}.

\end{proof}

We view the vertex statement in part (b) as the zonotope form of the classical corner
theorem stated by Winder. To state that theorem explicitly, form the finite set
of Chow vectors $\chi(g)$ over all Boolean functions $g:\B^n\to\B$, and take
its convex hull. By the definition of $Z_n$, this convex hull is precisely
$Z_n$. Winder's corner theorem says that the vertices of this convex hull are
precisely the vectors $\chi(f)$ with $f\in\Hn$
\cite[Theorem~7]{Winder1971}. Winder also treats the solid hull and its
supporting hyperplanes. The normal cone of a vertex of a polytope is the set of linear functionals
maximised at that vertex, and its interior is the set maximised there uniquely.
By part~(a) the normal cone at $\chi(f)$ is $\overline{C_f}$, so its interior is
the open chamber $C_f$. By the definition of the chamber, $C_f$ is the set of
weight vectors that strictly represent $f$, those $a$ with
$\langle a,\widetilde x\rangle>0$ at every $x$ where $f(x)=1$ and
$\langle a,\widetilde x\rangle<0$ at every $x$ where $f(x)=0$ (in particular
$\langle a,\widetilde x\rangle\ne0$ throughout). A strict representation of $f$
is therefore a vector $a$ for which $\chi(f)$ is the unique maximiser of
$\langle a,\cdot\rangle$ on $Z_n$, equivalently one whose supporting hyperplane
meets $Z_n$ only at the corner $\chi(f)$; this is Winder's observation
\cite[Section~2E]{Winder1971}. A weight vector on the boundary of $\overline{C_f}$ does not strictly
represent $f$: there $\langle a,\widetilde x\rangle=0$ for some cube point $x$,
so the generator $[-\widetilde x,\widetilde x]$ of $Z_n$ is orthogonal to $a$
and lies in the face on which $\langle a,\cdot\rangle$ is maximised. The maximum
is then attained on a face of $Z_n$ containing $\chi(f)$, not at $\chi(f)$
alone, so $\chi(f)$ is no longer the unique maximiser. The apex $0$ lies in
every $\overline{C_f}$, since all of $Z_n$ maximises the zero functional, and it
strictly represents no function, as $\langle 0,\widetilde x\rangle=0$ at every
$x$. In zonotope form, parts (a) and (b) are due to
Stenson, who defined this zonotope and identified its vertices with the threshold
functions \cite[Corollary~2.3]{Stenson2014}. Part (c)  records the edge-specification
interpretation used here.

Thus the theorem gives the following correspondence:
\[
\begin{array}{c@{\quad\longleftrightarrow\quad}c}
\text{threshold function } f & \text{vertex } \chi(f),\\
\overline{C_f} & \text{normal cone at } \chi(f),\\
x\in\Ess(f) & \text{incident edge direction } \widetilde x .
\end{array}
\]

For $n=2$, the rhombic dodecahedron $Z_2$ has two vertex types: eight trivalent
vertices and six tetravalent vertices. By Theorem~\ref{thm:zonotope}, these are
exactly the functions with $\sigma_2=3$ and $\sigma_2=4$, respectively, and the
average vertex degree is
\[
        \frac{8\cdot3+6\cdot4}{14}=\frac{24}{7}=\overline\sigma_2 .
\]

\begin{figure}[htbp]
\centering
\resizebox{\textwidth}{!}{%
\begin{tikzpicture}[line cap=round,line join=round]
  \begin{scope}[shift={(0.00,-2.54)},yscale=0.30]
    \shade[shading=radial,inner color=black!18,outer color=white] (0,0) circle (3.06);
  \end{scope}
  \definecolor{fcT0}{rgb}{0.838,0.892,0.973}
  \definecolor{fcB0}{rgb}{0.748,0.802,0.883}
  \definecolor{fcT1}{rgb}{0.675,0.717,0.780}
  \definecolor{fcB1}{rgb}{0.585,0.627,0.690}
  \definecolor{fcT2}{rgb}{0.851,0.906,0.989}
  \definecolor{fcB2}{rgb}{0.761,0.816,0.899}
  \definecolor{fcT3}{rgb}{0.678,0.720,0.784}
  \definecolor{fcB3}{rgb}{0.588,0.630,0.694}
  \definecolor{fcT4}{rgb}{0.841,0.895,0.977}
  \definecolor{fcB4}{rgb}{0.751,0.805,0.887}
  \definecolor{fcT5}{rgb}{0.848,0.903,0.986}
  \definecolor{fcB5}{rgb}{0.758,0.813,0.896}
  \shade[top color=fcT0,bottom color=fcB0] (2.550,0.194) -- (1.684,-0.323) -- (0.367,1.375) -- (1.233,1.893) -- cycle;
  \shade[top color=fcT1,bottom color=fcB1] (2.550,0.194) -- (1.317,-1.699) -- (0.451,-2.216) -- (1.684,-0.323) -- cycle;
  \shade[top color=fcT2,bottom color=fcB2] (-2.550,-0.194) -- (-0.866,-0.517) -- (0.367,1.375) -- (-1.317,1.699) -- cycle;
  \shade[top color=fcT3,bottom color=fcB3] (-2.550,-0.194) -- (-1.233,-1.893) -- (0.451,-2.216) -- (-0.866,-0.517) -- cycle;
  \shade[top color=fcT4,bottom color=fcB4] (0.367,1.375) -- (1.684,-0.323) -- (0.451,-2.216) -- (-0.866,-0.517) -- cycle;
  \shade[top color=fcT5,bottom color=fcB5] (0.367,1.375) -- (1.233,1.893) -- (-0.451,2.216) -- (-1.317,1.699) -- cycle;
  \draw[zonoink,line width=1.0pt] (2.550,0.194) -- (1.317,-1.699);
  \draw[zonoink,line width=1.0pt] (2.550,0.194) -- (1.684,-0.323);
  \draw[zonoink,line width=1.0pt] (2.550,0.194) -- (1.233,1.893);
  \draw[zonoink,line width=1.0pt] (-2.550,-0.194) -- (-1.233,-1.893);
  \draw[zonoink,line width=1.0pt] (-2.550,-0.194) -- (-0.866,-0.517);
  \draw[zonoink,line width=1.0pt] (-2.550,-0.194) -- (-1.317,1.699);
  \draw[zonoink,line width=1.0pt] (0.367,1.375) -- (1.684,-0.323);
  \draw[zonoink,line width=1.0pt] (0.367,1.375) -- (1.233,1.893);
  \draw[zonoink,line width=1.0pt] (0.367,1.375) -- (-0.866,-0.517);
  \draw[zonoink,line width=1.0pt] (0.367,1.375) -- (-1.317,1.699);
  \draw[zonoink,line width=1.0pt] (0.451,-2.216) -- (1.317,-1.699);
  \draw[zonoink,line width=1.0pt] (0.451,-2.216) -- (1.684,-0.323);
  \draw[zonoink,line width=1.0pt] (0.451,-2.216) -- (-1.233,-1.893);
  \draw[zonoink,line width=1.0pt] (0.451,-2.216) -- (-0.866,-0.517);
  \draw[zonoink,line width=1.0pt] (-0.451,2.216) -- (1.233,1.893);
  \draw[zonoink,line width=1.0pt] (-0.451,2.216) -- (-1.317,1.699);
  \draw[gray!55,line width=0.5pt] (3.95,-3.5) -- (3.95,3.7);
  \draw[zonoink,line width=1.3pt] (8.000,1.500) -- (8.000,2.000);
  \draw[zonoink,line width=1.3pt] (8.000,1.500) -- (6.939,0.750);
  \draw[zonoink,line width=1.3pt] (8.000,1.500) -- (9.061,0.750);
  \draw[zonoink,line width=1.3pt] (8.000,1.500) -- (8.000,0.462);
  \draw[zonoink,line width=1.3pt] (8.000,-1.500) -- (8.000,-2.000);
  \draw[zonoink,line width=1.3pt] (8.000,-1.500) -- (6.939,-0.750);
  \draw[zonoink,line width=1.3pt] (8.000,-1.500) -- (9.061,-0.750);
  \draw[zonoink,line width=1.3pt] (8.000,-1.500) -- (8.000,-0.462);
  \draw[zonoink,line width=1.3pt] (5.172,0.000) -- (8.000,2.000);
  \draw[zonoink,line width=1.3pt] (5.172,0.000) -- (6.939,0.750);
  \draw[zonoink,line width=1.3pt] (5.172,0.000) -- (8.000,-2.000);
  \draw[zonoink,line width=1.3pt] (5.172,0.000) -- (6.939,-0.750);
  \draw[zonoink,line width=1.3pt] (8.653,0.000) -- (9.061,0.750);
  \draw[zonoink,line width=1.3pt] (8.653,0.000) -- (8.000,0.462);
  \draw[zonoink,line width=1.3pt] (8.653,0.000) -- (9.061,-0.750);
  \draw[zonoink,line width=1.3pt] (8.653,0.000) -- (8.000,-0.462);
  \draw[zonoink,line width=1.3pt] (10.828,0.000) -- (8.000,2.000);
  \draw[zonoink,line width=1.3pt] (10.828,0.000) -- (9.061,0.750);
  \draw[zonoink,line width=1.3pt] (10.828,0.000) -- (8.000,-2.000);
  \draw[zonoink,line width=1.3pt] (10.828,0.000) -- (9.061,-0.750);
  \draw[zonoink,line width=1.3pt] (7.347,0.000) -- (6.939,0.750);
  \draw[zonoink,line width=1.3pt] (7.347,0.000) -- (8.000,0.462);
  \draw[zonoink,line width=1.3pt] (7.347,0.000) -- (6.939,-0.750);
  \draw[zonoink,line width=1.3pt] (7.347,0.000) -- (8.000,-0.462);
  \filldraw[fill=white,draw=zonored,line width=1pt] (8.000,0.000) circle (3.4pt);
  \fill[zonored] (8.000,0.000) circle (1.1pt);
  \filldraw[fill=zonoorange,draw=white,line width=1pt] (8.000,1.675) -- (8.175,1.500) -- (8.000,1.325) -- (7.825,1.500) -- cycle;
  \filldraw[fill=zonoorange,draw=white,line width=1pt] (8.000,-1.325) -- (8.175,-1.500) -- (8.000,-1.675) -- (7.825,-1.500) -- cycle;
  \filldraw[fill=zonoorange,draw=white,line width=1pt] (5.172,0.175) -- (5.347,0.000) -- (5.172,-0.175) -- (4.997,0.000) -- cycle;
  \filldraw[fill=zonoorange,draw=white,line width=1pt] (8.653,0.175) -- (8.828,0.000) -- (8.653,-0.175) -- (8.478,0.000) -- cycle;
  \filldraw[fill=zonoorange,draw=white,line width=1pt] (10.828,0.175) -- (11.003,0.000) -- (10.828,-0.175) -- (10.653,0.000) -- cycle;
  \filldraw[fill=zonoorange,draw=white,line width=1pt] (7.347,0.175) -- (7.522,0.000) -- (7.347,-0.175) -- (7.172,0.000) -- cycle;
  \filldraw[fill=zonoblue,draw=white,line width=1pt] (8.000,2.000) circle (4.6pt);
  \filldraw[fill=zonoblue,draw=white,line width=1pt] (6.939,0.750) circle (4.6pt);
  \filldraw[fill=zonoblue,draw=white,line width=1pt] (9.061,0.750) circle (4.6pt);
  \filldraw[fill=zonoblue,draw=white,line width=1pt] (8.000,0.462) circle (4.6pt);
  \filldraw[fill=zonoblue,draw=white,line width=1pt] (8.000,-2.000) circle (4.6pt);
  \filldraw[fill=zonoblue,draw=white,line width=1pt] (6.939,-0.750) circle (4.6pt);
  \filldraw[fill=zonoblue,draw=white,line width=1pt] (9.061,-0.750) circle (4.6pt);
  \filldraw[fill=zonoblue,draw=white,line width=1pt] (8.000,-0.462) circle (4.6pt);
  \node[lblblue] (Lb3) at (8.100,3.050) {$x_1\wedge x_2$};
  \draw[zonoblue,line width=0.7pt] (Lb3.south) -- (8.000,2.180);
  \node[lblorange] (Lb4) at (11.978,0.000) {$x_1$};
  \draw[zonoorange,line width=0.7pt] (Lb4.west) -- (11.028,0.000);
  \node[font=\small\sffamily,anchor=north] at (0,3.95) {3D orientation view};
  \node[font=\small\sffamily,anchor=north] at (8.00,3.95) {Schlegel view};
\end{tikzpicture}
}
\caption{Two views of the threshold zonotope $Z_2$, which is
a rhombic dodecahedron. The left panel is a three-dimensional orientation view.
The right panel gives a Schlegel-style drawing (a planar representation) of the
one-skeleton. In the right panel, the blue circular
vertices are the eight trivalent vertices, corresponding to the functions with
$\sigma_2=3$; these are the conjunctions and disjunctions of two literals. The
orange diamond vertices are the six tetravalent vertices, corresponding to the
functions with $\sigma_2=4$; these are the four single-literal functions and
the two constants. The labels show the representative vertices $x_1\wedge x_2$
and $x_1$. By Theorem~\ref{thm:zonotope}, vertex degree is specification number.
The circled central point is the common Chow point of XOR and its complement;
these two functions are not threshold functions and do not give vertices of
$Z_2$. The TikZ code for this figure was prepared with the assistance of Anthropic Claude and verified by the author.}
\label{fig:threshold-zonotope-n2}
\end{figure}
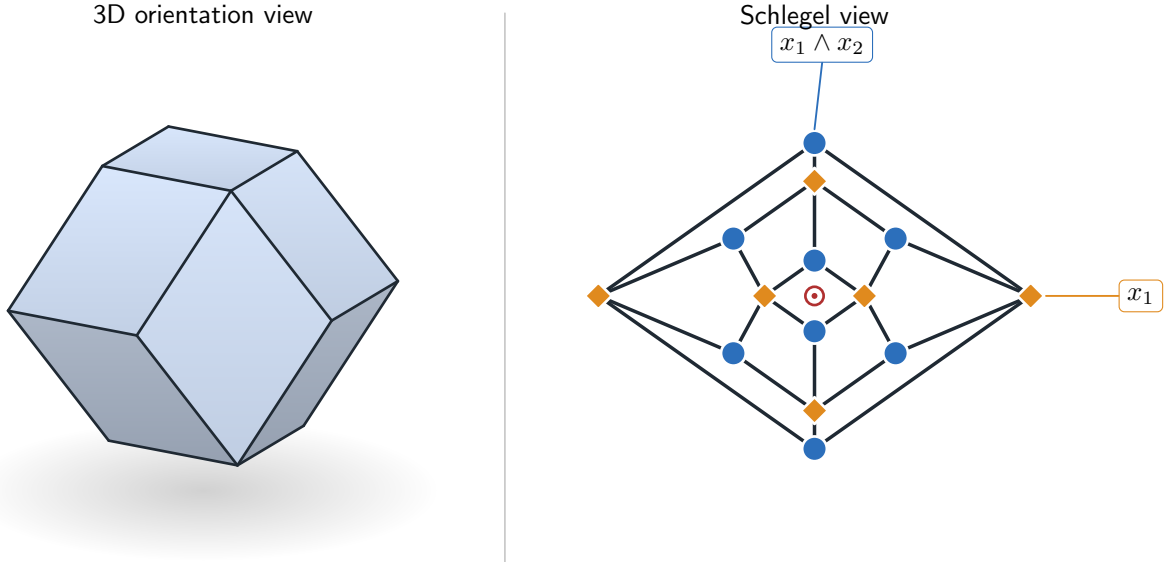

The zonotope makes several earlier quantities transparent. Under the vertex
identification $f\mapsto\chi(f)$, its one-skeleton is the one-inclusion graph
$G(\Hn)$. Thus chamber adjacency, one-inclusion adjacency, and zonotope edge
adjacency are the same structure: crossing the facet indexed by $x$ corresponds
to changing the value at $x$, and to traversing the zonotope edge parallel to
$\widetilde x$. In particular, the average specification number is the average
vertex degree of $Z_n$,
\[
        \overline\sigma_n=\frac{2|E(Z_n)|}{|V(Z_n)|},
        \qquad |V(Z_n)|=|\Hn|,
\]
and, comparing with Lemma~\ref{lem:facet-count}, the edge count is
$|E(Z_n)|=\sum_{x\in\B^n}r(\Arr_n^{H_x})=2^n r(\Res)$, so the resonance formula
of Corollary~\ref{cor:exact} is the statement that $Z_n$ has $2^n r(\Res)$
edges. In the same picture, Corollary~\ref{cor:lower} and
the known bound recorded in Theorem~\ref{thm:vc-average} give linear lower and upper bounds for the average
vertex degree of $Z_n$. The standard correspondence between a
central arrangement and its associated zonotope is classical
\cite[\S7.3]{Ziegler1995};  in the synthesis used here, the vertex coordinates are the Chow parameters and the vertex degree is the specification number. This places the facet correspondence, Zuev's threshold graph, the resonance count, and the Chow parameters on one object.

\section{Threshold functions with minimum specification number}\label{sec:Tn}

Lozin \emph{et al.}\ \cite{Lozin2017,Lozin2018LRO} disproved the converse of \cite[Theorem 2.10]{ABST1995} by exhibiting the case $k=n-1$ of what we call the \emph{Lozin family}
\[
        f_{n,k}
        =
        x_1x_2 \vee x_1x_3 \vee \cdots \vee x_1x_k
        \vee x_2x_3\cdots x_n,
        \qquad 3 \le k \le n-1.
\]
For $n\ge 4$, the function $f_{n,n-1}$ is threshold, non-lro, and has specification number $n+1$. The later paper \cite{Lozin2022} extended this to the full family $f_{n,k}$ and studied
\[
\Tn = \{f \in \Hn : f \text{ depends on all } n \text{ variables and } \sigma_n(f) = n+1\}
\]
and developed three operations, which we call the \emph{Lozin operations}: \emph{adding a variable}, \emph{extension on a variable}, and \emph{symmetric variables extension}, the first two of which already give an inductive description of $\Tn$ for $n \le 5$. They also found a member $f^\ast$ of $\Tk{6}$ not generated by any of these three operations, observed that $f^\ast$ arises from $f_{4,3}$ by a fourth operation, and conjectured that this operation always produces members of $\Tk{n+2}$, leaving the general case open \cite[\S6, equation (10)]{Lozin2022}.

\subsection{The simplicial-chamber reformulation}

The chamber-facet viewpoint reformulates $\Tn$ cleanly.

\begin{theorem}\label{thm:Tn-simplicial}
For $n\ge2$, $\Tn$ is the set of threshold functions $f\in\Hn$ whose chamber $\overline{C_f}$ in $\Arr_n$ is simplicial.
\end{theorem}

\begin{proof}
By Corollary~\ref{cor:simplicial}, $\overline{C_f}$ is simplicial if and only if $\sigma_n(f)=n+1$. For $n\ge2$, the condition $\sigma_n(f)=n+1$ forces $f$ to depend on all $n$ variables. Suppose instead that $f$ depends on only $r\le n-1$ of its variables. If $r=0$ then $f$ is constant; for $f\equiv1$ the function equal to $0$ at a single vertex and $1$ elsewhere is a threshold function agreeing with $f$ except at that vertex, so every vertex lies in any specifying set and $\sigma_n(f)=2^n$ (symmetrically for $f\equiv0$). If $1\le r\le n-1$, let $g$ be the restriction of $f$ to its relevant variables; by \cite[Proposition~2.14]{ABST1995} the $n-r$ irrelevant variables multiply the specification number by $2^{\,n-r}$, so $\sigma_n(f)=2^{\,n-r}\sigma_r(g)\ge 2^{\,n-r}(r+1)$, using $\sigma_r(g)\ge r+1$ (Corollary~\ref{cor:lower}). The minimum of $2^{\,n-r}(r+1)$ over $1\le r\le n-1$ is $2n$, attained at $r=n-1$, and $2^n\ge 2n$ for $n\ge1$, so in either case $\sigma_n(f)\ge 2n>n+1$, contradicting $\sigma_n(f)=n+1$. Hence, for $n\ge2$, the clause requiring dependence on all $n$ variables in the definition of $\Tn$ is implied by $\sigma_n(f)=n+1$, and $\Tn$ is exactly the set of threshold functions whose chamber $\overline{C_f}$ in $\Arr_n$ is simplicial. The hypothesis $n\ge2$ is necessary: at $n=1$ all four threshold functions on one variable have $\sigma_1=2=n+1$ and simplicial chambers, yet $\Tk{1}=\{x_1,\overline{x_1}\}$ excludes the two constants.
\end{proof}

\begin{remark}
Theorem~\ref{thm:Tn-simplicial} does more than translate the equality
$\sigma_n(f)=n+1$ into a facet count: it also shows that, for $n\ge2$, the
requirement that a member of $\Tn$ depend on all variables is redundant.
Minimum specification number already forces full relevance. Thus, for $n\ge2$,
$\Tn$ may be viewed intrinsically as the class of threshold functions with
$\sigma_n(f)=n+1$, or geometrically as the class of simplicial chambers of
$\Arr_n$. Such a chamber $\overline{C_f}$ has exactly $n+1$ facets, supported
by the hyperplanes $H_x$ with $x\in\Ess(f)$, so the problem is to determine
which labelled sets of $n+1$ essential points arise from functions in $\Tn$.
The labels matter: $\Ess(f)$ alone does not determine $f$, since $f$ and its
complement have the same essential points but opposite labels. The chamber
analysis of the operations that build $\Tn$ begins in
Subsection~\ref{ssec:lozin-operations} and continues, for the
symmetric-variables extension and the fourth operation, in
Subsection~\ref{ssec:symmetric} and Section~\ref{sec:fourth}.
\end{remark}

\subsection{Linear read-once, Chow, and forbidden subfunctions}\label{ssec:lro-chow}

The 2018 paper of Lozin \emph{et al.}\ gives complementary characterisations of the lro class that are useful for separating the lro theory from the minimum-specification theory. Among read-once Boolean functions, the following are equivalent \cite[Theorem~7]{Lozin2018LRO}: being lro, being a Chow function, and having no restriction obtained, up to variable negation, from
\[
  g_1=(x\vee y)\wedge(z\vee u), \qquad
  g_2=(x\wedge y)\vee(z\wedge u).
\]
Here a Chow function means a Boolean function uniquely determined by its Chow parameters (Section~\ref{sec:zonotope}). The phrase ``up to variable negation'' means that, before taking the restriction, we may replace any subset of variables by their complements. Thus the forbidden restrictions include $g_1$, $g_2$, and all functions obtained from them by relabelling variables and complementing variables.

Within the threshold universe, lro functions are characterised by a different forbidden-restriction family: a threshold function is lro if and only if it contains no restriction, up to relabelling and variable negation, of
\[
  q_m = x_1(x_2\vee\cdots\vee x_m) \vee x_2x_3\cdots x_m, \qquad m\ge 3,
\]
where $q_m$ is positive, threshold, non-lro, minimal with respect to restrictions, and has specification number $2m$ \cite[Theorem 8 and Lemma 8]{Lozin2018LRO}.

These results should not be taken as a characterisation of minimum specification. The functions $q_m$ are the minimal threshold obstructions to lro, yet the counterexamples $f_n$ of \cite[Theorem 9]{Lozin2018LRO} have minimum specification number $n+1$ while remaining non-lro, and the paper observes that this minimum-specification family is not closed under restrictions. Thus the forbidden-restriction characterisations of lro functions and the $\Tn$ simplicial-chamber problem are distinct. The lro functions form one natural subfamily of $\Tn$, but $\Tn$ is strictly larger and is not governed by the same forbidden-restriction characterisation.

\begin{remark}
We should regard simplicial chambers as a restricted subclass. By the theorem of Fukuda, Tamura, and Tokuyama \cite[Corollary~1.3]{Fukuda1993}, every spherical arrangement of this dimension has average chamber-facet count at most $2n$; in a generic central arrangement of $N$ hyperplanes in $\R^{n+1}$, the average tends to $2n$ as $N\to\infty$ (Section~\ref{sec:growth}), so a simplicial chamber, with its $n+1$ facets, is atypical. The Boolean threshold arrangement $\Arr_n$ has hyperplanes from the lifted cube points $(1,x)$, $x\in\B^n$, which satisfy many linear dependencies, and identifying which of its chambers are simplicial is a substantive condition on this structured arrangement. For $n\ge2$, the class $\Tn$ is precisely the set of threshold functions whose associated chamber closure is simplicial. The Lozin family $f_{n,k}$ shows that such simplicial chambers can arise from functions more elaborate than nested or linear read-once formulae.
\end{remark}

\subsection{Lozin operations as chamber operations}\label{ssec:lozin-operations}

Lozin \emph{et al.}\ build members of $\Tn$ inductively, beginning with small ones and
applying operations designed to hold the specification number at its minimum.
The chamber model explains why their first two operations do this. Each is, on
chambers, the same construction, the product with a closed half-line, which sends
a simplicial cone to a simplicial cone and adjoins a single facet. Since, for $n\ge2$, a member
of $\Tn$ is exactly a threshold function with simplicial chamber
(Theorem~\ref{thm:Tn-simplicial}), these operations carry $\Tn$ into $\Tk{n+1}$
(Corollary~\ref{cor:base-preserve}). We establish this and determine the essential
sets that result. Throughout this subsection, positive means monotone increasing
in each variable.

Two polyhedral cones $K\subseteq V$ and $K'\subseteq V'$ are \emph{linearly
equivalent} if there is an invertible linear map $T:V\to V'$ such that
$T(K)=K'$. Such a map preserves dimension and carries bounding halfspaces to
bounding halfspaces, so it maps the facets of $K$ bijectively to those of $K'$.
Linearly equivalent cones therefore have the same number of facets, and in
particular linear equivalence preserves simpliciality.

We start with the single geometric fact behind both operations: taking the product with
a closed half-line preserves simpliciality, adding exactly one facet.

\begin{lemma}[The half-line product preserves simpliciality]\label{lem:halfline-simplicial}
Let $K\subseteq\R^d$ be a full-dimensional pointed polyhedral cone. Then
$K\times\R_{\ge0}\subseteq\R^{d+1}$ is again full-dimensional and pointed, and its
facets are the sets $F\times\R_{\ge0}$, one for each facet $F$ of $K$, together
with the base $K\times\{0\}$. It therefore has exactly one more facet than $K$,
and it is simplicial if and only if $K$ is.
\end{lemma}

\begin{proof}
Pointedness and full dimensionality pass to the product. The nonempty faces of a product of polyhedra are the products $F\times G$ of a face $F$ of one factor and a face $G$ of the other, and $\dim(F\times G)=\dim F+\dim G$; a codimension-one face therefore pairs a facet of one factor with the whole of the other. Since $\R_{\ge0}$ has the single facet $\{0\}$, the
facets of $K\times\R_{\ge0}$ are the sets $F\times\R_{\ge0}$ for $F$ a facet of
$K$, together with the base $K\times\{0\}$. This is one more facet than $K$, in
one higher dimension. A full-dimensional pointed cone in $\R^{d+1}$ is simplicial
exactly when it has $d+1$ facets, so $K\times\R_{\ge0}$ is simplicial precisely
when $K$ has $d$ facets, that is, precisely when $K$ is simplicial.
\end{proof}

\begin{lemma}\label{lem:positive-weights}
Let $f$ be a positive threshold function depending on all $n$ variables, and
write $L_a(x)=a_0+\sum_{j=1}^n a_jx_j$. If $a\in C_f$, then
$a_j>0$ for $j=1,\ldots,n$. If $a\in\overline{C_f}$, then $a_j\ge0$ for
$j=1,\ldots,n$.
\end{lemma}

\begin{proof}
For each relevant variable $x_j$, positivity gives a point $u$ with $u_j=0$
such that $f(u)=0$ and $f(u+e_j)=1$. If $a\in C_f$, then
\[
        L_a(u)<0<L_a(u+e_j)=L_a(u)+a_j,
\]
so $a_j>0$. The closure statement follows by taking limits.
\end{proof}

\begin{proposition}[Adding a variable as a product with a closed half-line]\label{prop:add-variable-product}
Let $f$ be a positive threshold function depending on all $n$ variables.
Let $g=f\vee y$ and $h=f\wedge y$. Then both chamber closures
$\overline{C_g}$ and $\overline{C_h}$ in the $(n+1)$-variable threshold
arrangement are linearly equivalent to
\[
        \overline{C_f}\times\R_{\ge0}.
\]
Consequently, adding a variable preserves simpliciality of chamber closures.
More precisely,
\[
        \Ess(g)=\{(x,0):x\in\Ess(f)\}\cup\{(0,\ldots,0,1)\},
\]
and
\[
        \Ess(h)=\{(x,1):x\in\Ess(f)\}\cup\{(1,\ldots,1,0)\}.
\]
\end{proposition}

\begin{proof}
Use coordinates $(a,b)$ on the parameter space for threshold functions in the
variables $(x,y)$, and write
\[
        M_{a,b}(x,y)=a_0+\sum_{j=1}^n a_jx_j+by.
\]
Thus $(a,b)\in\overline{C_g}$, respectively $(a,b)\in\overline{C_h}$, means
that $M_{a,b}$ satisfies the weak sign conditions determined by $g$,
respectively by $h$, on $\B^{n+1}$.

For $g=f\vee y$, we have $g(x,0)=f(x)$ and $M_{a,b}(x,0)=L_a(x)$. Hence
the sign conditions on the layer $y=0$ are precisely the closed chamber
conditions $a\in\overline{C_f}$. The layer $y=1$ is required to be positive
at every point:
\[
        a_0+\sum_{j=1}^n a_jx_j+b\ge0
        \qquad (x\in\B^n).
\]
Since the layer $y=0$ gives $a\in\overline{C_f}$, and $f$ is positive,
Lemma~\ref{lem:positive-weights} gives $a_j\ge0$ for all $j$. Hence the
minimum of $L_a(x)$ over the cube occurs at $x=(0,\ldots,0)$, so this whole
layer is equivalent to
\[
        a_0+b\ge0.
\]
Combining the two layers, $\overline{C_g}=\{(a,b):a\in\overline{C_f}\text{ and }a_0+b\ge0\}$. The shear $(a,b)\mapsto(a,u)$ with $u=a_0+b$ is a linear isomorphism fixing $a$, and it carries these conditions to $a\in\overline{C_f}$ and $u\ge0$, hence identifies $\overline{C_g}$ with $\overline{C_f}\times\R_{\ge0}$.

For $h=f\wedge y$ the layer reproducing $f$ is $y=1$, where $M_{a,b}(x,1)=(a_0+b)+\sum_{j=1}^n a_jx_j$ is the original form for $f$ with constant term $a_0+b$. Put
\[
        c_0=a_0+b,
        \qquad c_j=a_j\quad(1\le j\le n),
\]
so that this layer imposes exactly $c\in\overline{C_f}$. The layer $y=0$ is
required to be non-positive:
\[
        a_0+\sum_{j=1}^n a_jx_j\le0
        \qquad (x\in\B^n).
\]
Since $c_j=a_j\ge0$, the maximum over the cube occurs at $(1,\ldots,1)$, so
this is equivalent to
\[
        a_0+\sum_{j=1}^n a_j\le0.
\]
Combining the two layers, $\overline{C_h}=\{(a,b):c\in\overline{C_f}\text{ and }a_0+\sum_{j=1}^n a_j\le0\}$ with $c=(a_0+b,a_1,\ldots,a_n)$. The map $(a,b)\mapsto(c,u)$ with $u=-a_0-\sum_{j=1}^n a_j$ is a linear isomorphism, carrying these conditions to $c\in\overline{C_f}$ and $u\ge0$, hence identifies $\overline{C_h}$ with $\overline{C_f}\times\R_{\ge0}$.

By Lemma~\ref{lem:halfline-simplicial}, $\overline{C_f}\times\R_{\ge0}$ has the facets of $\overline{C_f}$, each crossed with $\R_{\ge0}$, together with the base $u=0$; it has one more facet than $\overline{C_f}$, and is simplicial if and only if $\overline{C_f}$ is.

We obtain the essential-set descriptions by applying the chamber-facet identification to these facets (Theorem~\ref{thm:facet}), under which a cube point is essential exactly when its hyperplane supports a facet. On the unchanged layer $M_{a,b}$ reduces to the original affine form for $f$, so the old facets lie on $H_{(x,0)}$ for $g$ and on $H_{(x,1)}$ for $h$, with $x\in\Ess(f)$. The new facet $u=0$ is $a_0+b=0$ for $g$, namely $H_{(0,\ldots,0,1)}$, and $a_0+\sum_{j}a_j=0$ for $h$, namely $H_{(1,\ldots,1,0)}$. These displayed cube points are precisely the new essential points of the transformed functions.
\end{proof}

A Boolean function $f$ is \emph{self-dual} if $f(\mathbf 1-x)=1-f(x)$ for every $x\in\B^n$, that is, if complementing every input complements the output.

\begin{proposition}[Extension on a variable as a product with a closed half-line]\label{prop:extension-product}
Let $f$ be a positive threshold function depending on all $n$ variables, and
let $g=f^{(x_i,y)}$ be the extension on the variable $x_i$ of \cite{Lozin2022}, namely
\[
        g(x,y)
        =x_i\bigl(y\vee f|_{x_i=1}\bigr)
         \vee y\, f|_{x_i=0},
\]
where $f|_{x_i=1}$ and $f|_{x_i=0}$ are the restrictions of $f$ obtained by fixing $x_i$, each a function of the remaining $n-1$ variables.
If $f$ is not self-dual, then $\overline{C_g}$ is linearly equivalent to
\[
        \overline{C_f}\times\R_{\ge0}.
\]
In particular, extension on a variable preserves simpliciality for non-self-dual
positive threshold functions. For a coefficient vector $c=(c_0,\ldots,c_n)$ in
$C_f$, put
\[
        D(c)=2c_0+\sum_{j=1}^n c_j .
\]
This linear functional has constant non-zero sign on $C_f$. Writing
$\lambda_i(x)=(x,1-x_i)$, the essential set is
\[
        \Ess(g)=\{\lambda_i(x):x\in\Ess(f)\}\cup\{q^{(i)}\},
\]
where $q^{(i)}\in\B^{n+1}$ is determined by the sign of $D$ on $C_f$ as
follows. If $D>0$ on $C_f$, then
\[
        q^{(i)}_i=0,\qquad q^{(i)}_y=0,\qquad
        q^{(i)}_j=1\quad(j\ne i).
\]
If $D<0$ on $C_f$, then
\[
        q^{(i)}_i=1,\qquad q^{(i)}_y=1,\qquad
        q^{(i)}_j=0\quad(j\ne i).
\]
\end{proposition}

\begin{proof}
Use the same ambient affine form
\[
        M_{a,b}(x,y)=a_0+\sum_{j=1}^n a_jx_j+by
\]
for the transformed function $g$. Thus $(a,b)\in\overline{C_g}$ means that
$M_{a,b}$ satisfies the weak sign conditions determined by $g$ on $\B^{n+1}$.
We analyse these conditions after the linear change of coordinates
\[
        \Phi:\R^{n+2}\longrightarrow\R^{n+2},\qquad
        (a_0,\ldots,a_n,b)\longmapsto(c_0,\ldots,c_n,b),
\]
where
\[
        c_0=a_0+b,\qquad
        c_i=a_i-b,\qquad
        c_j=a_j\quad(j\ne i).
\]
This is a linear isomorphism, with inverse
\[
        a_0=c_0-b,\qquad
        a_i=c_i+b,\qquad
        a_j=c_j\quad(j\ne i).
\]
Writing $L_c(x)=c_0+\sum_{j=1}^n c_jx_j$, direct substitution gives
\[
        M_{a,b}(x,y)=L_c(x)+b\,(x_i+y-1).
\]
On the off-diagonal layers $(x_i,y)=(0,1)$ and $(x_i,y)=(1,0)$ the factor
$x_i+y-1$ vanishes, so $M_{a,b}(x,y)=L_c(x)$. These two layers contain the
embedded copy of $\B^n$ given by
\[
        \lambda_i(x)=(x,1-x_i).
\]
For every $x\in\B^n$, the defining formula for the extension gives
\[
        g(\lambda_i(x))=f(x),
\]
and the affine form satisfies
\[
        M_{a,b}(\lambda_i(x))=L_c(x).
\]
Thus the closed-chamber sign conditions for $g$ on this embedded copy are
precisely
\[
        L_c(x)\ge0\quad(f(x)=1),\qquad
        L_c(x)\le0\quad(f(x)=0),
\]
which is the condition $c\in\overline{C_f}$.

From the defining formula, $g$ is constant on each diagonal layer: $g=0$ when
$x_i=y=0$, and $g=1$ when $x_i=y=1$.

On the layer $(x_i,y)=(0,0)$ every point is therefore a zero, which in the
coordinates $c$ becomes
\[
        c_0-b+\sum_{j\ne i}c_jx_j\le0 \qquad (x_j\in\B,\ j\ne i).
\]
Since $c_j\ge0$ for $j\ne i$ by Lemma~\ref{lem:positive-weights}, the left
side is largest when $x_j=1$ for all $j\ne i$, so this reduces to
\[
        b\ge c_0+\sum_{j\ne i}c_j.
\]
On the layer $(x_i,y)=(1,1)$ every point is a one, which becomes
\[
        c_0+c_i+b+\sum_{j\ne i}c_jx_j\ge0 \qquad (x_j\in\B,\ j\ne i);
\]
here the left side is smallest when $x_j=0$, so this reduces to
\[
        b\ge -c_0-c_i.
\]
Hence $\overline{C_g}$ is the set of all $(c,b)$ such that
$c\in\overline{C_f}$ and
\[
        b\ge c_0+\sum_{j\ne i}c_j,
        \qquad
        b\ge -c_0-c_i.
\]

The two lower bounds differ by
\[
\left(c_0+\sum_{j\ne i}c_j\right)-(-c_0-c_i)
        =2c_0+\sum_{j=1}^n c_j
        =D(c).
\]
Also,
\[
        L_c(x)+L_c(\mathbf 1-x)=D(c)
        \qquad (x\in\B^n).
\]
Thus, if $D(c)=0$, then $L_c(\mathbf 1-x)=-L_c(x)$ for every $x$. Here
$c\in C_f$, so $L_c$ has the strict signs defining $f$ and vanishes at no cube
point. Hence this identity would imply
\[
        f(\mathbf 1-x)=1-f(x)
        \qquad (x\in\B^n),
\]
so $f$ would be self-dual. Since $f$ is not self-dual, $D$ does not vanish on
$C_f$. As $C_f$ is connected and $D$ is linear, $D$ has a constant strict sign
on $C_f$. Since $D$ is linear, this strict sign condition extends weakly to
$\overline{C_f}$: if $D>0$ on $C_f$, then $D(c)\ge0$ for all
$c\in\overline{C_f}$, while if $D<0$ on $C_f$, then $D(c)\le0$ for all
$c\in\overline{C_f}$.

If $D>0$ on $C_f$, then
\[
        c_0+\sum_{j\ne i}c_j\ge -c_0-c_i
        \qquad(c\in\overline{C_f}),
\]
so the first lower bound for $b$ implies the second. If $D<0$ on $C_f$, then
\[
        -c_0-c_i\ge c_0+\sum_{j\ne i}c_j
        \qquad(c\in\overline{C_f}),
\]
so the second lower bound implies the first. Define
\[
B(c)=
\begin{cases}
c_0+\sum_{j\ne i}c_j, & \text{if }D>0\text{ on }C_f,\\[1mm]
-c_0-c_i, & \text{if }D<0\text{ on }C_f.
\end{cases}
\]
Then
\[
        \overline{C_g}
        =
        \{(c,b):c\in\overline{C_f},\ b\ge B(c)\}.
\]
Finally, the linear change of coordinates
\[
        (c,b)\longmapsto(c,u),\qquad u=b-B(c),
\]
is a shear, with inverse $(c,u)\mapsto(c,u+B(c))$. It identifies
\[
        \{(c,b):c\in\overline{C_f},\ b\ge B(c)\}
\]
with
\[
        \overline{C_f}\times\R_{\ge0}.
\]
Hence $\overline{C_g}$ is linearly equivalent to
$\overline{C_f}\times\R_{\ge0}$. By Lemma~\ref{lem:halfline-simplicial},
$\overline{C_g}$ is then simplicial if and only if $\overline{C_f}$ is, with
one more facet.

To identify $\Ess(g)$, apply Theorem~\ref{thm:facet}: the essential points of
$g$ are the cube points whose hyperplanes support the facets of
$\overline{C_g}$. In the $(c,u)$-coordinates introduced above, the closed chamber
is exactly
\[
        \overline{C_f}\times\R_{\ge0}.
\]
Its facets are therefore the products
\[
        F\times\R_{\ge0},
\]
where $F$ ranges over the facets of $\overline{C_f}$, together with the single
additional facet
\[
        \overline{C_f}\times\{0\}.
\]
The preceding linear changes of coordinates identify this product cone with
$\overline{C_g}$, so the facets of $\overline{C_g}$ are the corresponding images
of these product facets.

Let $x\in\Ess(f)$. The corresponding facet of $\overline{C_f}$ is given by
$L_c(x)=0$. Since
\[
        M_{a,b}(\lambda_i(x))=L_c(x),
\]
the corresponding inherited facet of $\overline{C_g}$ is given, in the original
coordinates, by
\[
        M_{a,b}(\lambda_i(x))=0.
\]
Thus the inherited essential point of $g$ is $\lambda_i(x)$. As $\lambda_i$ is
injective, the inherited facets contribute exactly
$\{\lambda_i(x):x\in\Ess(f)\}$.

The remaining facet is the equality case of the surviving diagonal-layer
inequality, namely $u=0$, or equivalently $b=B(c)$. If $D>0$ on $C_f$, this
surviving inequality comes from the layer $(x_i,y)=(0,0)$ at $x_j=1$ for
$j\ne i$, namely from the point $q^{(i)}$ with $q^{(i)}_i=q^{(i)}_y=0$ and
$q^{(i)}_j=1$ for $j\ne i$; at this point
\[
        M_{a,b}(q^{(i)})=\left(c_0+\sum_{j\ne i}c_j\right)-b=-u.
\]
If $D<0$ on $C_f$, the surviving inequality comes from the layer
$(x_i,y)=(1,1)$ at $x_j=0$ for $j\ne i$, namely from the point $q^{(i)}$ with
$q^{(i)}_i=q^{(i)}_y=1$ and $q^{(i)}_j=0$ for $j\ne i$; at this point
\[
        M_{a,b}(q^{(i)})=(c_0+c_i)+b=u.
\]
In either case the remaining facet is given by $M_{a,b}(q^{(i)})=0$, so it
corresponds to the new essential point $q^{(i)}$. Theorem~\ref{thm:facet} gives
\[
        \Ess(g)=\{\lambda_i(x):x\in\Ess(f)\}\cup\{q^{(i)}\}.
\]
\end{proof}

\begin{corollary}[Adding a variable and extension on a variable preserve minimum specification]\label{cor:base-preserve}
Let $f\in\Tn$ be
positive, with $n\ge2$. Then $f\vee y$, $f\wedge y$, and $f^{(x_i,y)}$ all lie in
$\Tk{n+1}$, with specification number rising from $n+1$ to $n+2$. The same holds
for every $f\in\Tn$ , not necessarily positive: the three images lie in $\Tk{n+1}$. The reduction to a positive representative, by permuting and complementing variables (a congruence), is carried out within the proof.
\end{corollary}

\begin{proof}
By Theorem~\ref{thm:Tn-simplicial}, $f\in\Tn$ exactly when $\overline{C_f}$ is
simplicial, so $\overline{C_f}$ has $n+1$ facets.

\emph{Self-duality is excluded for $n\ge2$.} For self-dual $f$ the essential set
is invariant under $x\mapsto\mathbf 1-x$:
\begin{align*}
        x\in\Ess(f)
        &\iff f\oplus\mathbf 1_{\{x\}}\ \text{threshold}\\
        &\iff (1-f)\oplus\mathbf 1_{\{\mathbf 1-x\}}\ \text{threshold}\\
        &\iff f\oplus\mathbf 1_{\{\mathbf 1-x\}}\ \text{threshold}\\
        &\iff \mathbf 1-x\in\Ess(f),
\end{align*}
where the first step substitutes $\mathbf 1-z$ for the input $z$, which
preserves the threshold property and, by self-duality $f(\mathbf 1-z)=1-f(z)$,
carries $f\oplus\mathbf 1_{\{x\}}$ to $(1-f)\oplus\mathbf 1_{\{\mathbf 1-x\}}$. The second
complements the output, which preserves it too. The involution
$x\mapsto\mathbf 1-x$ has no fixed point on $\B^n$, so it splits $\Ess(f)$ into
disjoint complementary pairs $\{x,\mathbf 1-x\}$, and $|\Ess(f)|$ is even. The
lifted points of a pair sum to
$\widetilde x+\widetilde{\mathbf 1-x}=(2,1,\ldots,1)$, so two distinct pairs
$\{x,\mathbf 1-x\}$ and $\{y,\mathbf 1-y\}$ would give the relation
$\widetilde x+\widetilde{\mathbf 1-x}-\widetilde y-\widetilde{\mathbf 1-y}=0$
among four distinct lifted points. The facet normals of $\overline{C_f}$ are the
signed lifts $\varepsilon_f(z)\widetilde z$ for $z\in\Ess(f)$, each a non-zero
scalar multiple of the corresponding lift. Scaling the vectors of a set by
non-zero scalars leaves its linear dependences intact, so the displayed relation
shows that the four facet normals $\varepsilon_f(x)\widetilde x$,
$\varepsilon_f(\mathbf 1-x)\widetilde{\mathbf 1-x}$, $\varepsilon_f(y)\widetilde y$,
and $\varepsilon_f(\mathbf 1-y)\widetilde{\mathbf 1-y}$ are linearly dependent
(explicitly, $\varepsilon_f(x)$, $\varepsilon_f(\mathbf 1-x)$, $-\varepsilon_f(y)$,
$-\varepsilon_f(\mathbf 1-y)$ are non-zero coefficients of a vanishing
combination of them). Since $\overline{C_f}$
is simplicial in $\R^{n+1}$, it has exactly $n+1$ facets. By
Lemma~\ref{lem:facet-normals-span}, its facet normals span $\R^{n+1}$, so being $n+1$ spanning vectors in $\R^{n+1}$ they are linearly independent, as is every subset of them. Four linearly dependent facet normals are therefore impossible. Thus
$\Ess(f)$ has at most one pair, $n+1=|\Ess(f)|\le2$, and $n\le1$. For $n\ge2$
the extension of Proposition~\ref{prop:extension-product} therefore applies
without further hypothesis.

By Propositions~\ref{prop:add-variable-product} and~\ref{prop:extension-product}
the chamber of each image, in the $(n+1)$-variable arrangement, is linearly
equivalent to $\overline{C_f}\times\R_{\ge0}$, a simplicial cone with $n+2$
facets in $\R^{n+2}$ by Lemma~\ref{lem:halfline-simplicial}. Theorem~\ref{thm:Tn-simplicial},
applied in dimension $n+1$, places each image in $\Tk{n+1}$. In particular, the
relevance condition is automatic. Corollary~\ref{cor:simplicial} gives its
specification number as $(n+1)+1=n+2$.

\emph{Congruence.} It remains to pass from the positive case to functions
congruent to it. Every threshold function is congruent to a positive one. Choose a strict representation; replacing each variable carrying a negative weight by its
complement gives a representation with non-negative weights, which is positive
(monotone increasing in each variable), and such complementation is a congruence.
Since congruence preserves membership in $\Tn$, every $f\in\Tn$ is congruent to
a positive member of $\Tn$, and it therefore suffices to treat the positive case.
Let $h\in\Tn$ be congruent to a positive function $p\in\Tn$, and
choose $S\subseteq[n]$ such that
\[
        h(x)=p(x\oplus\mathbf 1_S).
\]
Congruence preserves membership in $\Tk{n+1}$, since it merely relabels the cube
points and hence the corresponding chamber facets. The added-variable images are
immediate:
\[
        (h\vee y)(x,y)=(p\vee y)(x\oplus\mathbf 1_S,y),
        \qquad
        (h\wedge y)(x,y)=(p\wedge y)(x\oplus\mathbf 1_S,y).
\]
Since the corresponding images of $p$ lie in $\Tk{n+1}$, the corresponding
images of $h$ do too.

It remains only to discuss extension on a variable. Write $E_i(u)=u^{(x_i,y)}$.
If $i\notin S$, then complementing the variables in $S$ leaves the distinguished
coordinate $x_i$ unchanged. Substitution in the defining formula of
Proposition~\ref{prop:extension-product} gives
\[
        E_i(h)(x,y)=E_i(p)(x\oplus\mathbf 1_S,y),
\]
so $E_i(h)$ is again congruent to an image already known to lie in
$\Tk{n+1}$.

Now suppose $i\in S$, and put $T=S\setminus\{i\}$. Let $P_1$ denote the
restriction $p|_{x_i=1}$ after complementing the variables in $T$, and let $P_0$
denote the corresponding restriction $p|_{x_i=0}$ after complementing the
variables in $T$. Since $h(x)=p(x\oplus\mathbf 1_S)$, complementing $x_i$
interchanges the two restrictions, so
\[
        h|_{x_i=1}=P_0,
        \qquad
        h|_{x_i=0}=P_1.
\]
Hence
\[
        E_i(h)=x_i(y\vee P_0)\vee yP_1.
\]
On the other hand, if in $E_i(p)$ we complement the variables in $T$ and then
interchange $x_i$ and $y$, its defining formula becomes
\[
        y(x_i\vee P_1)\vee x_iP_0.
\]
The two expressions are equal, since both expand to
\[
        x_iy\vee x_iP_0\vee yP_1.
\]
Thus $E_i(h)$ is a congruent copy of $E_i(p)$, and therefore
$E_i(h)\in\Tk{n+1}$.
\end{proof}

\begin{remark}[Adding a variable and extension on a variable coincide on chambers]\label{rem:two-extensions-coincide}
The two operations are distinct as Boolean constructions: adding a variable
attaches a fresh variable by a single conjunction or disjunction, while
extension on a variable entangles a new variable with an existing one
through the two restrictions of $f$. Propositions~\ref{prop:add-variable-product}
and~\ref{prop:extension-product} show that on chambers they are the same
operation, the product with a closed half-line,
\[
        \overline{C}\ \longmapsto\ \overline{C}\times\R_{\ge0}.
\]
Corollary~\ref{cor:base-preserve} is the consequence Lozin \emph{et al.}\ draw from this: both
keep a function in the minimum-specification class, raising the specification
number from $n+1$ to $n+2$, because the half-line contributes a single new facet
and hence one new essential point. The product structure also makes
thresholdness automatic, the image being again a chamber of the larger
arrangement. By contrast, the symmetric-variables extension of \cite{Lozin2022}
is not a product of this form, and need not preserve thresholdness. We treat it
in Subsection~\ref{ssec:symmetric}, and resolve the fourth operation conjectured
by \cite{Lozin2022} in Section~\ref{sec:fourth}.
\end{remark}

\subsection{The symmetric-variables extension}
\label{ssec:symmetric}

We next turn to the third operation of Lozin \emph{et al.}, the symmetric-variables extension.
In their notation, if $u$ and $v$ are symmetric variables of a positive Boolean
function $f$, and $y$ is a new variable, the $(u,v,y)$-s-extension is
\[
        f^{(u,v,y)}
        =
        uvy\, f|_{u=1,v=1}
        \vee
        (u\vee v\vee y)\, f|_{u=1,v=0}
        \vee
        f|_{u=0,v=0}.
\]
This formula uses the three possible types of restriction obtained from the
original pair $u,v$: both set to $1$, exactly one set to $1$, and both set to
$0$. Since $u$ and $v$ are symmetric, the two restrictions with exactly one of
them set to $1$ agree, so it is enough to write $f|_{u=1,v=0}$.

We generalise this operation by adjoining more new variables to the same
original symmetric pair. For $\ell\ge3$, we start with two original symmetric
variables $u,v$ and adjoin $\ell-2$ new variables. The transformed function has
a final symmetric block of size $\ell$, consisting of $u,v$ and the $\ell-2$ new
variables. Thus the operation of Lozin \emph{et al.} is the case $\ell=3$.

The purpose of allowing general $\ell$ is organisational: it exposes the
mechanism behind the $s$-extension rather than changing the object of primary
interest. The thresholdness criterion and the chamber description are most
transparent when written for a final symmetric block of size $\ell$.

Unlike adding a variable and extension on a variable, this operation can fail to
produce a threshold function. We therefore first identify when the transformed
function is threshold. In those cases, the representing chamber of the
transformed function can be described from the original chamber by one explicit
piecewise-linear map.

Write the original function as $f(z,u,v)$, where $u$ and $v$ are the two
symmetric original variables and $z$ denotes the remaining arguments. Define the
collapse map
\[
        \pi:\B^\ell\to\B^2,
        \qquad
        \pi(r)=\left(\bigvee_{s=1}^{\ell} r_s,
                     \bigwedge_{s=1}^{\ell} r_s\right).
\]
For $\ell\ge3$, define the \emph{$\ell$-fold symmetric extension}
$g_\ell=g_\ell(z,r_1,\ldots,r_\ell)$ by
\[
        g_\ell(z,r)=f(z,\pi(r))
        =f\left(z,\bigvee_{s=1}^{\ell}r_s,
                  \bigwedge_{s=1}^{\ell}r_s\right).
\]
Thus the block $r$ is collapsed to three possible cases:
\[
        g_\ell(z,r)=
        \begin{cases}
        f(z,0,0), & r=(0,\ldots,0),\\[1mm]
        f(z,1,1), & r=(1,\ldots,1),\\[1mm]
        f(z,1,0), & \text{otherwise}.
        \end{cases}
\]
The last case is well defined because the two original variables are symmetric,
so $f(z,1,0)=f(z,0,1)$.

Write $c=(c_0,c_1,\ldots,c_{n-2})$, and set
\[
        A_c(z)=c_0+\sum_{k=1}^{n-2}c_k z_k .
\]
A representing affine form for $f$ will be written as
\[
        L_{c,p,q}(z,u,v)=A_c(z)+pu+qv.
\]
For the transformed Boolean function $g_\ell(z,r)$, consider an arbitrary
affine form on the variables $(z,r)$,
\[
        M_{c,b}(z,r)=A_c(z)+\sum_{s=1}^{\ell} b_s r_s .
\]
The following analysis determines exactly when there are coefficients $(c,b)$
for which $M_{c,b}$ has the signs prescribed by $g_\ell$. This is precisely the
condition that $g_\ell$ be threshold.

For a block-weight vector $b=(b_1,\ldots,b_\ell)$, write
\[
        s(b)=\sum_{i=1}^{\ell}b_i,
        \qquad
        m(b)=\min_i b_i .
\]
The all-zero block contributes $0$, the all-one block contributes $s(b)$, and a
mixed block contributes a nonempty proper subset sum
\[
        \sum_{i\in I}b_i,\qquad \varnothing\ne I\subsetneq[\ell].
\]
When all $b_i\ge0$, the least mixed-block contribution is $m(b)$, attained by
a singleton block on a minimum-weight coordinate, and the greatest is
$s(b)-m(b)$, attained by its complement.

Given the block weights $b$, set
\[
        p=m(b),\qquad q=s(b)-m(b).
\]
Under this non-negativity hypothesis, $A_c(z)+p$ is the least value attained by $M_{c,b}(z,r)$ over mixed
blocks $r$, while $A_c(z)+p+q=A_c(z)+s(b)$ is the value at the all-one
block.
We encode this by
\[
        \Phi(c,b)=\bigl(c,m(b),s(b)-m(b)\bigr).
\]

\begin{lemma}[Relevance is preserved in the positive case]
\label{lem:relevance}
If $f$ is positive and depends on all its variables, then $g_\ell$ depends on
all $n+\ell-2$ of its variables.
\end{lemma}

In the applications below we work with positive functions satisfying these
relevance hypotheses.

\begin{proof}
Let $z_j$ be relevant to $f$, and choose a relevance witness consisting of two
points that differ only in $z_j$. Choose $r\in\B^\ell$ so that $\pi(r)$ is either
$(u,v)$, if $(u,v)\in\{(0,0),(1,0),(1,1)\}$, or $(1,0)$, if $(u,v)=(0,1)$.
In the last case this replacement does not change the value of $f$, because
$u$ and $v$ are symmetric. Keep this same block $r$ in both inputs. The two
resulting inputs to $g_\ell$ differ only in $z_j$, and their values are the
corresponding two values of $f$. Hence $z_j$ is relevant to $g_\ell$.

For a block variable $r_i$, use relevance of $u$. Positivity gives $z$ and
$v\in\B$ such that $f(z,0,v)=0$ and $f(z,1,v)=1$. If $v=0$, toggle $r_i$ from
the all-zero block to $e_i$; if $v=1$, toggle it from $\mathbf1-e_i$ to the
all-one block. In either case the value of $g_\ell$ changes, so every $r_i$ is
relevant.
\end{proof}

For the transformed function $g_\ell$, define
\[
        \Rep(g_\ell)
        =
        \{(c,b): M_{c,b}(z,r)\text{ has the signs prescribed by }g_\ell
        \text{ on all }(z,r)\in\B^{n+\ell-2}\}.
\]
This set may be empty. Non-emptiness is exactly the assertion that $g_\ell$ is
threshold. We will also use the corresponding weak sign cone
\[
        \Rep^{\mathrm{wk}}(g_\ell)
        =
        \{(c,b): M_{c,b}(z,r)\ge0\text{ when }g_\ell(z,r)=1,
        \text{ and }M_{c,b}(z,r)\le0\text{ when }g_\ell(z,r)=0\}.
\]
When $g_\ell$ is threshold, this weak sign cone is the closed chamber
$\overline{C_{g_\ell}}$.

\begin{proposition}[Collapse]
\label{prop:collapse}
Assume that $f=f(z,u,v)$ is positive, depends on all its variables, and has
$u$ and $v$ symmetric. For every
$\ell\ge3$,
\[
        \Rep(g_\ell)=\Phi^{-1}(C_f).
\]
The corresponding weak sign cones satisfy
\[
        \Rep^{\mathrm{wk}}(g_\ell)=\Phi^{-1}(\overline{C_f}).
\]
Hence, whenever $g_\ell$ is threshold,
\[
        \overline{C_{g_\ell}}=\Phi^{-1}(\overline{C_f}).
\]
\end{proposition}

\begin{proof}
Suppose first that 
\[
        \Phi(c,b)=(c,m(b),s(b)-m(b))\in C_f.
\]
Lemma~\ref{lem:positive-weights} gives $m(b)>0$. Fix
$z\in\B^{n-2}$ and write
\[
        A=A_c(z),\qquad m=m(b),\qquad s=s(b).
\]
Since $(c,m,s-m)\in C_f$, the four evaluations of $L_{c,m,s-m}$ give
\[
        A,\quad A+m,\quad A+s-m,\quad A+s
\]
at $(u,v)=(0,0),(1,0),(0,1),(1,1)$, respectively. Hence these four quantities
have the signs prescribed by
\[
        f(z,0,0),\quad f(z,1,0),\quad f(z,0,1),\quad f(z,1,1).
\]
Since $u$ and $v$ are symmetric, $f(z,0,1)=f(z,1,0)$. Thus both $A+m$ and
$A+s-m$ have the sign prescribed by the mixed block.

Let $I$ be a nonempty proper subset of $[\ell]$, and put
\[
        t_I=\sum_{i\in I} b_i .
\]
Since all $b_i>0$,
\[
        m\le t_I\le s-m.
\]
The affine value at the block whose support is $I$ is $A+t_I$. We have already
shown that both endpoints $A+m$ and $A+s-m$ have sign $f(z,1,0)$. If this sign
is $1$, then both endpoints are positive, and hence every value between them,
including $A+t_I$, is positive. If this sign is $0$, then both endpoints are
negative, and hence every value between them, including $A+t_I$, is negative.
Thus $A+t_I$ has sign $f(z,1,0)$ for every nonempty proper $I\subsetneq[\ell]$.

It remains only to assemble the three block cases. For the all-zero block,
\[
        M_{c,b}(z,0^\ell)=A,
\]
which has the sign prescribed by $f(z,0,0)=g_\ell(z,0^\ell)$. For the
all-one block,
\[
        M_{c,b}(z,1^\ell)=A+s,
\]
which has the sign prescribed by $f(z,1,1)=g_\ell(z,1^\ell)$. Finally, if
$r$ is mixed, with support $I=\{i:r_i=1\}$, then $I$ is a nonempty proper
subset of $[\ell]$, and
\[
        M_{c,b}(z,r)=A+\sum_{i\in I}b_i.
\]
By the preceding paragraph this has the sign prescribed by
$f(z,1,0)=g_\ell(z,r)$. These three cases cover all $r\in\B^\ell$, so
$M_{c,b}$ has the signs prescribed by $g_\ell$ at every point $(z,r)$. Since
this holds for every $z$, we have $(c,b)\in\Rep(g_\ell)$.

Conversely, suppose $(c,b)\in\Rep(g_\ell)$. Let $j$ be such that
$b_j=m(b)$. Fix $z\in\B^{n-2}$, and write
\[
        A=A_c(z),\qquad m=m(b),\qquad s=s(b).
\]
We show that the affine form
\[
        A_c(z)+mu+(s-m)v
\]
has the signs prescribed by $f$.

At the all-zero block $0^\ell$, the value of $M_{c,b}$ is $A$, and this block
collapses to $(0,0)$. Hence $A$ has the sign prescribed by $f(z,0,0)$. At the
all-one block $1^\ell$, the value is $A+s$, and this block collapses to
$(1,1)$. Hence $A+s$ has the sign prescribed by $f(z,1,1)$.

At the singleton block $e_j$, the value is $A+m$. This block is mixed, so it
collapses to $(1,0)$. Hence $A+m$ has the sign prescribed by $f(z,1,0)$. At
the complementary block $\mathbf 1-e_j$, the value is $A+s-m$. This block is
also mixed, so it collapses to $(1,0)$, and hence $A+s-m$ has the sign
prescribed by $f(z,1,0)$. Since $u$ and $v$ are symmetric,
$f(z,1,0)=f(z,0,1)$. Thus $A+s-m$ has the sign prescribed by $f(z,0,1)$,
which is the sign needed for the $(u,v)=(0,1)$ case of the original affine
form.

Thus the four values
\[
        A,\quad A+m,\quad A+s-m,\quad A+s
\]
have the signs prescribed by
\[
        f(z,0,0),\quad f(z,1,0),\quad f(z,0,1),\quad f(z,1,1),
\]
respectively. Since this holds for every $z$, the coefficient vector
\[
        (c,m,s-m)=\Phi(c,b)
\]
strictly represents $f$. Hence $\Phi(c,b)\in C_f$.

Using the closure part of Lemma~\ref{lem:positive-weights}, the same reasoning
with non-strict inequalities gives
\[
        \Rep^{\mathrm{wk}}(g_\ell)=\Phi^{-1}(\overline{C_f}).
\]
\end{proof}

Since $u$ and $v$ are symmetric, $(c,p,q)\in C_f$ if and only if
$(c,q,p)\in C_f$. Thus the thresholdness condition may be stated using
$\max\{p,q\}$ and $\min\{p,q\}$, without choosing an order for the two
symmetric variables.

\begin{theorem}[Exact thresholdness criterion]
\label{thm:s-threshold}
Assume that $f=f(z,u,v)$ is positive, depends on all its variables, and has $u$ and $v$ symmetric. For every
$\ell\ge3$, the transformed function $g_\ell$ is threshold if and only if $f$
has a representing coefficient vector $(c,p,q)\in C_f$, where $p$ and $q$ are
the coefficients of the two original symmetric variables, such that
\[
        p,q>0
        \quad\text{and}\quad
        \max\{p,q\}>(\ell-1)\min\{p,q\}.
\]
Equivalently,
\[
        C_f\cap
        \left\{(c,p,q): p,q>0,\;
        \max\{p,q\}>(\ell-1)\min\{p,q\}\right\}
        \ne\varnothing.
\]
\end{theorem}

\begin{proof}
Suppose first that $g_\ell$ is threshold. Then there is
$(c,b)\in\Rep(g_\ell)$. By Proposition~\ref{prop:collapse},
\[
        \Phi(c,b)=(c,m(b),s(b)-m(b))\in C_f.
\]
Put $p=m(b)$ and $q=s(b)-m(b)$. Then $p>0$, and $q$ is the sum of the remaining
$\ell-1$ block weights. Since $f$ is positive, so is $g_\ell$. By
Lemma~\ref{lem:relevance}, $g_\ell$ depends on all its variables. Applying
Lemma~\ref{lem:positive-weights} to the strict representation $(c,b)$ of
$g_\ell$, all block weights $b_i$ are positive. In particular, each remaining
weight is at least $p$. Hence
\[
        q\ge(\ell-1)p.
\]
If this inequality is strict, we have obtained the required point of $C_f$. If
$q=(\ell-1)p$, then, since $C_f$ is open, we may increase $q$ slightly while
keeping the other coefficients fixed and remain in $C_f$. Hence $C_f$ contains
a point satisfying
\[
        q>(\ell-1)p.
\]
Using the symmetry of $u$ and $v$, this gives the condition stated with
$\max\{p,q\}$ and $\min\{p,q\}$.

Conversely, suppose $(c,p,q)\in C_f$, with $p,q>0$, and, after interchanging
$u$ and $v$ if necessary, suppose
\[
        q>(\ell-1)p.
\]
Choose block weights by setting
\[
        b_1=\cdots=b_{\ell-1}=p,
        \qquad
        b_\ell=q-(\ell-2)p.
\]
Then $b_\ell>p$, so all $b_i>0$, $m(b)=p$, and
\[
        s(b)-m(b)=b_2+\cdots+b_\ell=q.
\]
Hence
\[
        \Phi(c,b)=(c,p,q)\in C_f.
\]
By Proposition~\ref{prop:collapse}, $(c,b)\in\Rep(g_\ell)$, and therefore
$g_\ell$ is threshold.
\end{proof}

For the original operation of Lozin \emph{et al.}, the final symmetric block has
size $\ell=3$. Thus Theorem~\ref{thm:s-threshold} says that the
$(u,v,y)$-s-extension is threshold if and only if $C_f$ contains a coefficient
vector $(c,p,q)$, with $p,q>0$, such that
\[
        \max\{p,q\}>2\min\{p,q\}.
\]
Equivalently, after interchanging the symmetric variables $u$ and $v$ if
necessary, there is a representation with
\[
        q>2p.
\]
The ratio depends on the chosen symmetric pair, so define
\[
        \rho_{u,v}(f)
        =\sup_{(c,p,q)\in C_f}
          \frac{\max\{p,q\}}{\min\{p,q\}}
        \in[1,\infty].
\]
Theorem~\ref{thm:s-threshold} says precisely that
\[
        g_\ell\text{ is threshold}
        \iff \rho_{u,v}(f)>\ell-1.
\]
Thus the allowable block size is governed by the extent to which the chamber
permits the two symmetric weights to become unbalanced. In particular,
thresholdness is monotone in the block size: if $g_\ell$ is threshold, then
$g_m$ is threshold for every $3\le m\le \ell$.
\begin{example}
\label{ex:threshold}
Let $f=x_1\vee(x_2\wedge x_3)$, with $x_2,x_3$ the symmetric pair. For any
$0<\eta<\varepsilon<1$, the affine form
\[
        2x_1+\varepsilon x_2+x_3-(1+\eta)
\]
strictly represents $f$. Indeed, $x_1=1$ already gives a positive value, while,
on the layer $x_1=0$, the points $00,10,01,11$ in the coordinates $(x_2,x_3)$
give values
\[
        -(1+\eta),\quad
        \varepsilon-(1+\eta),\quad
        1-(1+\eta),\quad
        1+\varepsilon-(1+\eta),
\]
of which only the last is positive. Hence $f$ has strict representations with
coefficient ratio $1/\varepsilon$ on the symmetric pair $x_2,x_3$, and this
ratio is unbounded as $\varepsilon\downarrow0$. Thus
$\rho_{x_2,x_3}(f)=\infty$, and every $g_\ell$ is threshold. Explicitly,
\[
        g_\ell=x_1\vee(r_1\wedge\cdots\wedge r_\ell)
\]
is represented by weight $\ell$ on $x_1$, weight $1$ on each $r_i$, and
threshold $\ell$. Theorem~\ref{thm:s-preservation} below gives
$g_\ell\in\Tk{\ell+1}$ for every $\ell\ge3$.
\end{example}
Example~\ref{ex:threshold} shows that the imbalance condition can hold
strongly. The next example shows the opposite behaviour and recovers the
non-threshold $s$-extension exhibited by Lozin \emph{et al.}
\begin{example}[A concrete non-threshold $s$-extension]
\label{ex:lozin-nonthreshold}
Lozin \emph{et al.}\ give the positive threshold function
\[
        f=x_1x_2\vee (x_1\vee x_2)x_3x_4x_5,
\]
where $x_1$ and $x_2$ are symmetric variables, and show that its
$(x_1,x_2,x_6)$-s-extension is not threshold~\cite[Example~2]{Lozin2022}.
Theorem~\ref{thm:s-threshold} explains this failure directly.
Lozin \emph{et al.}\ prove non-thresholdness by exhibiting a direct
2-summability obstruction. Namely, for
\[
        h=f^{(x_1,x_2,x_6)}
        =x_1x_2x_6\vee (x_1\vee x_2\vee x_6)x_3x_4x_5,
\]
the points
\[
        (1,1,0,0,0,1),\qquad (0,1,1,1,1,0)
\]
are true points, while
\[
        (0,1,1,0,1,1),\qquad (1,1,0,1,0,0)
\]
are false points, and
\[
        (1,1,0,0,0,1)+(0,1,1,1,1,0)
        =
        (0,1,1,0,1,1)+(1,1,0,1,0,0).
\]
Thus $h$ is 2-summable, and hence not threshold. The argument below gives the
corresponding chamber explanation: the original chamber $C_f$ does not contain
a representation in which the two coefficients on the symmetric variables
$x_1,x_2$ are sufficiently unbalanced.
Let
\[
        L=d+p x_1+q x_2+a x_3+b x_4+c x_5
\]
be a strict positive representation of $f$, so that $L>0$ on true points and
$L<0$ on false points. By symmetry of $x_1,x_2$, assume $p\ge q$, and put
\[
        r=\min\{a,b,c\}.
\]
The point $(0,1,1,1,1)$ is true for $f$, since $x_2=1$ and
$x_3=x_4=x_5=1$. Hence
\[
        d+q+a+b+c>0.
\]
Now take the point $(1,0,1,1,1)$ and change to $0$ one of the coordinates
$x_3,x_4,x_5$ whose coefficient is $r$. The resulting point has $x_1=1$,
$x_2=0$, but not all of $x_3,x_4,x_5$ equal to $1$, so it is false for $f$.
Hence
\[
        d+p+a+b+c-r<0.
\]
Combining these two inequalities gives
\[
        d+p+a+b+c-r<0<d+q+a+b+c,
\]
and therefore
\[
        p-q<r.
\]
Next, $(1,1,0,0,0)$ is true for $f$, because $x_1x_2=1$. Thus
\[
        d+p+q>0.
\]
Comparing this with the earlier inequality $d+p+a+b+c-r<0$, we obtain
\[
        d+p+a+b+c-r<0<d+p+q,
\]
so
\[
        a+b+c-r<q.
\]
But $a+b+c-r$ is the sum of the two coefficients among $a,b,c$ other than the
minimum one, and each of those two coefficients is at least $r$. Hence
\[
        a+b+c-r\ge 2r.
\]
Together with the inequality $a+b+c-r<q$ obtained above, this gives
\[
        q>2r,
\]
and so
\[
        r<\frac q2.
\]
Combining $p-q<r$ with $r<\frac q2$, we get
\[
        p<q+r<q+\frac q2=\frac32 q.
\]
Since $p\ge q$, every strict positive representation of $f$ satisfies
\[
        \frac{\max\{p,q\}}{\min\{p,q\}}
        =
        \frac pq
        <
        \frac32 .
\]
For the original symmetric-variables extension of Lozin \emph{et al.}, the
final symmetric block has size $\ell=3$. Theorem~\ref{thm:s-threshold}
therefore requires a representation with
\[
        \frac{\max\{p,q\}}{\min\{p,q\}}>2.
\]
No such representation exists for the function above. Hence its
$(x_1,x_2,x_6)$-s-extension is not threshold, in agreement with the example of
Lozin \emph{et al.} Thus Lozin \emph{et al.}'s 2-summability certificate proves
non-thresholdness directly in the extended cube. The point of the present
criterion is different: it detects the same failure already in the representing
chamber of the original function. In this example the chamber never reaches the
coefficient-imbalance region required by the $s$-extension.
\end{example}

\subsubsection*{Preservation of minimum specification}

We now combine the thresholdness criterion above with the preservation theorem
of Lozin \emph{et al.}\ for the original single-variable symmetric extension.
We pause to explain why the iteration is legitimate. The $\ell$-fold block
extension is obtained by applying the single-variable $s$-extension $\ell-2$
times, and thresholdness of the final block extension implies thresholdness at
all intermediate stages by Theorem~\ref{thm:s-threshold}.

Lozin \emph{et al.}\ proved the single-variable case, which we iterate
\cite[Theorem~7]{Lozin2022}: if $g$ is a positive function in $\Tk{m}$ with
symmetric variables $x_i$ and $x_j$, and the $(x_i,x_j,y)$-$s$-extension $g'$ is
threshold, then $g'\in\Tk{m+1}$. The $(x_i,x_j,y)$-$s$-extension is the function
\[
        g'=x_ix_jy\,g|_{x_i=1,x_j=1}
        \vee(x_i\vee x_j\vee y)\,g|_{x_i=1,x_j=0}
        \vee g|_{x_i=0,x_j=0}
\]
of \cite[Definition~3]{Lozin2022}, in which $x_i$, $x_j$, and $y$ are pairwise
symmetric and each restriction is taken in the remaining variables.

\begin{theorem}[Preservation of minimum specification when thresholdness holds]
\label{thm:s-preservation}
Let $f=f(z,u,v)\in\Tn$ be positive, with $u$ and $v$ symmetric variables, and
let $\ell\ge3$. Let $g_\ell$ be the $\ell$-fold symmetric-block extension of
$f$. If $g_\ell$ is threshold, then
\[
        \sigma_{n+\ell-2}(g_\ell)=n+\ell-1,
        \qquad
        \text{so } g_\ell\in\Tk{n+\ell-2}.
\]
\end{theorem}

\begin{proof}
Write $g_2=f$, and for $m\ge3$ let $g_m$ denote the $m$-fold symmetric-block
extension of the original pair $u,v$. The operation passing from $g_m$ to
$g_{m+1}$ is exactly the single-variable symmetric extension of Lozin
\emph{et al.}, applied to any two variables in the current symmetric block.
After one such extension, the old two variables and the new variable are
pairwise symmetric in the extended function, so the same operation may be
iterated on the growing block \cite[Definition~3 and p.~8]{Lozin2022}.

Since $g_\ell$ is threshold, Theorem~\ref{thm:s-threshold} gives
\[
        \rho_{u,v}(f)>\ell-1.
\]
Hence, for every $3\le m\le\ell$, we have $\rho_{u,v}(f)>m-1$, and another
application of Theorem~\ref{thm:s-threshold} shows that each intermediate
function $g_m$ is threshold.

We now apply this preservation theorem of Lozin \emph{et al.} inductively. The base case is
$g_2=f\in\Tn$. Suppose that $g_m\in\Tk{n+m-2}$ for some $2\le m<\ell$. Each $g_m$ is
positive by construction, and the variables in the growing block remain
pairwise symmetric. The function $g_{m+1}$ is threshold by the preceding
paragraph, and it is the single-variable symmetric extension of $g_m$ on a
symmetric pair in its block. Lozin \emph{et al.}'s theorem therefore gives
\[
        g_{m+1}\in\Tk{n+m-1}.
\]
By induction, $g_\ell\in\Tk{n+\ell-2}$, which is precisely the displayed equality
for the specification number.
\end{proof}

Combining Theorems~\ref{thm:s-threshold} and~\ref{thm:s-preservation}, for a
positive $f\in\Tn$ and the chosen symmetric pair $u,v$,
\[
        g_\ell\in\Tk{n+\ell-2}
        \iff g_\ell\text{ is threshold}
        \iff \rho_{u,v}(f)>\ell-1.
\]
At $\ell=3$ this recovers \cite[Theorem~7]{Lozin2022} and adds the exact
existence criterion.

\section{The fourth operation}\label{sec:fourth}

Lozin \emph{et al.}\ also proposed a further construction applied to the family
\(f_{n,k}\). We close the main development by analysing this construction.
We prove a result specific to that family: the proposed image has
minimum specification number exactly in the boundary case \(k=n-1\), while for
\(k\le n-2\) the image is not threshold.

\subsection{Statement and tools}

We use two standard tools. Elgot's theorem \cite{Elgot1961} says that a
Boolean function is threshold if and only if it is \emph{asummable}: there are no
\(r\ge2\) ones \(y_1,\dots,y_r\) and \(r\) zeros \(z_1,\dots,z_r\), not
necessarily distinct, with \(\sum y_i=\sum z_i\). The case \(r=2\) is
\(2\)-summability, so one four-point coincidence proves non-thresholdness. By
the facet correspondence of Theorem~\ref{thm:facet},
\(\sigma_m(h)=|\Ess(h)|\) for a threshold function on \(m\) variables. For a
positive threshold function depending on all its variables, the essential points
are extremal, that is, minimal ones or maximal zeros \cite[Theorem~2.5]{ABST1995}. We use
summability below the boundary and the signature count at it.


Lozin \emph{et al.}\ \cite{Lozin2022} study \(\Tk{n}\) through the family
\[
        f_{n,k}=x_1x_2\vee x_1x_3\vee\cdots\vee x_1x_k\vee x_2x_3\cdots x_n,
        \qquad 3\le k\le n-1,
\]
each member of which lies in \(\Tk{n}\). At the end of their Section~6 they
introduce a further operation, their equation~(10), taking \(f_{n,k}\) to the
\((n+2)\)-variable function
\[
        f'=f_{n,k}|_{x_1=1,x_n=0}\,(x_1\vee x_{n+2})
        \;\vee\;
        f_{n,k}|_{x_1=0,x_n=1}\,(x_n\vee x_{n+1})
        \;\vee\;
        x_1(x_{n+1}\vee x_{n+2}).
\]
The operation forms the image from the two restrictions of \(f_{n,k}\) in \(x_1\) and
\(x_n\), each multiplied by a clause in the fresh variables
\(x_{n+1},x_{n+2}\), together with a final clause carrying \(x_1\). Lozin
\emph{et al.}\ verified by hand that the image of \(f_{4,3}\) lies in
\(\Tk{6}\), and conjectured that this operation applied to any \(f_{n,k}\)
again has minimum specification number, leaving the general case open. We settle
it: the conjecture holds at the top of the range and fails everywhere below it.

\begin{theorem}[Resolution of the fourth operation]
\label{thm:main}
Let \(3\le k\le n-1\) and let \(f'\) be the image of \(f_{n,k}\) under this
operation. Then \(f'\) is a threshold function with minimum specification number,
that is, \(f'\in\Tk{n+2}\), if and only if \(k=n-1\). For every \(k\le n-2\) the
function \(f'\) is \(2\)-summable, hence not threshold.
\end{theorem}

We divide the proof on the value of \(k\). For \(k\le n-2\) a single coincidence of
four points shows \(f'\) is \(2\)-summable (Section~\ref{ssec:disproof}). For the
boundary value \(k=n-1\), an explicit weight system, an enumeration of the
extremal points, and a count of the essential ones place \(f'\) in \(\Tk{n+2}\)
(Section~\ref{ssec:boundary}).

Throughout, the restrictions are
\[
        f_{n,k}|_{x_1=1,x_n=0}=x_2\vee\cdots\vee x_k,
        \qquad
        f_{n,k}|_{x_1=0,x_n=1}=x_2x_3\cdots x_{n-1},
\]
so the operation reads
\[
        f'=(x_2\vee\cdots\vee x_k)(x_1\vee x_{n+2})
        \vee(x_2\cdots x_{n-1})(x_n\vee x_{n+1})
        \vee x_1(x_{n+1}\vee x_{n+2}).
\]

\subsection{The conjecture fails for \texorpdfstring{$k\le n-2$}{k<=n-2}}
\label{ssec:disproof}

When \(k\le n-2\) the middle conjunction \(x_2\cdots x_{n-1}\) in the displayed
formula for \(f'\) involves a variable, any \(x_i\) with \(k+1\le i\le n-1\), that
the first clause \(x_2\vee\cdots\vee x_k\) does not reach. The certificate
exploits this gap.

\begin{proposition}
\label{prop:disproof}
If \(3\le k\le n-2\), then \(f'\) is \(2\)-summable, and hence not a threshold
function.
\end{proposition}

\begin{proof}
Since \(k\le n-2\), the index set \(\{k+1,\dots,n-1\}\) is non-empty. Define
points of \(\B^{n+2}\), in coordinates \((x_1,\dots,x_n,x_{n+1},x_{n+2})\), by
\[
\begin{aligned}
        T_1&:\ x_1=0,\ x_i=1\ (2\le i\le n),\ x_{n+1}=x_{n+2}=0,\\
        T_2&:\ x_1=1,\ x_i=0\ (2\le i\le n),\ x_{n+1}=1,\ x_{n+2}=0,\\
        F_1&:\ x_1=0,\ x_i=1\ (2\le i\le k),\ x_i=0\ (k+1\le i\le n),\ x_{n+1}=1,\ x_{n+2}=0,\\
        F_2&:\ x_1=1,\ x_i=0\ (2\le i\le k),\ x_i=1\ (k+1\le i\le n),\ x_{n+1}=0,\ x_{n+2}=0.
\end{aligned}
\]
From the displayed formula for \(f'\): \(f'(T_1)=1\) because \(x_2\cdots
x_{n-1}=1\) and \(x_n=1\) fire the middle term; \(f'(T_2)=1\) because \(x_1=1\)
and \(x_{n+1}=1\) fire the last term. Next \(f'(F_1)=0\): the first and last
terms vanish as \(x_1=x_{n+2}=0\), and the middle term vanishes because some
\(x_i\) with \(k+1\le i\le n-1\) equals \(0\), so \(x_2\cdots x_{n-1}=0\). And
\(f'(F_2)=0\): the last term vanishes as \(x_{n+1}=x_{n+2}=0\), the first
vanishes as \(x_2=\cdots=x_k=0\), and the middle vanishes as \(x_2\cdots
x_{n-1}=0\). Finally, coordinatewise,
\[
        T_1+T_2=F_1+F_2=(1,\underbrace{1,\dots,1}_{2\le i\le n},1,0),
\]
since in each coordinate both sides equal \(1\), except \(x_{n+2}\) where both
equal \(0\). Thus \(f'\) is \(2\)-summable, and by Elgot's theorem not threshold.
\end{proof}

The two ones \(T_1,T_2\) and the two zeros \(F_1,F_2\) share a coordinate sum, so
no affine form can hold the ones above a threshold and the zeros below it. The
certificate is uniform: only the pattern of the four points changes with \(n\)
and \(k\), never their number.

\begin{example}
It is useful to keep the smallest failing case in mind: \(n=5\), \(k=3\), where \(f'\) has seven variables.
In coordinates \((x_1,\dots,x_5,x_6,x_7)\) the four points are
\[
\begin{aligned}
        T_1&=(0,1,1,1,1,0,0), & T_2&=(1,0,0,0,0,1,0),\\
        F_1&=(0,1,1,0,0,1,0), & F_2&=(1,0,0,1,1,0,0),
\end{aligned}
\]
with \(f'(T_1)=f'(T_2)=1\), \(f'(F_1)=f'(F_2)=0\), and
\(T_1+T_2=F_1+F_2=(1,1,1,1,1,1,0)\).
\end{example}

\subsection{The boundary family \texorpdfstring{$k=n-1$}{k=n-1}}
\label{ssec:boundary}

At \(k=n-1\) the first clause \(x_2\vee\cdots\vee x_{n-1}\) and the middle clause
\(x_2\cdots x_{n-1}\) run over the same middle variables, the gap of the previous
case closes, and \(f'\) becomes threshold. We make this explicit, realising
\(f'\) by weights, listing its extremal points, and discarding the inessential
ones to count the rest.

Fix \(k=n-1\) and \(n\ge4\). Write \(h=n-2\) for the number of \emph{middle}
variables \(x_2,\dots,x_{n-1}\), and \(s=x_2+\cdots+x_{n-1}\in\{0,\dots,h\}\) for
their sum. With \(\vee\) over an empty set read as \(0\) and \(\wedge\) as \(1\),
the displayed formula for \(f'\) becomes
\[
        f'=[\,s\ge1\,](x_1\vee x_{n+2})
        \vee[\,s=h\,](x_n\vee x_{n+1})
        \vee x_1(x_{n+1}\vee x_{n+2}).
\]
We write a point as \((x_1;\,m;\,x_n,x_{n+1},x_{n+2})\), where \(m\in\B^h\) is
the middle block; \(\mathbf 0^h,\mathbf 1^h,e_j\) denote the zero block, the
all-ones block, and the \(j\)-th unit block.

\subsubsection*{Extremal points}

Before giving the general count, it may help to see the smallest boundary case.
Take \(n=4\), so that \(h=n-2=2\), and write points as
\[
        (x_1;\,x_2,x_3;\,x_4,x_5,x_6).
\]
Here the middle variables are \(x_2,x_3\), the variables whose sum appears in
the two bracketed conditions. The transformed function is
\[
        [x_2+x_3\ge 1](x_1\vee x_6)
        \vee
        [x_2+x_3=2](x_4\vee x_5)
        \vee
        x_1(x_5\vee x_6).
\]
The minimal true points are obtained by turning on one of the three displayed
terms in a minimal way. The first term is turned on minimally either by taking
one middle variable together with \(x_6\), or by taking \(x_1\) together with one
middle variable:
\[
\begin{aligned}
T_1&=(0;\,1,0;\,0,0,1),&
T_2&=(0;\,0,1;\,0,0,1),\\
T_3&=(1;\,1,0;\,0,0,0),&
T_4&=(1;\,0,1;\,0,0,0).
\end{aligned}
\]
The second term is turned on minimally by taking both middle variables and
exactly one of \(x_4,x_5\):
\[
\begin{aligned}
T_5&=(0;\,1,1;\,0,1,0),&
T_6&=(0;\,1,1;\,1,0,0).
\end{aligned}
\]
The third term is turned on minimally by taking \(x_1\) and exactly one of
\(x_5,x_6\):
\[
\begin{aligned}
T_7&=(1;\,0,0;\,0,0,1),&
T_8&=(1;\,0,0;\,0,1,0).
\end{aligned}
\]
Thus there are \(8=2n\) minimal true points.

The maximal false points are obtained in the opposite way: we keep all three
terms false, while setting every variable not forced to be zero as high as
possible. If both middle variables are zero, then the first two terms are off.
With \(x_1=0\), the third term is also off, and we may set \(x_4,x_5,x_6\) to
one:
\[
        F_1=(0;\,0,0;\,1,1,1).
\]
If exactly one middle variable is one, then the first term forces
\(x_1=x_6=0\), while \(x_4\) and \(x_5\) may still be one:
\[
\begin{aligned}
F_2&=(0;\,0,1;\,1,1,0),&
F_3&=(0;\,1,0;\,1,1,0).
\end{aligned}
\]
If both middle variables are one and \(x_1=0\), then the first two terms force
\(x_4=x_5=x_6=0\):
\[
        F_4=(0;\,1,1;\,0,0,0).
\]
Finally, if \(x_1=1\), then to keep the first and third terms false we need
\(x_2=x_3=x_5=x_6=0\), while \(x_4\) may still be one:
\[
        F_5=(1;\,0,0;\,1,0,0).
\]
Thus there are \(5=n+1\) maximal false points.

We obtain the general case from this example by replacing the two middle
coordinates by a middle block \(m\in\mathbb B^h\), where \(h=n-2\). The singleton
middle blocks \((1,0)\) and \((0,1)\) become the \(h\) unit vectors \(e_j\); the
all-one middle block becomes \(\mathbf 1^h\); the all-zero middle block becomes
\(\mathbf 0^h\); and the middle blocks with exactly one zero become
\(\mathbf 1^h-e_j\). This clause-by-clause construction is exactly what the
general statement below records for a middle block of arbitrary size.

\begin{lemma}
\label{lem:extremal}
The minimal ones of \(f'\) are the \(2n\) points
\[
\begin{aligned}
        T^{(1)}_j&=(0;e_j;0,0,1), &&1\le j\le h,\\
        T^{(2)}_j&=(1;e_j;0,0,0), &&1\le j\le h,\\[1mm]
        T^{(3)}&=(0;\mathbf 1^h;0,1,0),\\
        T^{(4)}&=(0;\mathbf 1^h;1,0,0),\\[1mm]
        T^{(5)}&=(1;\mathbf 0^h;0,0,1),\\
        T^{(6)}&=(1;\mathbf 0^h;0,1,0).
\end{aligned}
\]
and the maximal zeros are the \(n+1\) points
\[
\begin{aligned}
        F^{(1)}&=(0;\mathbf 0^h;1,1,1),
        & F^{(2)}_j&=(0;\mathbf 1^h-e_j;1,1,0)
                 \quad(1\le j\le h),\\
        F^{(3)}&=(0;\mathbf 1^h;0,0,0),
        & F^{(4)}&=(1;\mathbf 0^h;1,0,0).
\end{aligned}
\]
In particular \(f'\) has \(3n+1\) extremal points and depends on all \(n+2\)
variables.
\end{lemma}

\begin{proof}[Proof of Lemma~\ref{lem:extremal}]
Write
\[
        a=x_1,
        \qquad b=x_{n+2},
        \qquad c=x_n,
        \qquad e=x_{n+1}.
\]
Then the boundary formula is
\[
        f'(x)=1 \iff
        \big(s\ge1\wedge(a\vee b)\big)
        \vee
        \big(s=h\wedge(c\vee e)\big)
        \vee
        \big(a\wedge(e\vee b)\big).
\]
The preceding example already indicates the construction. We verify that the
listed points are exactly the extremal points.

First consider true points. The first term
\[
        [s\ge 1](x_1\vee x_{n+2})
\]
is minimally made true either by taking one middle variable and \(x_{n+2}\),
giving \(T^{(1)}_j\), or by taking \(x_1\) and one middle variable, giving
\(T^{(2)}_j\). The second term
\[
        [s=h](x_n\vee x_{n+1})
\]
is minimally made true by taking the full middle block and exactly one of
\(x_n,x_{n+1}\), giving \(T^{(3)}\) and \(T^{(4)}\). The third term
\[
        x_1(x_{n+1}\vee x_{n+2})
\]
is minimally made true by taking \(x_1\) and exactly one of \(x_{n+1},x_{n+2}\),
giving \(T^{(5)}\) and \(T^{(6)}\). Conversely, every true point makes at least one term true and dominates
the minimal point listed for that term, so each minimal true point appears in
the list; as the listed points are pairwise incomparable, they are precisely
the minimal true points.

Now consider false points, taking \(x_1=0\) and \(x_1=1\) in turn. Suppose
first that \(x_1=0\), so the third term vanishes; the first term is then
\([s\ge1]\,x_{n+2}\) and the second is \([s=h](x_n\vee x_{n+1})\). If \(s=0\),
both terms vanish for every choice of \(x_n,x_{n+1},x_{n+2}\), so maximality
sets all three to \(1\), giving \(F^{(1)}\). If \(1\le s\le h-1\), the first
term forces \(x_{n+2}=0\), while the second vanishes because \(s<h\), leaving
\(x_n\) and \(x_{n+1}\) free; the point is then false for every middle block of
weight \(s\). When \(s\le h-2\) such a point is not maximal, since raising any
zero middle coordinate increases \(s\) to a value still at most \(h-1\) and
leaves the point false. Maximality therefore forces \(s=h-1\), with
\(x_n=x_{n+1}=1\) and \(x_{n+2}=0\), giving the points \(F^{(2)}_j\), indexed by
the single zero middle coordinate \(j\). If \(s=h\), the first term forces
\(x_{n+2}=0\) and the second forces \(x_n=x_{n+1}=0\), giving \(F^{(3)}\).
Suppose instead that \(x_1=1\). The third term is then \(x_{n+1}\vee x_{n+2}\),
so falsity forces \(x_{n+1}=x_{n+2}=0\), and the first term, now \([s\ge1]\),
forces \(s=0\), after which the second vanishes; maximality allows \(x_n=1\),
giving \(F^{(4)}\). These are precisely the maximal false points.

Thus there are
\[
        h+h+4=2h+4=2n
\]
minimal true points and
\[
        1+h+1+1=h+3=n+1
\]
maximal false points, as claimed.

Every variable occurs with value \(1\) in at least one listed minimal true point.
Changing that coordinate to \(0\) gives a zero by minimality, so each variable is
relevant and \(f'\) depends on all \(n+2\) variables.
\end{proof}

\subsubsection*{Threshold realisation}

In the same smallest case \(n=4\), a corresponding threshold realisation is
given by
\[
        (w_1,\ldots,w_6)=(7,4,4,1,2,5),
        \qquad t=9.
\]
Thus the threshold inequality is
\[
        7x_1+4x_2+4x_3+x_4+2x_5+5x_6\ge 9.
\]
For instance, the minimal true point
\[
        (0;\,1,0;\,0,0,1)
\]
has weight \(4+5=9\), while the maximal false point
\[
        (0;\,0,0;\,1,1,1)
\]
has weight \(1+2+5=8\).

\begin{proposition}
\label{prop:threshold}
With weights
\[
        w_1=4n-9,\quad w_i=4\ (2\le i\le n-1),\quad
        w_n=1,\quad w_{n+1}=2,\quad w_{n+2}=4n-11,
\]
and threshold \(t=4n-7\), one has
\(f'(x)=1\iff\langle w,x\rangle\ge t\). In particular, \(f'\) is threshold.
\end{proposition}

\begin{proof}[Proof of Proposition~\ref{prop:threshold}]
All weights are positive for \(n\ge4\). Writing \(W=\langle w,x\rangle\) and using
\(h=n-2\), the values on the minimal ones of Lemma~\ref{lem:extremal} are
\[
\begin{array}{lll}
        W(T^{(1)}_j)=4n-7, & W(T^{(2)}_j)=4n-5, & W(T^{(3)})=4n-6,\\
        W(T^{(4)})=4n-7, & W(T^{(5)})=8n-20, & W(T^{(6)})=4n-7,
\end{array}
\]
all at least \(t\); for \(T^{(5)}\), \(8n-20\ge4n-7\) when \(n\ge4\). On the
maximal zeros,
\[
        W(F^{(1)})=4n-8,\quad W(F^{(2)}_j)=4n-9,\quad
        W(F^{(3)})=4n-8,\quad W(F^{(4)})=4n-8,
\]
all at most \(t-1\). Since the function is monotone increasing and the weights are positive, every true point dominates one of the listed minimal ones and every false point is dominated by one of the listed maximal zeros, so \(W\ge t\) at every true point and \(W\le t-1\) at every false point.
\end{proof}

\subsubsection*{Inessential points and the count}

\begin{lemma}
\label{lem:inessential}
The \(2n-2\) extremal points \(T^{(2)}_1,\dots,T^{(2)}_h\), \(T^{(3)}\),
\(T^{(5)}\), \(F^{(2)}_1,\dots,F^{(2)}_h\) are inessential.
\end{lemma}

\begin{proof}[Proof of Lemma~\ref{lem:inessential}]
For a point \(x\in\mathbb B^{n+2}\), let \(f'^{\,x}\) denote the Boolean
function obtained from \(f'\) by changing the value at \(x\) and leaving all
other values unchanged. The point \(x\) is inessential iff \(f'^{\,x}\) is not
threshold. By Elgot's theorem, it is enough to show that \(f'^{\,x}\) is
\(2\)-summable.

Put
\[
        P=(0;\mathbf 0^h;0,1,1),
        \qquad
        Q_j=(0;e_j;0,1,0).
\]
Both are zeros of \(f'\). Indeed, for \(P\) we have \(s=0\) and \(a=0\), so the
first and third terms vanish, while \(s<h\), so the second term vanishes. For
\(Q_j\) we have \(a=b=0\) and \(s=1<h\) since \(n\ge4\), so again all three
terms vanish.

The following identities hold coordinatewise:
\[
\begin{aligned}
        T^{(3)}+F^{(4)}&=T^{(4)}+T^{(6)}, & T^{(5)}+Q_j&=T^{(1)}_j+T^{(6)},\\
        T^{(2)}_j+P&=T^{(1)}_j+T^{(6)}, & F^{(2)}_j+T^{(1)}_j&=F^{(1)}+F^{(3)}.
\end{aligned}
\]
For the first three identities, the point being proved inessential is one of
the true points on the left-hand side. After flipping that point, the two points
on the left-hand side are zeros of \(f'^{\,x}\), while the two points on the
right-hand side remain ones of \(f'^{\,x}\). Thus \(f'^{\,x}\) is
\(2\)-summable, and hence is not threshold.

More explicitly, if \(x=T^{(3)}\), then
\[
        T^{(3)}+F^{(4)}=T^{(4)}+T^{(6)}
\]
expresses a sum of two zeros of \(f'^{\,x}\) as a sum of two ones. If
\(x=T^{(5)}\), then, for any \(j\),
\[
        T^{(5)}+Q_j=T^{(1)}_j+T^{(6)}
\]
does the same. If \(x=T^{(2)}_j\), then
\[
        T^{(2)}_j+P=T^{(1)}_j+T^{(6)}
\]
does the same. Therefore in each case \(f'^{\,x}\) is not threshold.

Finally, if \(x=F^{(2)}_j\), then \(F^{(2)}_j\) is flipped from false to true, and
\[
        F^{(2)}_j+T^{(1)}_j=F^{(1)}+F^{(3)}
\]
expresses a sum of two ones of \(f'^{\,x}\) as a sum of two zeros. Hence
\(f'^{\,x}\) is again \(2\)-summable, and therefore is not threshold.
\end{proof}

\begin{theorem}[The boundary case $k=n-1$]
\label{thm:boundary}
For \(k=n-1\) and \(n\ge4\), the image \(f'\) of \(f_{n,n-1}\) under the operation
above satisfies \(\sigma_{n+2}(f')=n+3\), so \(f'\in\Tk{n+2}\). Its essential
points are exactly
\[
        \Ess(f')=\{T^{(1)}_1,\dots,T^{(1)}_h,\ T^{(4)},\ T^{(6)},\ F^{(1)},\ F^{(3)},\ F^{(4)}\}.
\]
\end{theorem}

\begin{proof}
By Proposition~\ref{prop:threshold}, \(f'\) is threshold. The general lower
bound for threshold functions on \(\B^{n+2}\) gives
\[
        \sigma_{n+2}(f')\ge n+3.
\]
On the other hand, by Theorem~\ref{thm:facet}, \(\sigma_{n+2}(f')=|\Ess(f')|\), and for
this positive function \(\Ess(f')\) is contained in the set of \(3n+1\) extremal
points of Lemma~\ref{lem:extremal}. Lemma~\ref{lem:inessential} removes
\(2n-2\) of them, leaving at most \((3n+1)-(2n-2)=n+3\) candidates, the \(n+3\)
points displayed. Hence \(|\Ess(f')|\le n+3\), so \(\sigma_{n+2}(f')=n+3\), every
displayed point is essential, and \(f'\in\Tk{n+2}\).
\end{proof}

Theorem~\ref{thm:main} follows from Proposition~\ref{prop:disproof} and
Theorem~\ref{thm:boundary}.

\begin{remark}
The failure for \(k\le n-2\) has a one-line interpretation in the chamber language of
Section~\ref{sec:chamber}: with \(\widetilde x=(1,x)\), the certificate of
Proposition~\ref{prop:disproof} gives \(\widetilde T_1+\widetilde
T_2=\widetilde F_1+\widetilde F_2\) with \(T_1,T_2\) true and \(F_1,F_2\) false,
so no affine form separates them and the candidate chamber is empty. We record
the result combinatorially, as the summability certificate is self-contained and
uniform.
\end{remark}

\subsection{Product operations and the special role of the half-line}
\label{ssec:product-operations}

We close this section by recording how the four operations of \cite{Lozin2022}
sit together in the chamber picture. The half-line product
of Lemma~\ref{lem:halfline-simplicial} occupies a distinguished position, which
the following observation makes precise: an operation that acts on the chamber by
taking the product with a fixed factor adds that factor's facet count to the
specification number, and lands in the minimum class exactly when both the input
and the factor are simplicial.

\begin{proposition}[Specification numbers under products with a fixed factor]
\label{prop:product-operations}
Let $f\in\mathcal H_n$ with $n\ge2$, let $k\ge1$, and let $Q\subseteq\R^k$ be a
full-dimensional pointed polyhedral cone. Suppose an operation $\Psi$ produces a
threshold function $\Psi(f)\in\mathcal H_{n+k}$ whose closed chamber
$\overline{C_{\Psi(f)}}$ is linearly equivalent to the product
$\overline{C_f}\times Q$. Then
\[
        \sigma_{n+k}\bigl(\Psi(f)\bigr)=\sigma_n(f)+\#\{\text{facets of }Q\},
\]
and
\[
        \Psi(f)\in\Tk{n+k}
        \quad\Longleftrightarrow\quad
        f\in\Tn\ \text{and}\ Q\ \text{is simplicial}.
\]
Whenever $Q$ is simplicial it is linearly equivalent to $\R_{\ge0}^k$, so
$\overline{C_{\Psi(f)}}$ is linearly equivalent to the product
$\overline{C_f}\times\R_{\ge0}^k$ of $\overline{C_f}$ with the orthant; for $k=1$
this is the half-line product of Lemma~\ref{lem:halfline-simplicial}.
\end{proposition}

\begin{proof}
The nonempty faces of a product of polyhedra are the products of a face of each factor, and dimension is additive across the product. The codimension-one faces are therefore exactly the sets $F\times Q$, for $F$ a facet of $\overline{C_f}$, together with the sets $\overline{C_f}\times G$, for $G$ a facet of $Q$, because a product $F\times G$ of facets from both factors has codimension two and so is not a facet. Hence $\overline{C_f}\times Q$ has
$\#\{\text{facets of }\overline{C_f}\}+\#\{\text{facets of }Q\}$ facets. By
Theorem~\ref{thm:facet} the facet counts of $\overline{C_f}$ and of
$\overline{C_{\Psi(f)}}$ equal $\sigma_n(f)$ and $\sigma_{n+k}(\Psi(f))$
respectively, and linear equivalence preserves facet counts, so
$\sigma_{n+k}(\Psi(f))=\sigma_n(f)+\#\{\text{facets of }Q\}$. A full-dimensional
pointed cone is simplicial exactly when its facet count equals its dimension.
Since the facet count and the dimension of a product are both additive, and each
factor has at least as many facets as its dimension, $\overline{C_f}\times Q$ is
simplicial if and only if both $\overline{C_f}$ and $Q$ are. By
Theorem~\ref{thm:Tn-simplicial}, applied in dimensions $n$ and $n+k$,
$\overline{C_f}$ is simplicial exactly when $f\in\Tn$ and $\overline{C_{\Psi(f)}}$
exactly when $\Psi(f)\in\Tk{n+k}$. Combining these gives the stated equivalence.
Finally, a simplicial cone in $\R^k$ is spanned by $k$ linearly independent rays,
so it is the image of the orthant $\R_{\ge0}^k$ under the invertible linear map
carrying the standard basis to those rays. Hence a simplicial $Q$ is linearly
equivalent to $\R_{\ge0}^k$, and $\overline{C_f}\times Q$ to the product
$\overline{C_f}\times\R_{\ge0}^k$ of $\overline{C_f}$ with the orthant. For $k=1$
the full-dimensional pointed cones in $\R^1$ are $\R_{\ge0}$ and $\R_{\le0}$,
linearly equivalent to one another.
\end{proof}

The first two operations are the case $Q=\R_{\ge0}$: adding a variable and
extension on a variable coincide on chambers with the half-line product
(Remark~\ref{rem:two-extensions-coincide}), so they carry every $f\in\Tn$ into
$\Tk{n+1}$. The other two constructions are not products with a fixed factor. The
symmetric-variables extension gives a simplicial chamber only when the input $f$
admits a representation with $\max\{p,q\}>(\ell-1)\min\{p,q\}$
(Theorem~\ref{thm:s-threshold}); and the fourth operation, which we analyse only
on the specific family $f_{n,k}$ to which it is applied, yields a threshold image
only in the boundary case $k=n-1$, where that image has minimum specification
number (Theorem~\ref{thm:boundary}, Proposition~\ref{prop:disproof}).

\section{Concluding remarks and further directions}\label{sec:open}

The chamber-facet identification (Theorem~\ref{thm:facet}) leaves several
directions open.

\paragraph{Classifying simplicial chambers.} The minimum-specification
problem becomes: classify the chambers of $\Arr_n$ whose closures are
simplicial. In the language of labelled examples, this asks which labelled
sets of $n+1$ cube points can occur as the essential points of an
all-variable-relevant threshold function. The examples of Lozin \emph{et al.}\ show that this
class is strictly larger than the linear read-once class, so a classification
cannot follow only from the known forbidden-restriction characterisations of
lro functions.

\paragraph{Operations preserving minimum specification.}
Proposition~\ref{prop:product-operations} shows that an operation acting on the
chamber as a product with a fixed factor $Q$ sends $\Tn$ into $\Tk{n+k}$ exactly
when $Q$ is simplicial, in which case the new chamber is the product of the old
one with an orthant. It would be of interest to delimit a class of operations on
threshold functions wide enough to include the constructions of \cite{Lozin2022}
yet structured enough to admit a classification, and to ask whether the half-line
product is, up to congruence, the only single-variable operation in that class
that preserves minimum specification for all inputs. Some restriction is
essential: with none, an operation could ignore its input and return a fixed
member of $\Tk{n+1}$, so the interest lies in the choice of class.

\paragraph{Algorithmic complexity.} From the truth table,
the signature, $\Ess(f)$, and $\sigma_n(f)$ can be computed in time
polynomial in $2^n$ by running the feasibility test of
Proposition~\ref{prop:lp-test} for each cube point. The open question
concerns \emph{compact} representations: given a threshold function by a
compact integer weight representation, can its essential points (and hence, by
Theorem~\ref{thm:facet}, its minimum specifying set) be computed in
output-polynomial time? The chamber-facet identification gives the linear-programming certificate of
Proposition~\ref{prop:lp-test} for testing individual candidate facets, but it does not
by itself give an efficient way to avoid testing many non-essential extremal
candidates. The complexity of computing the signature from a compact
representation, and the related complexity of recognising membership in
$\Tn$, both remain natural questions.

\paragraph{The constant in the linear growth.} Theorem~\ref{thm:theta-n}
settles the order of $\overline\sigma_n$ as linear and, through
Corollary~\ref{cor:exact}, pins down the resonance count $r(\Res)$ to within a
constant factor. What remains is the constant itself:
Problem~\ref{prob:constant} asks whether $\overline\sigma_n/(n+1)$ converges,
and if so to what value in $[1,2]$.

\paragraph{Chow parameters and hybrid specification.} A further direction is
to compare labelled-example specification with Chow-parameter specification
\cite{Chow1961,Winder1971}. The zonotope of Section~\ref{sec:zonotope} makes
the comparison concrete: by Theorem~\ref{thm:zonotope}, labelled-example
specification is the local edge structure of $Z_n$ at the vertex $\chi(f)$,
while Chow-parameter specification is the position of that vertex, so hybrid
specification asks how the two together locate the vertex. Building on the
reconstruction results of Goldberg \cite{Goldberg2006}, O'Donnell and Servedio
\cite{ODonnellServedio2011}, and Servedio \cite{Servedio2007}, we may ask for
hybrid specification bounds: if some Chow parameters are known exactly, or if
all are known approximately, how many labelled examples are still needed to
determine the threshold function? A parallel question in the polynomial
setting is suggested by the robust Chow-parameter determination of
Diakonikolas and Kane \cite{DiakonikolasKane2019}.

\paragraph{Polynomial threshold functions.} The chamber-facet mechanism
extends formally to degree-$d$ polynomial threshold functions
(Proposition~\ref{prop:ptf-average}), but for $d>1$ the corresponding
restricted arrangements are no longer the ordinary resonance arrangement.
Useful region-count estimates for these higher-degree analogues would be a
natural extension of the present approach.

\subsection*{Declaration of generative AI and AI-assisted technologies in the manuscript preparation process}

During the preparation of this work, the author used OpenAI ChatGPT and
Anthropic Claude for assistance with literature review  and
TeX/TikZ code for mathematical diagrams. The author
verified all content and takes full responsibility for the published article.

\end{document}